\newif\iffullversion
\newif\ifdraft
\newif\ifanonymous

\fullversiontrue

\iffullversion
\documentclass[11pt]{article}
\else
\documentclass[a4paper,USenglish,autoref,thm-restate\ifanonymous, anonymous\fi]{lipics-v2021}
\let\hl\undefined 
\fi

\iffullversion\usepackage[ascii]{inputenc}\fi
\usepackage[T1]{fontenc}
\usepackage{paralist}
\usepackage{amsmath,amssymb}
\usepackage{declmath}
\iffullversion
\usepackage{makeidx}\makeindex
\usepackage[a4paper,margin=1in]{geometry}
\usepackage[backend=bibtex,maxbibnames=100,style=alphabetic]{biblatex}
\fi
\usepackage{stmaryrd}
\usepackage{version}
\usepackage{tocloft}
\usepackage{mathpartir}
\usepackage[toc]{multitoc}
\usepackage[framemethod=tikz]{mdframed}
\iffullversion
\usepackage[pdfborderstyle={/S/U/W 0.3}]{hyperref}
\fi

\iffullversion
\newcommand\fullshort[2]{#1}
\includeversion{fullversion}
\excludeversion{shortversion}
\else
\newcommand\fullshort[2]{#2}
\excludeversion{fullversion}
\includeversion{shortversion}
\fi

\newcommand\fullonly[1]{\fullshort{#1}{}}
\newcommand\shortonly[1]{\fullshort{}{#1}}

\makeatletter

\newcommand\sectionToToc{%
  \let\OLDsection\section
  \def\section*##1{%
    \OLDsection*{##1}%
    \addcontentsline{toc}{section}{##1}%
  }%
}

\newcommand\rhlstacktemplate[3]{\arraycolsep=0pt\,\, \begin{array}{l}
                                                       #2\\\hfill \overset{\scriptscriptstyle\mathsf{#1}}\sim #3
                                                     \end{array}\,}
\newcommand\rhltotstack{\rhlstacktemplate{tot}}
                                                   
\newcommand\showafter[4]{%
  \ifnum\year>#1\relax
    #4
  \else
    \ifnum\year=#1\relax
      \ifnum\month>#2\relax
        #4%
      \else
        \ifnum\month=#2\relax
          \ifnum\day>#3\relax
            #4%
          \else
            \ch@removeduplicatespace
          \fi
        \else
          \ch@removeduplicatespace
        \fi
      \fi
    \else
      \ch@removeduplicatespace
    \fi
  \fi}

\newcommand\pagelabel[1]{\phantomsection\label{#1}}

\newlistof{rule}{rule}{{List of rules}}

\newcommand\rulestrut{}
\newcommand\rulespace{}
\newcommand\RULElabel[3][]{%
  \refstepcounter{rule}%
  \label{rule:#3-#2}%
  \index{#3@\textsc{#3} (rule)}%
  \addcontentsline{rule}{rule}{\textsc{#3}%
    \ifx\@nil#1\@nil\else\space -- #1\fi}}
\newcommand\implicitRULE[2]{\textsc{#2}\RULElabel{#1}{#2}}
\newcommand\RULE[5][]{%
  \inferrule[{\rulestrut{#3\ifx\norulelemmalink\undefined
    \fi}%
  \RULElabel[{#1}]{#2}{#3}%
  }]{#4}{#5}\rulespace%
}
\newcommand\RULEX[4][]{%
  \inferrule[\rulestrut{#2}%
      \ifx\@nil#1\@nil\else\space -- #1\fi]{#3}{#4}\rulespace%
}

\newcommand\refrule[2]{\hyperref[rule:#2-#1]{\hbox{\textsc{#2}}} rule}
\newcommand\ruleref[2]{\hyperref[rule:#2-#1]{rule \hbox{\textsc{#2}}}}
\newcommand\rulerefx[2]{\hyperref[rule:#2-#1]{\hbox{\textsc{#2}}}}
\newcommand\Ruleref[2]{\hyperref[rule:#2-#1]{Rule \hbox{\textsc{#2}}}}

\mdfdefinestyle{proofstyle}{%
  backgroundcolor=black!4!white,
  innertopmargin=3pt,
  innerbottommargin=3pt,
  innerleftmargin=3pt,
  innerrightmargin=3pt,
  linewidth=0pt
}

\mdfdefinestyle{claimproofstyle}{%
  backgroundcolor=black!8!white,
  innertopmargin=0pt,
  innerbottommargin=3pt,
  innerleftmargin=3pt,
  innerrightmargin=3pt,
  leftmargin=10pt,
  rightmargin=10pt,
}

\fullonly{
  \newenvironment{proof}[1][]{%
    \begin{mdframed}[style=proofstyle]%
      \begin{trivlist}%
      \item \textit{Proof\ifx\proof#1\proof\else\ #1\fi. }%
      }{
        \qed\global\let\qed\qedmacro
      \end{trivlist}
    \end{mdframed}}%

}

\newcommand\qedboxclaim{\ensuremath\diamond}
\newcommand\qedbox{\m@th$\Box$}

\global\let\qed\qedmacro
\newcommand\skipqed{\global\let\qed\empty}
\newcommand\mathQED{\tag*{\qedbox}\skipqed}
\newcommand\txtrel[1]{\stackrel{\hskip-1in\text{\tiny#1}\hskip-1in}}
\newcommand\eqrefrel[1]{\txtrel{\eqref{#1}}}
\newcommand\starrel{\txtrel{$(*)$}}
\newcommand\starstarrel{\txtrel{$(**)$}}
\newcommand\tristarrel{\txtrel{$\scriptscriptstyle (\!*\!*\!*\!)$}}

\shortonly{\appto\endproof{\global\let\qed\qedmacro}}

\def\parenthesis@resizecommand@open{}
\def\parenthesis@resizecommand@close{}
\def\parenthesis@#1#2#3#4#5{%
  \global\let\parenthesis@resizecommand@open\empty
  \global\let\parenthesis@resizecommand@close\empty
  #1#3%
  {#5}%
  #2#4%
}
\DeclareRobustCommand\parenthesis[3]{%
  \expandafter\expandafter\expandafter
  \parenthesis@\expandafter\expandafter\expandafter{\expandafter\parenthesis@resizecommand@open\expandafter}\expandafter{\parenthesis@resizecommand@close}{#1}{#2}{#3}}

\def\double@parenthesis@#1#2#3#4#5#6#7#8#9{%
  \global\let\parenthesis@resizecommand@open\empty
  \global\let\parenthesis@resizecommand@close\empty
  #1#3{#4}#2#5%
  {#6}%
  #1#7{#8}#2#9%
}

\DeclareRobustCommand\doubleParenthesis[7]{%
  \expandafter\expandafter\expandafter
  \double@parenthesis@\expandafter\expandafter\expandafter{\expandafter\parenthesis@resizecommand@open\expandafter}\expandafter{\parenthesis@resizecommand@close}{#1}{#2}{#3}{#4}{#5}{#6}{#7}}

\DeclareRobustCommand\resize@next@paren[2]{%
  \gdef\parenthesis@resizecommand@open{#1}%
  \gdef\parenthesis@resizecommand@close{#2}%
}

\newcommand\paren{\parenthesis()}
\newcommand\braces{\parenthesis\{\}}
\newcommand\bracket{\parenthesis[]}

\newcommand\pb{\resize@next@paren\bigl\bigr}
\newcommand\pbb{\resize@next@paren\biggl\biggr}
\newcommand\plr{\resize@next@paren\left\right}
\newcommand\pB{\resize@next@paren\Bigl\Bigr}
\newcommand\pBb{\resize@next@paren\biggl\biggr}

\newcommand\bit{\{0,1\}}

\newcommand\fsq{\tfrac1{\sqrt2}}
  
\declare{macro=\setN, code=\mathbb N, nopage=true,
  description={Natural numbers $1,2,3,\dots$}}
\declare{macro=\setZ, code=\mathbb Z, nopage=true,
  description={Integers}}
\declare{macro=\setRpos, code=\mathbb R_{\geq0}, nopage=true,
  description={Nonnegative reals}}
\declare{macro=\setC, code=\mathbb C, nopage=true,
  description={Complex numbers}}

\declare{macro=\xx, code=\mathbf x, wrap=group,
  placeholder={\xx,\yy,\zz},
  description={(Program) variable}}
\newcommand\yy{{\mathbf y}}
\newcommand\zz{{\mathbf z}}

\declare{macro=\XX, code=\mathbf X,
  placeholder={\XX,\YY,\ZZ},
  description={List/set of program variables}}
\newcommand\YY{\mathbf Y}
\newcommand\ZZ{\mathbf Z}

\declare{macro=\ket, argspec=[1], code=\parenthesis\lvert\rangle{#1},
  placeholder=\ket x,
  description={Basis state $x$}}
\declare{macro=\bra, argspec=[1], code=\parenthesis\langle\rvert{#1},
  placeholder=\bra x,
  description={Adjoint of $\ket x$, i.e., $\adj{\ket x}$}}

\declare{macro=\elltwo, argspec=[1], code={\ell^2(#1)},
  placeholder=\elltwo X,
  description={Hilbert space with basis indexed by $X$}}
\declare{macro=\elltwov, argspec=[1], code={\ell^2\parenthesis[]{#1}},
  placeholder=\elltwov\XX,
  description={Pure quantum memories on $\XX$}}
\declare{macro=\elltwovb, argspec=[1], code={\ell^2\parenthesis[]{#1}^{\bot}},
  placeholder=\elltwovb\XX,
  description={Pure quantum $\bot$-memories on $\XX$}}

\declare{macro=\mm, code=\mathsf m, wrap=group,
  description={An assignment}}

\declare{macro=\tensor, code=\otimes,
  placeholder={A\otimes B},
  description={Tensor product of vectors/operators/spaces $A$ and $B$}}

\declare{macro=\opon, argspec=[2], code={#1\,\mathbf{on}\,#2},
  placeholder=\opon U\XX,
  description={Operator $U$ applied to variables $\XX$}}
\newcommand\oponp[2]{\paren{\opon{#1}{#2}}}

\declare{macro=\PA, code=\mathsf A,
  placeholder={\PA,\PB,\PC},
  description={(Quantum) expectations}}

\newcommand\PB{\mathsf B}
\newcommand\PC{\mathsf C}
\newcommand\PD{\mathsf D}

\declare{macro=\norm, argspec=[1], code={\parenthesis\lVert\rVert{#1}}, nopage=true,
  placeholder=\norm \psi, description={Norm of vector $\psi$}}

\declare{macro=\SPAN, code={\operatorname{span}}, nopage=true,
  placeholder=\SPAN A,
  description={Span, smallest subspace containing $A$}}

\declare{macro=\QUANTEQ, code={\equiv_\mathfrak q},
  placeholder={\XX\QUANTEQ\XX'},
  description={Predicate (subspace): $\XX$ and $\XX'$ are in the same state}}

\declare{macro=\SWAPXX, argspec=[2], code=\mathsf{SWAP}_{#1#2},
  placeholder=\SWAPXX{\XX_1}{\XX_2},
  description={Unitary that swaps $\XX_1$ and $\XX_2$}}

\declare{macro=\SWAPlr, code=\mathsf{SWAP},
  description={Unitary that swaps left and right memory}}

\declare{macro=\SWAPlrbot, code=\mathsf{SWAP}_{\bot},
  description={Unitary that swaps left and right $\bot$-memory}}

\declare{macro=\fv, code=\mathsf{fv},
  placeholder=\fv(a),
  description={Free variables of expectation/program $a$}}

\declare{macro=\bc, code=\pmb{\mathfrak c}, nopage=true,
  placeholder={\bc,\bd},
  description={A program}}
\newcommand\bd{\pmb{\mathfrak d}}

\declare{macro=\apply, argspec=[2], code={{\applykw}\ #1\ \underline{\mathbf{to}}\ #2},
  placeholder=\apply U\XX,
  description={Program: Apply $U$ to variables $\XX$}}
\newcommand\applykw{\underline{\mathbf{apply}}}

\declare{macro=\init, argspec=[2], code={#1\leftarrow#2},
  placeholder=\init\XX\psi,
  description={Program: Initialize $\XX$ with $\psi$}}

\declare{macro=\ifte, argspec=[4], code={\underline{\mathbf{if}}\ {#1}[#2]\ \underline{\mathbf{then}}\ #3\ \underline{\mathbf{else}}\ #4},
  placeholder=\ifte M\xx\bc\bd,
  description={Program: If (conditional)}}
\newcommand\ifskip[2]{\underline{\mathbf{if}}\ {#1}[#2]}

\newcommand\whilekw{\underline{\mathbf{while}}}
\declare{macro=\while, argspec=[3], code={\whilekw\ {#1}[#2]\ \underline{\mathbf{do}}\ #3},
  placeholder=\while M\xx\bc,
  description={Program: While (loop)}}

\declare{macro=;,
  placeholder=\bc;\bd,
  description={Program: execute $\bc$ then $\bd$}}

\declare{macro=\SKIP, code=\underline{\mathbf{skip}},
  description={Program: does nothing}}

\declare{macro=\abort, code=\underline{\mathbf{abort}},
  description={Nonterminating program}}

\declare{macro=\denot, argspec=[1], code={\parenthesis\llbracket\rrbracket{#1}},
  placeholder={\denot \bc},
  description={Denotation of a program $\bc$}}

\declare{macro=\denotbot, argspec=[1], code=\denot{#1}^{\scriptscriptstyle\bot},
  placeholder={\denotbot \bc},
  description={Denotation of a program $\bc$}}

\declare{macro=\proj, argspec=[1], code={\mathsf{proj}\parenthesis(){#1}},
  placeholder=\proj \psi,
  description={Projector onto $\psi$, i.e., $\psi\adj \psi$}}

\declare{macro=\adj, argspec=[1], code=#1^*,
  placeholder={\adj a},
  description={Adjoint of operator/vector $a$}}
\newcommand\adjp[1]{\adj{\paren{#1}}}

\declare{macro=\tr, code=\operatorname{tr}, nopage=true,
  placeholder=\tr M,
  description={Trace of matrix/operator $M$}}

\declare{macro=\id, code=\mathsf{id},
  description=Identity}

\declare{macro=\restrict, argspec=[1], code=\mathord{\downarrow_{#1}},
  placeholder=\restrict i(\rho),
  description=Mixed state restricted to measurement outcome $i$}

\declare{macro=\restrictbot, code=\restrict\bot,
  placeholder={\restrictbot(\rho), \restrictnbot(\rho)},
  description=$\bot$-memory restricted to measurement outcome $\ket\bot$ / the orthogonal of $\ket\bot$}
\newcommand\restrictnbot{\restrict{\not\bot}}

\declare{macro=\restrictast, argspec=[1], code=\mathord{\downarrow_{#1}^\ast},
  placeholder={\restrictast t(\PA),\ \restrictast{t,u}(\PA)},
  description=Dual of the restriction $\restrict i$ on states}

\declare{macro=\partr, argspec=[1], code=\tr_{#1},
  placeholder=\partr{\XX} \rho,
  description={Partial trace (removing variables $\XX$)}}

\declare{macro=\XXall, code=\XX^{\mathbf{all}},
  description={Set of program variables that can be used in the execution of a program}}

\declare{macro=\hl, argspec=[3],
  code={\doubleParenthesis\{{#1}\}{\,#2\,}\{{#3}\}},
  placeholder=\hl \PA\bd\PB,
  description={Non-relational Hoare judgment}}

\newcommand\simtemplate[1]{\overset{\scriptscriptstyle\mathsf{#1}}\sim}
\newcommand\rhltemplate[5]{\doubleParenthesis\{{#2}\}{\,#3\,\simtemplate{#1}\,#4\,}\{{#5}\}}

\declare{macro=\rhlgen, argspec=[4],
  code=\rhltemplate{gen}{#1}{#2}{#3}{#4},
  placeholder=\rhlgen \PA\bc\bd\PB,
  description={Relational Hoare judgment, generic}}

\declare{macro=\rhlsemi, argspec=[4],
  code=\rhltemplate{semi}{#1}{#2}{#3}{#4},
  placeholder=\rhlsemi \PA\bc\bd\PB,
  description={Relational Hoare judgment, semi-partial}}

\declare{macro=\rhltot, argspec=[4],
  code=\rhltemplate{tot}{#1}{#2}{#3}{#4},
  placeholder=\rhltot \PA\bc\bd\PB,
  description={Relational Hoare judgment, total}}

\declare{macro=\rhlpar, argspec=[4],
  code=\rhltemplate{par}{#1}{#2}{#3}{#4},
  placeholder=\rhlpar \PA\bc\bd\PB,
  description={Relational Hoare judgment, partial}}

\declare{macro=\rhltheirs, argspec=[4],
  code=\rhltemplate{\!\!bhyyz\!\!}{#1}{#2}{#3}{#4},
  placeholder=\rhltheirs \PA\bc\bd\PB,
  description={Relational Hoare judgment from \cite{theirs}}}

\newcommand\rhlany{\rhltemplate{any}}
\newcommand\rhlinformal{\rhltemplate{}}

\newcommand\couprel{=}

\declare{macro=\varconcat, argspec=[1], code=#1,
  placeholder=\varconcat{\XX\YY},
  description={Concatenation/disjoint union of variable lists/sets $\XX,\YY$}}

\declare{macro=\incr, code=\mathsf{INCR},
  description={Unitary mapping $\ket i$ to $\ket{i+1}$}}

\declare{macro=\true, code=\mathsf{true}, nopage=true,
  description={Boolean truth value true}}
\declare{macro=\false, code=\mathsf{false}, nopage=true,
  description={Boolean truth value false}}

\declare{macro=\Mtrue, argspec=[1], code={#1}_\mathsf{true},
  placeholder=\Mtrue M,
  description={Operator corresponding to outcome $\true$ of measurement $M$}}
\declare{macro=\Mfalse, argspec=[1], code={#1}_\mathsf{false},
  placeholder=\Mfalse M,
  description={Operator corresponding to outcome $\false$ of measurement $M$}}
\newcommand\Mtruefalse[2]{{#1}_{#2}}

\declare{macro=\equal, code=\mathsf{EQUAL},
  description={Quantum equality}}

\declare{macro=\negXX, code=\neg,
  placeholder=\neg\XX,
  description={Short for $\XXall_1\XXall_2\setminus\XX$}}

\declare{macro=\calE, code=\mathcal E, nopage=true,
  placeholder={\calE},
  description={Usually stands for a superoperator}}

\declare{macro=\calH, code=\mathcal H, nopage=true,
  description={A Hilbert space}}

\declare{macro=\Ebot, code=\mathcal E_\bot,
  placeholder=\Ebot(\rho),
  description={Measurement whether a mixed $\bot$-memory is $\bot$}}

\declare{macro=\Pbot, code=P_\bot,
  description={Projector on $\ket\bot$}}
\declare{macro=\Pnbot, code=P_{\not\bot},
  description={Projector on complement of $\ket\bot$}}

\declare{macro=\liftb, argspec=[1], code=#1\oplus0_\bot,
  placeholder={\liftb A},
  description={Operator $A$ extended to $\elltwovb\XX$}}
\declare{macro=\liftbb, argspec=[1], code=#1\oplus0_{\bot\bot},
  placeholder={\liftbb A},
  description={Operator $A$ extended to $\elltwovb{\XX_1}\otimes\elltwovb{\XX_2}$}}
\newcommand\liftbbp[1]{\paren{\liftbb{#1}}}

\declare{macro=\sandwich, code=\bullet,
  placeholder={A\sandwich B},
  description={Short for $AB\adj A$}}

\shortonly{
\bibliographystyle{plainurl}
}

\fullonly{

\newtheorem{lemma}{Lemma}

\newtheorem{definition}{Definition}

}

\fullonly{
  \bibliography{expectation-qrhl}
}

\newcommand\delaytextsv[2]{#2}

\newcommand{\mathcolorbox}[2]{\mathchoice%
  {\colorbox{#1}{$\displaystyle#2$}}%
  {\colorbox{#1}{$\textstyle#2$}}%
  {\colorbox{#1}{$\scriptstyle#2$}}%
  {\colorbox{#1}{$\scriptscriptstyle#2$}}}%

\def\symbolindexmarkhighlight#1{%
  \setlength{\fboxsep}{2pt}%
  {\TextOrMath{\colorbox{gray!20}{#1}}{\mathcolorbox{gray!20}{#1}}}}

\AtBeginDocument{
  \ifx\PreviewMacro\undefined\else
  \PreviewMacro*\footnote
  \PreviewMacro*\fullonly
  \PreviewMacro*\shortonly
  \PreviewMacro*\fullshort
  \PreviewMacro*\symbolindexmarkhighlight
  \PreviewMacro*[{}]\symbolindexmark
  \PreviewMacro*\delaytextsv
  \fi
}

\title{Quantum Relational Hoare Logic with Expectations%
  \thanks{Presented at \href{http://easyconferences.eu/icalp2021/}{ICALP 2021} \cite{expectations-qrhl-icalp}}}

\ifanonymous\else
\fullshort{
  \author{
    Yangjia Li\\\small University of Tartu\\\small Institute of Software, CAS, Beijing
    \and
    Dominique Unruh\\\small University of Tartu
  }
}{
  \author{Yangjia Li}{University of Tartu, Estonia\and SKLCS, Institute of Software, CAS, Beijing, China}{yangjia.li@ut.ee}{}{}
  \author{Dominique Unruh}{University of Tartu, Estonia}{unruh@ut.ee}{https://orcid.org/0000-0001-8965-1931}{}
  \authorrunning{Y.\,Li and D.\,Unruh}
  \Copyright{Y.\,Li and D.\,Unruh}
}
\fi

\shortonly{
  \ccsdesc[500]{Theory of computation~Quantum computation theory}
  \ccsdesc[500]{Theory of computation~Hoare logic}
  \ccsdesc[500]{Theory of computation~Program verification} 
  
  \keywords{Quantum cryptography, Hoare logics, formal verification}
  \relatedversiondetails{Full version}{https://arxiv.org/abs/1903.08357}
}

\begin{document}

\maketitle


\begin{abstract}
  We present a variant of the quantum relational Hoare logic from (Unruh, POPL 2019)
    that allows us to use ``expectations'' in pre- and postconditions.
    That is, when reasoning about pairs of programs, our logic allows us to
    quantitatively reason about how much certain pre-/postconditions
    are satisfied that refer to the relationship between the programs inputs/outputs.
\end{abstract}

{\renewcommand*{\multicolumntoc}{2}
\setlength{\columnseprule}{.4pt}
\tableofcontents}

\section{Introduction}

Relational Hoare logics (RHL) are logics that allow us to reason about the relationship between two programs.
Roughly speaking, they can express facts like ``if the variable $\xx$ in program $\bc$ is equal to $\xx$ in program $\bd$,
then after executing $\bc$ and $\bd$, respectively, the content of variable $\yy$ in program $\bc$ is greater than that of $\yy$ in $\bd$.''
RHL was introduced in the deterministic case by \cite{benton04relational}, and generalized to probabilistic programs by \cite{certicrypt} (pRHL)
and to quantum programs by \cite{qrhl} (qRHL).
RHLs have proven especially useful in the context of verification of cryptographic schemes.
For example, the CertiCrypt tool \cite{certicrypt,certicrypt-web} and its successor EasyCrypt \cite{easycrypt,fosad-easycrypt}
use pRHL to create formally verified cryptographic proofs.
And \cite{github-source} implements a tool for verifying quantum cryptographic proofs based on qRHL.

On the other hand, ``normal'' (i.e., not relational)
quantum Hoare logics have been developed in the quantum setting,
starting with the predicate transformers from
\cite{dhondt06weakest},
see \cite{feng07proof,ying12floyd,chadha06reasoning,kakutani09logic}.
Out of these, \cite{dhondt06weakest,feng07proof,ying12floyd}
use ``expectations'' instead of ``predicates'' for the pre- and postconditions
of the Hoare judgments.
To understand the difference, consider the case of classical probabilistic programs.
Here, a predicate is (logically equivalent to) a set of program states
(and a program state is a function from variables to values).
In contrast, an expectation is a function from program states to real numbers,
basically assigning a value to each program state.
Probabilistic Hoare logic with expectations, implicit in \cite{kozen83probabilistic},
uses expectations as the pre- and postconditions of a Hoare judgment.
Then, roughly speaking, the preexpectation tells us what the expected value of the postexpectation is
after running the program.
This can be used to express much more fine-grained properties of probabilistic programs,
giving quantitative guarantees about their probabilistic behavior,
instead of just qualitative (a certain final state can or cannot occur).
As \cite{dhondt06weakest} showed, the same approach can be used for quantum programs.
Here, an expectation is modeled by a self-adjoint operator $\PA$ on the space of all program states.
The ``value'' of a given program state $\rho$ is then computed as the trace $\tr\PA\rho$.
While at the first glance not as obvious as the meaning of classical expectations,
this formalism has nice mathematical properties and is also equivalent to taking
the expectation value of the outcome of a real-valued measurement.
By using this approach, \cite{dhondt06weakest,feng07proof,ying12floyd}
can express more fine-grained judgments about quantum programs,
by not just expressing which final states are possible, but also with what probabilities.

Yet, qRHL \cite{qrhl} did not follow this approach (only mentioning it as possible future work).
As a consequence, qRHL does not enable as fine-grained reasoning
about probabilities as the non-relational quantum Hoare logics.
On the other hand, the non-relational quantum Hoare logics
do not allow us to reason about the relationship between programs.

In this work, we combine the best of two worlds.
We present a variant of qRHL, expectation-qRHL,
that reasons about pairs of programs, and at the same time
supports expectations as the pre- and postconditions,
thus being as expressive as the calculi from \cite{dhondt06weakest,feng07proof,ying12floyd}
when it comes to the probabilistic behavior of the programs.

\paragraph{Related work.}
  The relevant prior work has already been discussed above.
  Concurrently and independently, \cite{theirs}
  presented a different formalization of expectation-qRHL.
  (The first versions on arXiv appeared within two months of each other.)
  The biggest difference is the definition of couplings
  which in our setting are separable quantum states, and in
  their setting nonseparable quantum states. Therefore, the soundness proofs are totally different in \cite{theirs} and in the present paper, even for the same rules.
  As a consequence, we can avoid having to reason about judgments with side-conditions, making compositional reasoning about more complex programs much easier\fullonly{ (see \autoref{sec:compare.theirs})}.
  \fullonly{A unique contribution of the present paper is that relational properties between non-terminating programs can be proved in our logic, while the logic in \cite{theirs} is mainly designed for terminating programs.} \fullonly{We give a detailed technical comparison with \cite{theirs} in \autoref{sec:compare.theirs}.}

\paragraph{Organization.} In \autoref{sec:var.mem.pred} we
  introduce notation and preliminaries, including the concept of
  expectations. In \autoref{sec:qprogs} we give syntax and semantics of
  the imperative quantum programming language that we study.
  In \autoref{sec:qrhl} we give the definition of expectation-qRHL.
  In \autoref{sec:soundness}, we give the soundness proofs of all our rules.
  In \autoref{sec:rules}, we present sound rules
    for reasoning about
    expectation-qRHL judgments. And in \autoref{sec:example}, we analyze the quantum Zeno
  effect as an example of using our logic.

\shortonly{In the full version, 
we give a detailed comparison of our logic with \cite{theirs} and full proofs of our results.}

\fullonly{  In \autoref{sec:compare.theirs}, we give a detailed comparison of our logic with \cite{theirs}.
    }
  
\section{Preliminaries: Variables, Memories, and Predicates}
\label{sec:var.mem.pred}

In this section, we introduce some fundamental concepts and notations
needed for this paper, and recap some of the needed quantum background
as we go along.
When introducing some notation $X$, the place of definition is marked like this: \symbolindexmarkhighlight{$X$}. \fullonly{All symbols are listed in the symbol index.\shortonly{\interfootnotelinepenalty=0 \footnote{Symbol index:
    \fboxsep=0pt
    \renewcommand\symbolindexpage[2]{\hyperref[#2]{p.#1}}%
    \renewcommand\symbolindexentry[4]{%
      \ifx\@nil#4\@nil\else
      \mbox{\colorbox{gray!20}{$#2$}\,#4}%
      \space
      \fi
    }%
  \renewenvironment{thesymbolindex}{}{}%
  \printsymbolindex
}}}
For further mathematical background we recommend \cite{conway97functional,conway00operator},
and for an introduction to quantum mechanics \cite{nielsenchuang-10year}.

\paragraph{Variables.}
Before we introduce the syntax and semantics of programs, we first
need to introduce some basic concepts. A
\emph{variable}\index{variable} is described by a variable name \symbolindexmark\xx{$\xx,\yy,\zz$}
that identifies the variable, and a nonempty type $T$.  The \emph{type}%
\index{type!(of a variable)} of $\xx$
is simply the nonempty set of all (classical) values the variable can
take. E.g., a variable might have type $\bit$,
or $\setN$.{\footnote{We
  stress that we do not assume that the type is a finite or even a
  countable set. Consequently, the Hilbert spaces considered in this
  paper are not necessarily finite dimensional or even
  separable. However, all results can be informally understood by
  thinking of all sets as finite and hence of all Hilbert spaces as
  $\setC^N$
  for suitable $N\in\setN$.}}
Lists\pagelabel{page:variable.conventions} or sets of variables will be denoted
\symbolindexmark\XX{$\XX,\YY,\ZZ$}. 
Given a list
$\XX=\xx_1\dots\xx_n$
of variables, we say its \emph{type} \index{type!(of a list of
  variables)} is $T_{1}\times\dots\times T_{n}$
if $T_{i}$ is the type of $\xx_i$.
We write \symbolindexmark\varconcat{$\varconcat{\XX\YY}$} for the concatenation/disjoint union
of lists/sets of variables $\XX,\YY$.

\paragraph{Memories and quantum states.}
An \emph{assignment}\index{assignment} assigns to each variable a
classical value. Formally, for a set $\XX$,
the \emph{assignments over $\XX$}
are all functions \symbolindexmark\mm{$\mm$}
with domain $\XX$
such that: for all $\xx\in\XX$
with type $T_\xx$,
$\mm(\xx)\in T_\xx$.
That is, assignments can represent the content of classical memories.

To model quantum memories, we simply consider superpositions of
assignments: A \emph{(pure) quantum memory}%
\index{pure quantum memory}%
\index{quantum memory!(pure)}%
\index{memory!(pure) quantum} is a superposition of
assignments. Formally, \symbolindexmark\elltwov{$\elltwov{\XX}$},
the set of all quantum memories over ${\XX}$,
is the Hilbert space with basis{\footnote{When we say ``basis'', we
  always mean an orthonormal Hilbert-space basis.}} $\{\ket\mm\}_\mm$
where $\mm$
ranges over all assignments over~${\XX}$. Here \symbolindexmark\ket{$\ket\mm$} simply denotes the
basis vector labeled $\mm$. We often write $\ket\mm_{{\XX}}$ to stress
which space we are talking about.
We call a quantum memory $\psi$ \emph{normalized}\index{normalized} iff $\norm{\psi}=1$.
Intuitively, a normalized quantum memory over ${\XX}$
represents a state a quantum computer with variables ${\XX}$
could be in.  We also consider quantum states over
arbitrary sets $X$ (as opposed to sets of assignments).
Namely, \symbolindexmark\elltwo{$\elltwo X$}
denotes the Hilbert space with orthonormal basis $\{\ket x\}_{x\in X}$.
(In that notation, $\elltwov{\XX}$
is simply $\elltwo A$
where $A$
is the set of all assignments on ${\XX}$.)
Normalized elements of $\elltwov\XX$
represent quantum states.

We often treat elements of $\elltwo T$ and $\elltwov{\XX}$ interchangeably if $T$ is
the type of ${\XX}$ since there is a natural isomorphism between those spaces.

In many situations, we additionally need an additional symbol
  $\bot$
  that denotes that a memory is not available because the program did
  not terminate. A \emph{quantum $\bot$-memory
    over $\XX$}%
  \index{quantum $\bot$-memory}%
  \index{$\bot$-memory!quantum}
  is an element of $\symbolindexmark\elltwovb{\elltwovb\XX}:=\elltwo{A\cup\{\bot\}}$
  where $A$
  is the set of all assignments on $\XX$.
  That is, a quantum $\bot$-memory is a superposition between a quantum memory and $\ket\bot$.

The tensor product \symbolindexmark\tensor{$\tensor$} combines two quantum states
$\psi\in\elltwo X,\phi\in\elltwo Y$ into a joint system
$\psi\otimes\phi\in\elltwo{X\times Y}$. In the case of quantum
memories $\psi,\phi$ over ${\XX},{\YY}$, respectively,
$\psi\otimes\phi\in\elltwov{{\XX\YY}}$.
(And in this case, $\psi\otimes\phi=\phi\otimes\psi$ since we are composing
``named'' systems.)

We will need to consider entangled pairs of memories. Specifically,
  a \emph{quantum bimemory}%
  \index{quantum bimemory}%
  \index{bimemory!quantum} over $\XX_1,\XX_2$
  is an element of $\elltwov{\XX_1}\otimes\elltwov{\XX_2}=\elltwov{\XX_1\XX_2}$.
  Similarly, a \emph{quantum $\bot$-bimemory}%
  \index{quantum $\bot$-bimemory}%
  \index{$\bot$-bimemory!quantum}
  is an element of
  $\elltwovb{\XX_1}\otimes\elltwovb{\XX_2}$,
  i.e., a tensor product of two quantum $\bot$-memories.
  (Note: ``one-sided-$\bot$'' states such as $\ket\mm\otimes\ket{\bot}$ are included in this.)
  For clarity, we often write $\ket{\bot_1},\ket{\bot_2}$ instead of $\bot$
  to emphasize whether we are talking about elements of
  $\elltwovb{\XX_1}$ or $\elltwovb{\XX_2}$.

For a vector (or operator) $a$, we write \symbolindexmark\adj{$\adj
  a$} for its adjoint.  (In the finite dimensional case, the adjoint
is simply the conjugate transpose of a vector/matrix. The literature
also knows the notation $a^\dagger$.)
The adjoint of a vector $\ket x$ is also written as \symbolindexmark\bra{$\bra x$}.
We abbreviate \symbolindexmarkonly\proj$\symbolindexmarkhighlight{\proj{\psi}}:=\psi\adj\psi$.
This is the projector onto $\psi$ when $\norm\psi=1$.

\paragraph{Mixed quantum memories.}
In many situations, we need to model probabilistic quantum states
(e.g., a quantum state that is $\ket0$
with probability $\frac12$
and $\ket1$
with probability $\frac12$). This is modeled using \emph{mixed
  states}%
\index{mixed state}%
\index{state!mixed} (a.k.a.~\emph{density operators}%
\index{density operator}%
\index{operator!density}). Having normalized state $\psi_i$
with probability $p_i$
is represented by the operator $\rho:=\sum_i p_i\proj{\psi_i}$.\footnote{Mathematically,
  density operators are positive Hermitian trace-class operators on
  $\elltwo X$.
  The requirement ``trace-class'' ensures that the trace exists and
  can be ignored in the finite-dimensional case.}\fullonly{\,\footnote{Sums without
    index set are always assumed to have an arbitrary (not necessarily finite or even
    countable) index set.  In the case of sums of vectors in a Hilbert space,
    convergence is with respect to the Hilbert space norm, and in the case of sums
    of positive operators, the convergence is with respect to the Loewner order.}}
In particular, $\proj{\psi}$ is the density operator of a pure quantum state $\psi$. Then $\rho$
encodes all observable information about the distribution of the
quantum state (that is, two distributions of quantum states have the
same~$\rho$
iff they cannot be distinguished by any physical process).  And
$\tr\rho$
is the total probability $\sum_ip_i$.
Note that we do not formally impose the condition $\tr\rho=1$ or $\tr\rho\leq 1$
  unless explicitly specified.
  We call a mixed state \emph{normalized}\index{normalized} iff $\tr\rho=1$.
We will often need
to consider mixed states of quantum memories (i.e., mixed states with
underlying Hilbert space $\elltwov{\XX}$). We call them
\emph{mixed (quantum) memories} over ${\XX}$.%
\index{mixed (quantum) memory}%
\index{quantum memory!mixed}%
\index{memory!mixed (quantum)}
Analogously, we define
  \emph{mixed bimemories}%
  \index{mixed bimemory}%
  \index{bimemory!mixed}
  and 
  \emph{mixed $\bot$-bimemories}%
  \index{mixed $\bot$-bimemory}%
  \index{$\bot$-bimemory!mixed}
  as mixed states of quantum ($\bot$-)bimemories.

Let \symbolindexmarkonly\Pbot\pagelabel{page:def:Pbot}$\symbolindexmarkhighlight\Pbot:=\proj{\ket\bot}$ and \symbolindexmarkonly\Pnbot$\symbolindexmarkhighlight\Pnbot:=\id-\Pbot$. We can easily access the terminating and non-terminating part of a mixed $\bot$-memory:
\symbolindexmarkonly\restrictbot$\symbolindexmarkhighlight{\restrictbot(\rho)}:=\Pbot\rho\Pbot$
and
$\symbolindexmarkhighlight{\restrictnbot(\rho)}:=\Pnbot\rho\Pnbot$
are the memory $\rho$ after measuring that we have termination
or do not have termination, respectively.
  
For a mixed ($\bot$-)bimemory $\rho$
over ${\XX_1}{\XX_2}$ the \emph{partial trace}%
\index{trace!partial}%
\index{partial trace}
\symbolindexmark\partr{$\partr{i}\rho$}\label{page:partr} ($i=1,2$)
is the result of throwing away the left/right memory
(i.e., it is a mixed memory over ${\XX_i}$).
Formally, $\partr{i}$ is defined as the continuous linear function satisfying
$\partr{1}(\sigma\otimes\tau):=\tau\cdot\tr\sigma$,
$\partr{2}(\sigma\otimes\tau):=\sigma\cdot\tr\tau$.

A mixed memory $\rho$
is \emph{$({\XX},{\YY})$-separable}\index{separable}
(i.e., not entangled between ${\XX}$ and ${\YY}$) iff it can be written as
$\rho=\sum_i\rho_i\otimes\rho_i'$ for mixed memories
$\rho_i,\rho_i'$ over ${\XX},{\YY}$, respectively. (Potentially infinite sum.)
When $\XX,\YY$ are clear from the context, we simply say \emph{separable}.\footnote{%
  Note that this is not the only possible definition of separability.
  (And maybe, from a physics point of view, not even the most natural one.)
  For example, \cite{holevo05entanglement} show that the set of separable states as defined here is not closed w.r.t.~the trace-norm.
  Instead, they define the set of separable states as the closure of the set we consider.
  (This is, roughly speaking, equivalent to taking an integral instead of the sum in 
  $\sum_i\rho_i\otimes\rho_i'$.)
  However, our definition of separability works well in our setting.
  We do not know whether the definition from \cite{holevo05entanglement} would give us the same set of reasoning rules.}
In this paper, when we write infinite sums of operators, convergence is always
  with respect to the trace norm: $\Vert A\Vert_{\tr}:=\tr\sqrt{A\adj A}$. (In the finite-dimensional case, the choice of norm
  is irrelevant since all norms are equivalent then.)

\paragraph{Operations on quantum states.}
An operation in a closed quantum system is modeled by an isometry $U$
on $\elltwo X$.\footnote{That
  is, a norm-preserving linear operation. Often, one models quantum
  operations as unitaries instead because in the finite-dimensional
  case an isometry is automatically unitary. However, in the
  infinite-dimensional case, unitaries are unnecessarily
  restrictive. Consider, e.g., the isometry $\ket i\mapsto\ket{i+1}$
  with $i\in\setN$
  which is a perfectly valid quantum operation but not a unitary.} If
we apply such an operation on a mixed state $\rho$,
the result is $U\rho \adj U$. In particular, denote by \symbolindexmark\id{$\id$} the identity operation, i.e. $\id\psi=\psi$ for all pure states $\psi$ in this space.

An operator $A$ on $\elltwov\XX$ can be interpreted as an operator
  on  $\elltwovb\XX$ by setting $A\ket\bot:=0$.
  To avoid confusion, we often write \symbolindexmark\liftb{$\liftb A$} for the operator on $\elltwovb\XX$.
  Similarly, an operator $A$ on $\elltwov{\XX_1}\otimes\elltwov{\XX_2}$ can be seen
  as an operator on  $\elltwovb{\XX_1}\otimes\elltwovb{\XX_2}$,
  we write \symbolindexmark\liftbb{$\liftbb A$}\pagelabel{page:liftbb} for the operator on  $\elltwovb{\XX_1}\otimes\elltwovb{\XX_2}$.

Most often,
isometries will occur in the context of operations that are performed
on a single variable or list of variables, i.e., an isometry $U$
on $\elltwov{\XX}$.
Then $U$
can also be applied to $\elltwov{\YY}$
with ${\YY}\supseteq{\XX}$:
we identify $U$
with $U\otimes\id_{{\YY}\setminus{\XX}}$.
Furthermore, if ${\XX}$
has type $T$,
then an isometry $U$
on $\elltwo T$
can be seen as an isometry on $\elltwov{\XX}$
since we identify $\elltwo T$ and $\elltwov{\XX}$. If we want to make
$\XX$ explicit, we write \symbolindexmark\opon{$\opon
  U{\XX}$}\label{page:opon} for the isometry $U$ on $\elltwov{\YY}$.
For example, if $U$ is a $2\times 2$-matrix and $\xx$ has type bit,
then $\opon U\xx$ can be applied to quantum memories over $\xx\yy$,
acting on $\xx$ only.
This notation is not limited to isometries, of course, but applies
to other operators, too.
(By \emph{operator}\index{operator}\pagelabel{page:bounded} we always mean a bounded linear operator in this paper.)

  In slight overloading of notation, we also write $\opon U\XX$ for $U$
  acting on $\bot$-memories, where $\oponp U\XX\ket\bot=0$.
  (That is, $\opon U\XX$ is short for the more precise $\liftb{\oponp U\XX}$.)
  We also write $\opon U{\XX_1}$ for $U$ acting on $\bot$-bimemories.
  In this case, we simply have $\opon U{\XX_1}:=\oponp U{\XX_1}\otimes\id$.
  In particular, $\oponp U{\XX_1}\paren{\ket\mm\otimes\ket\bot}
  = \oponp U{\XX_1}\ket\mm\otimes\ket\bot$ but
  $\oponp U{\XX_1}\paren{\ket\bot\otimes\ket\mm}=0$.
  Analogously for $\opon U{\XX_2}$.

We will use only binary measurements in this paper. A \emph{binary measurement}%
  \index{measurement!binary}\index{binary measurement} $M$ on $\elltwov\XX$
  has outcomes $\true,\false$ and is described by two bounded operators
  \symbolindexmark\Mtrue{$\Mtrue M$},
  \symbolindexmark\Mfalse{$\Mfalse M$} on $\elltwov\XX$ that satisfy $\adj{\Mtrue M}{\Mtrue M}+\adj{\Mfalse M}{\Mfalse M}=\id$,
  its \emph{Kraus operators}.%
  \index{Kraus operator}\index{operator!Kraus}
  Given a mixed memory $\rho$, the probability of measurement outcome $t$ is
  $p_t:=\tr \Mtruefalse Mt\rho\adj{\Mtruefalse Mt}$, and the post-measurement state is
  $\Mtruefalse Mt\rho\adj{\Mtruefalse Mt}/p_t$.

\paragraph{Expectations.}   In this work,
we will use expectations as pre- and postconditions in Hoare
judgments.  The idea of using expectations originated in \cite{kozen83probabilistic}
for reasoning about (classical) probabilistic programs.
Intuitively, an expectation is a quantitative predicate, that is
for any memory (or bimemory, in our case), it does not tell us whether the memory satisfies the predicate
but \emph{how much} it satisfies the predicate.
Thus, classically, an expectation is simply a function from assignments to reals.
By analogy, in the quantum setting, one might want to define expectations, e.g., as functions $f$ from
quantum bimemories to reals (i.e., an expectation would be a function $\elltwov\XX\to\setRpos$).
However, such expectations might behave badly, for example, it is not clear
how to compute the expected value $f(\psi)$ for a random $\psi$ if
the distribution of $\psi$ is given in terms of a density operator.
A better approach was introduced by \cite{dhondt06weakest}. Following their approach, we
define an \emph{expectation}\index{expectation} as a positive operator $\PA$ on quantum bimemories.\footnote{%
  Recall from \autopageref{page:bounded} that operators are always bounded in our context.
  This means that $\PA$ is bounded, too. 
  This means that the values that an expectation $\PA$ can assign
  to states are between $0$ and $B$ for some finite $B$.}
(We use letters \symbolindexmark\PA{$\PA,\PB,\PC,\dots$} for expectations in this paper.)
This expectation then assigns the value $\adj\psi\PA\psi$ to the quantum memory $\psi$
(equivalently, $\tr\PA\,\proj\psi$).
To understand this, it is best to first look at the special case where $\PA$ is a projector.
Then $\adj\psi\PA\psi=1$ iff $\psi$ is in the image of $\PA$, and 
$\adj\psi\PA\psi=0$ iff $\psi$ is orthogonal to the image of $\PA$.
Such an $\PA$ is basically a predicate (by outputting $1$ for states that satisfy the predicate).
Of course, states that neither satisfy the predicate nor are orthogonal to it
will output a value between $0$ and $1$.
Any expectation $\PA$ can be written as $\sum_i p_i\PA_i$ with projectors $\PA_i$.
Thus, $\PA$ would give $p_i$ ``points'' for satisfying the predicate $\PA_i$.
In this respect, expectations in the quantum setting are similar to classical ones:
classical expectations give a certain amount of ``points'' for each possible classical
input.

The nice thing about this formalism is that, given a density operator $\rho=\sum p_i\proj{\psi_i}$,
we can easily compute the expected value of the expectation $\PA$.
More precisely, the expected value of $\adj\psi\PA\psi=\tr\PA\,\proj\psi$ with $\psi:=\psi_i$ with probability $p_i$.
That expected value is $\sum p_i\tr\PA\,\proj{\psi_i} = \tr\PA\paren{\sum p_i\proj{\psi_i}} = \tr\PA\rho$.
This shows that we can evaluate how much a density operator satisfies the expectation $\PA$ by just
computing $\tr\PA\rho$. This formula will be the basis for our definitions!

\fullonly{(A note for physicists: an expectation $\PA$ in our setting is nothing else but an observable,
and $\tr\PA\rho$ is the expected value of the outcome of measuring the observable $\PA$ when the system is
in state $\rho$.)}

Analogously, we define \emph{$\bot$-expectations}%
  \index{$\bot$-expectation}\pagelabel{page:bot-expectation}
  on quantum $\bot$-bimemories.
  However, we add one restriction: The value of a $\bot$-expectation
  should not change if we measure whether the $\bot$-bimemory
  is in $\ket\bot$
  or not. Formally, a $\bot$-expectation
  is a positive operator on quantum $\bot$-bimemories that is invariant under $\Ebot\otimes\Ebot$
  where
  \symbolindexmarkonly\Ebot\pagelabel{page:def:Ebot}$\symbolindexmarkhighlight{\Ebot}(\rho):=
  \Pbot \rho \Pbot +
  \Pnbot \rho \Pnbot$.
  ($\Ebot$ corresponds to measuring and forgetting whether a given mixed $\bot$-memory is $\ket\bot$ or not.)
Note that for an expectation $\PA$, the operator $\liftbb \PA$ is a $\bot$-expectation.
  We can thus see expectations as special cases of $\bot$-expectations.

A very simple example of an expectation would be the matrix
$\PA:=\scriptscriptstyle
\begin{pmatrix}
  1\\&\frac12
\end{pmatrix}
$ that assigns $1$ to $\ket0$, and $\frac12$ to $\ket1$.
And given the density operator $\rho=\frac12\id$ (representing a uniform qubit),
$\tr\PA\rho=\frac34$ are intuitively expected.

\fullonly{\paragraph{Free variables.}
Given an expectation $\PA$,
we will often wish to indicate which variables it talks about, i.e.,
what are its \emph{free variables}%
\index{free variables}%
\index{variables!free}. Since our definition of expectations is semantic
(i.e., we are not limited to expectations expressed using a particular syntax) we cannot simply speak about the variables occurring in the
expression describing $\PA$.
Instead, we say $\PA$
contains only variables from ${\YY}$
(written: \symbolindexmarkonly\fv$\symbolindexmarkhighlight{\fv(\PA)}\subseteq{\YY}$)
iff there exists an expectation $\PA'$ over $\YY$
such that $\PA=\PA'\otimes\id$.
Note that there is a certain abuse of notation here: We formally
defined ``$\fv(\PA)\subseteq{\YY}$'',
but we do not define $\fv(\PA)$;
$\fv(\PA)\subseteq{\YY}$
should formally just be seen as an abbreviation for
``there exists $\PA$ over $\YY$ such that $\PA=\PA'\otimes\id$''.\fullonly{\footnote{In
  fact, defining $\fv(\PA)$
  is possible only if there is a smallest set ${\YY}$
  such that $\exists \PA'.\ \PA=\PA'\otimes\id$.
  This is not necessarily the case. For example:
  Let $\bullet_\xx$ denote an arbitrary element of the type of $\xx$ for all variables~$\xx$.
  For a set $\XX$ of variables,
  let $\PA_{\XX}\ket \mm:=\ket\mm$
  for all assignments $\mm$ over $\XX$ where $\mm(\xx)\neq\bullet_{\xx}$    only for finitely
  many $\xx$. Let $\PA_{\XX}\ket \mm:=0$ otherwise.
  Then $\PA_{\XX}=\PA_{\YY}\otimes\id$ for all co-finite $\YY\subseteq\XX$.
  But for any non-co-finite $\YY\subseteq\XX$, $\PA_{\XX}\neq \PB\otimes\id$ for all $\PB$ over $\YY$.
  So $\fv(\PA_{\XX})$ would have to be the smallest co-finite subset of $\XX$.
  But if $\XX$ is infinite, there is no smallest co-finite subset of $\XX$.
}}}

\paragraph{Quantum equality.}%
\index{quantum equality}%
\index{equality!quantum}
In \cite{qrhl}, a specific predicate \symbolindexmark\QUANTEQ{$\XX_1\QUANTEQ\XX_2$}
was introduced to describe the
fact that two quantum variables (or list of quantum variables) have the same state.
Formally, $\XX_1\QUANTEQ\XX_2$ is the subspace consisting of all quantum memories in $\elltwov{\XX_1\XX_2}$ that
are invariant under \symbolindexmark\SWAPXX{$\SWAPXX{\XX_1}{\XX_2}$}, the unitary that swaps variables $\XX_1$ and $\XX_2$.%
\footnote{That is, $\SWAPXX{\XX_1}{\XX_2}(\psi\otimes\phi)=\phi\otimes\psi$ for $\psi\in\elltwov{\XX_1}$, $\phi\in\elltwov{\XX_2}$.}
Or equivalently, $\XX_1\QUANTEQ\XX_2$ denotes the subspace spanned by all quantum memories of
the form $\phi\otimes\phi$ with  $\phi\in\elltwov{\XX_1}=\elltwov{\XX_2}$. We write $\opon{\symbolindexmark\equal\equal}{\XX_1\XX_2}$ for the projector onto $\XX_1\QUANTEQ\XX_2$.

\section{Quantum programs}\label{sec:qprogs}

\paragraph{Syntax.}
We will now define a small imperative quantum language.\footnote{Very
    similar to \cite{ying12floyd}, except that we replace their case-statement by an if-statement
    and allow initialization with arbitrary states instead of just $\ket0$.
    }
The set of all
programs is described by the following syntax:\symbolindexmarkonly\bc
\[
  \symbolindexmarkhighlight{\bc,\bd} ::=
  \apply U\XX
  \ |\
  \init \XX\psi
  \ |\
  \ifte M{\XX}\bc\bd
  \ |\
  \while M{\XX}\bc
  \ |\
  \bc;\bd
  \ |\
  \SKIP
  \ |\ \abort
\]
Here $\XX$ is a list of  variables and $U$
an isometry on $\elltwov\XX$,
$\psi\in\elltwov\XX$ a normalized state, and $M$ is a binary measurement on $\elltwov\XX$.
(There are no fixed sets of allowed $U$ and $\psi$, any isometry/state that we can
describe can be used here).\footnote{We will assume throughout the
  paper that all programs satisfy those well-typedness constraints. In
  particular, rules may implicitly impose type constraints on the
  variables and constants occurring in them by this assumption.}

Intuitively, \symbolindexmark\apply{$\apply U\XX$}
means that the operation $U$
is applied to the quantum variables~$\XX$.
E.g., $\apply H\xx$
would apply the Hadamard gate to the variable $\xx$
(we assume that $H$
denote the Hadamard matrix). It is important that we can apply $U$
to several variables $\XX$ simultaneously,
otherwise no entanglement between variables can ever be produced.

The program \symbolindexmark\init{$\init\XX\psi$}
initializes the variables $\XX$
with the quantum state $\psi$.
The program \symbolindexmark\ifte{$\ifte M{\XX}\bc\bd$}
will measure the 
  variables $\XX$ with the measurement $M$,
and, if the outcome is $\true$, execute~$\bc$, otherwise execute~$\bd$.

The program \symbolindexmark\while{$\while M{\XX}\bc$}
measures $\XX$,
and if the outcome is $\true$,
it executes~$\bc$. This is repeated until the outcome is $\false$.

Finally, \symbolindexmark;{$\bc;\bd$}
executes $\bc$
and then $\bd$.
And \symbolindexmark\SKIP{$\SKIP$}
does nothing. We will always implicitly treat ``$;$''\pagelabel{page:seq.assoc.skip.neutral}
as associative and $\SKIP$ as its neutral element. \symbolindexmark\abort{$\abort$} never terminates.

\paragraph{Semantics.}
The \emph{denotational semantics}%
\index{denotational semantics}%
\index{semantics!denotational} of our programs $\bc$
are represented as completely positive trace-reducing maps
\symbolindexmark\denot{$\denot\bc$} on the mixed memories over
$\XXall$, defined
by recursion on the structure of the programs.  Here \symbolindexmark\XXall{$\XXall$}\pagelabel{page:XXall}
is a fixed set of program variables, and we will assume that
all variables under consideration are contained in this set.\fullonly{\footnote{We fix some set $\XXall$
  in order to avoid a more cumbersome notation $\denot\bc^{\XX}$
  where we explicitly indicate the set $\XX$ of program
  variables with respect to which the semantics is defined.}}
The obvious cases are $\denot\SKIP:=\id$
and $\denot{\bc;\bd}:=\denot\bd\circ\denot\bc$
and $\denot\abort(\rho):=0$.
And application of an isometry $U$
is also fairly straightforward given the syntactic sugar introduced
above:
$\denot{\apply U\XX}(\rho):=\paren{\opon U\XX}\rho\adj{\paren{\opon
    U\XX}}$. (The notation $\opon U\XX$ was introduced on \autopageref{page:opon}.)

Initialization of quantum variables is slightly more complicated:
$\init\XX\psi$
initializes the variables $\XX$
with $\psi$,
which is the same as removing $\XX$,
and then creating a new variable $\XX$
with content $\psi$.
Removing $\XX$
is done by the operation $\partr\XX$
(partial trace, see \autopageref{page:partr}). And creating new variables $\XX$ in state $\psi$
is done by the operation $\otimes\proj{\psi}$.
Thus we define $\denot{\init\XX\psi}(\rho):=\partr\XX\rho\otimes\proj{\psi}$.

The if-command first performs a measurement and then branches
depending on the outcome. We then have that the state after measurement (without
renormalization) is
$\paren{\opon{\Mtruefalse Mt}{\XX}}\rho\adj{\paren{\opon{\Mtruefalse Mt}{\XX}}}$ for outcome $t=\true,\false$. Then $\bc$
or $\bd$
is applied to that state and the resulting states are added together
to get the final mixed state. Altogether:\symbolindexmarkonly\restrict
\begin{multline*}
  \pb\denot{\ifte M{\XX}\bc\bd}(\rho)
  :=
  \denot\bc\pb\paren{\restrict{\true}(\rho)}
  +
  \denot\bd\pb\paren{\restrict{\false}(\rho)}
  \\
  \text{where}\qquad
  \symbolindexmarkhighlight{\restrict t(\rho)}:=
  \paren{\opon{\Mtruefalse Mt}{\XX}}
  \rho\adj
  {\paren{\opon{\Mtruefalse Mt}{\XX}}}
\end{multline*}
While-commands are modeled similarly: In an execution of a while
statement, we have $n\geq 0$
iterations of ``measure with outcome $\true$
and run $\bc$''
(which applies $\denot\bc\circ\restrict{\true}$
to the state), followed by ``measure with outcome $\false$''
(which applies $\restrict{\false}$
to the state). Adding all those branches up, we get the
definition:
\begin{equation*}
  \pb\denot{\while M{\XX}\bc}(\rho) :=
  \sum_{n=0}^\infty \restrict{\false} \pb\paren{ (\denot\bc\circ\restrict{\true})^n(\rho) }
\end{equation*}

We call a program $\bc$ \emph{terminating}\index{terminating} iff $\tr\denot\bc(\rho)=\tr\rho$ for all~$\rho$.   \paragraph{Semantics with explicit non-termination.}
  $\denot\bc$
  is not trace-preserving if $\bc$
  is not terminating. For technical reasons, we will need a variant of
  this function that is trace-preserving. This can be achieved by
  outputting a mixed state that has an explicit state
  $\proj{\ket\bot}$
  to denote non-termination. This semantic function $\denotbot\bc$ takes
  mixed $\bot$-memories to mixed $\bot$-memories and
  can be easily
  derived from $\denot\bc$:\symbolindexmarkonly\denotbot\pagelabel{page:denotbot}
  \[
    \symbolindexmarkhighlight{\denotbot\bc}(\rho):=\denot\bc(\restrictnbot(\rho))
    + \pb\paren{\tr\rho-\tr\denot\bc(\restrictnbot(\rho))}\proj{\ket\bot}.
    \ \footnote{It may not be immediately obvious that $\denotbot\bc$ is completely positive. Complete positivity is shown as follows:
    $f(\rho):=\tr\rho - \tr\denot\bc(\rho)$
    is a linear functional. For $\rho\geq0$, $f(\rho)\in[0,\tr\rho]$ since $\denot\bc$ is trace-reducing.
    Hence $f$ is positive.
    Every positive linear functional is completely positive \cite[Lemma~34.5]{conway00operator}.
    The function $g(c):=c\cdot \proj{\ket\bot}$ is completely positive.
    Thus $g\circ f$ is completely positive.
    Also $\denot\bc\circ\restrictnbot$ is completely positive as a composition of two completely positive functions.
    Thus $\denotbot\bc = \denot\bc\circ\restrictnbot + g\circ f$ is completely positive.
  }
  \]  %
  ($\Pbot,\Pnbot$
  are defined on \autopageref{page:def:Pbot}.)  Operationally,
  $\denotbot\bc$
  first measures if the state is $\bot$.
  If so, nothing happens.  Otherwise, $\bc$
  is applied. If $\bc$
  does not terminate, the final output memory is set to be $\bot$.
  $\denotbot\bc$ is easily seen to be trace-preserving.
    Moreover, we have the composition property
    $\denotbot{\bc;\bd}=\denotbot{\bd}\circ\denotbot{\bc}$, since
    \begin{align*}
      &\denotbot{\bd}(\denotbot{\bc}(\rho))=\denotbot{\bd}\left[\denot{\bc}(\rho_{\not\bot})+(\tr\rho-\tr\denot{\bc}(\rho_{\not\bot}))\proj{\ket\bot}\right]\\
      &=\denot{\bd}(\denot{\bc}(\rho_{\not\bot}))+\pb\bracket{\tr\denotbot\bc(\rho)-\tr\denot{\bd}(\denot{\bc}(\rho_{\not\bot}))}\proj{\ket\bot}\\
      &=\denot{\bc;\bd}(\rho_{\not\bot})+(\tr\rho-\tr\denot{\bc;\bd}(\rho_{\not\bot}))\proj{\ket\bot}=\denotbot{\bc;\bd}(\rho).
        \end{align*}
\section{qRHL with expectations}
\label{sec:qrhl}

\paragraph{Defining the logic.}
  We now present our definition of expectation-qRHL. We follow the
  approach from \cite{qrhl} to use separable couplings to describe the
  relationship between programs. A \emph{coupling}\index{coupling}
  between two mixed states $\rho_1$ and $\rho_2$
    is a mixed state $\rho$ that has $\rho_1$ and $\rho_2$ as marginals. (That is, $\partr{\XX_2}\rho=\rho_1$
  and $\partr{\XX_1}\rho=\rho_2$ if $\rho_1,\rho_2$ are states over $\XX_1,\XX_2$, respectively.)
  This is analogous to probabilistic couplings: a coupling of distributions $\mu_1,\mu_2$
  is a distribution $\mu$ with marginals $\mu_1,\mu_2$.  
  Note that couplings trivially always exist if $\rho_1$ and $\rho_2$ have the same trace
  (namely, $\rho:=\rho_1\tensor\rho_2/\tr\rho_1$).
  Couplings become interesting when
  we put additional constraints on the state $\rho$.
  For example, if we require the support of $\rho$ to be in the subspace $C:=\SPAN\{\ket{00},\ket{11}\}$,
  then $\rho_1=\proj{\ket0}$ and $\rho_2=\proj{\ket0}$ have a coupling
  (namely, $\rho=\proj{\ket{00}}$),
  as do $\rho_1=\proj{\ket1}$ and $\rho_2=\proj{\ket1}$ 
  (namely, $\rho=\proj{\ket{11}}$),
  but not 
  $\rho_1=\proj{\ket0}$ and $\rho_2=\proj{\ket1}$.
  Things become particularly interesting when $\rho_1,\rho_2$ are not pure states.
  E.g., 
  $\rho_1=\frac12\proj{\ket0}+\frac12\proj{\ket1}$ and
  $\rho_2=\frac12\proj{\ket0}+\frac12\proj{\ket1}$ have such a coupling as well
  (namely, $\rho=\frac12\proj{\ket{00}}+\frac12\proj{\ket{11}}$
  but $\rho:=\rho_1\tensor\rho_2$ is not a coupling with support in $C$).

  Thus, a subspace such as $C$ can be seen as a predicate describing the relationship of $\rho_1,\rho_2$.
  The states $\rho_1,\rho_2$ satisfy $C$ iff there is a coupling with support in $C$.
  This idea leads to the following tentative definition of qRHL:
  \begin{definition}[qRHL, tentative, without expectations]\label{def:qrhl.tent}
    For subspaces $A$, $B$ (i.e., spaces of quantum memories over $\XXall_1\XXall_2$),
    $\rhlinformal A\bc\bd B$ holds iff for any $\rho_1,\rho_2$ that have a coupling
    with support in $A$,
    the final states $\denot\bc(\rho_1),\denot{\bd}(\rho_2)$ have a coupling with support in $B$.
  \end{definition}
  However, it was noticed in \cite{qrhl} that the definition becomes easier to handle if
  we impose another condition on the couplings.
  Namely, the coupling should be separable, i.e., there should be no entanglement between
  the two systems corresponding to $\rho_1,\rho_2$.
  That is, the definition of qRHL used in \cite{qrhl}
  is \autoref{def:qrhl.tent} with ``coupling'' replaced by ``separable coupling''.
  We will also adopt the separability condition in our definition of expectation-qRHL.\footnote{\label{footnote:why.sep}%
    \cite{qrhl} was not able to prove the \textsc{Frame} rule without adding this
    separability condition. Our reasons for adopting the separability condition
    are slightly different: we do not have a \textsc{Frame} rule anyway,
    but for other elementary rules such as the rule \textsc{Equal} in \cite{qrhl} with quantum expectations and quantum variables, it is unclear how to prove them
    without the separability condition.}

  So far, we have basically recapped the definition from \cite{qrhl}.
  However, that definition only allows us to express Hoare judgments that do not involve expectations
  since $A$ and $B$ in \autoref{def:qrhl.tent} are subspaces (predicates), not expectations.
  To define expectation-qRHL, we follow the same idea,
  but instead of quantifying over only the initial states satisfying the precondition,
  we quantify over all initial states, and merely require that (the coupling of) the final states
  satisfies the postexpectation at least as much as (the coupling of the) initial states satisfy the
  preexpectation. That is:
  \begin{definition}[Expectation-qRHL (eqRHL), first attempt]\label{def:eqrhl.first}\index{eqRHL}
    For expectations $\PA,\PB$,
    $\rhlinformal A\bc\bd B$ holds iff for any $\rho_1,\rho_2$ with separable coupling $\rho$,
    the final states $\denot\bc(\rho_1),\denot{\bd}(\rho_2)$ have a separable coupling $\rho'$
    such that $\tr\PA\rho\leq\tr\PB\rho'$.
    (Recall that $\tr\PA\rho$ indicates how much $\rho$ satisfies $\PA$,
    and analogously $\tr\PB\rho'$, cf.~\autoref{sec:var.mem.pred}.)
  \end{definition}
  
For terminating programs $\bc,\bd$, this definition already works well.
    Unfortunately, if $\bc,\bd$ do not terminate with probability $1$,
    we have some undesired effects: For example, assume that $\bc=\SKIP$,
    and $\bd$ is a program that with probability $1-\varepsilon$ does nothing ($\SKIP$),
    and but with probability $\varepsilon$ does not terminate ($\abort$).
    Then $\rhlinformal\PA\bc\bd\PB$ does not hold \emph{for any $\PA,\PB$}. Why?
    The final states $\denot\bc(\rho_1),\denot{\bd}(\rho_2)$ have trace $1$ and $1-\varepsilon$,
    respectively.
    Therefore there exists no coupling $\rho'$ of $\denot\bc(\rho_1),\denot{\bd}(\rho_2)$.
    (It follows from the definition of couplings that the coupling must have the same trace
    as its marginals.)
    Hence $\rhlinformal\PA\bc\bd\PB$ does not hold.
    Similarly,  $\rhlinformal\PA\bc\bd\PB$ does not hold whenever there are two input states
    $\rho_1,\rho_2$ such that $\bc,\bd$ terminate with different probabilities. (Even if this
    nontermination only occurs for input states for which $\PA$ evaluates to $0$!)
    This makes it near impossible to reason about non-terminating programs.\footnote{%
      Even if we are interested in a Hoare logic with total correctness,
      this behavior is undesired.
      Instead, we want that a nontermination with small probability simply introduces
      some small penalty in the expectations.
      For example, in the case of non-relational Hoare with total
      correctness,
      $\hl{(1-\varepsilon)\id}\PA{\id}$ means that $\PA$ is nonterminating with probability
      $\leq\varepsilon$.}

  There are a number of approaches how one can circumvent this
  problem. E.g., one could allow $\rho'$ to be a ``subcoupling'' instead of a coupling (i.e.,
  its marginals do not have to equal  $\denot\bc(\rho_1),\denot{\bd}(\rho_2)$
  but only lower bound them);\footnote{This is explored further in \autoref{sec:reformulate} for special cases of our definition.}
  the subcoupling always exists, even if the traces are not equal.
  However, we find that adding such ``hacks'' to the definition makes it more difficult
  to understand what the definition exactly does in case of non-terminating programs.
  
  Instead, we choose an approach that makes non-termination explicit.
  That is, when a program does not terminate, we assign a specific state $\ket\bot$ to its output,
  and we allow expectation to explicitly refer to it (e.g., an expectation could assign value $1$ to
  nontermination, and value $0$ to termination).
  The denotation $\denotbot\bc$ defined on \autopageref{page:denotbot} does exactly that.
  And expectations that live on a space that has an explicit nontermination-state $\ket\bot$
  were introduced as $\bot$-expectations on \autopageref{page:bot-expectation}. This leads to the following definition:

\begin{definition}[Expectation-qRHL, informal]\label{def:eqrhl.informal}
  For $\bot$-expectations $\PA,\PB$,
  $\rhlinformal A\bc\bd B$ holds iff for any $\rho_1,\rho_2$ with separable coupling $\rho$,
  the final states $\denotbot\bc(\rho_1),\denotbot{\bd}(\rho_2)$ have a separable coupling $\rho'$
  such that $\tr\PA\rho\leq\tr\PB\rho'$.
\end{definition}
Note that a coupling of  $\denotbot\bc(\rho_1),\denotbot{\bd}(\rho_2)$ always exists
since $\denotbot\cdot$ is trace-preserving.
(Below, we will derive certain specific variants of eqRHL such as eqRHL with total correctness
as specific cases of this definition. Also, we will see that subcouling-based definitions
can be recovered as special cases in \autoref{lemma:qrhl.equiv}.) By plugging in the definition of couplings into \autoref{def:eqrhl.informal}, we get the following precise definition:

\begin{definition}[Expectation-qRHL, generic]\label{def:eqrhl.generic}
  Let $\PA$, $\PB$ be $\bot$-expectations and $\bc$, $\bd$  programs.
  Then \symbolindexmark\rhlgen{$\rhlgen\PA\bc\bd\PB$} holds iff for any separable mixed $\bot$-bimemory $\rho$ over $\XXall_1,\XXall_2$,
  there is a separable mixed $\bot$-bimemory $\rho'$ over $\XXall_1,\XXall_2$ such that
  \begin{compactitem}
  \item $\partr{2}\rho' \couprel \denotbot\bc(\partr{2}\rho)$.
  \item $\partr{1}\rho' \couprel \denotbot\bd(\partr{1}\rho)$.
  \item $\tr\PA\rho\leq\tr\PB\rho'$.
  \end{compactitem}
\end{definition}

In this definition, $\XXall_1,\XXall_2$ are isomorphic copies of the set $\XXall$ of variables.
  That is, while strictly speaking, $\denotbot\bc$ maps mixed $\bot$-memories over $\XXall$ to
  mixed $\bot$-memories over $\XXall$, we can also see it as mapping
  mixed $\bot$-memories over $\XXall_1$ to
  mixed $\bot$-memories over $\XXall_1$. Analogously for $\bd$ and $\XXall_2$.
  We make use of this in the preceding definition when we apply $\denotbot\bc,\denotbot\bd$ to
  $\rho_1,\rho_2$, respectively. In the rest of the paper, we simply call a mixed state \emph{$(\rho_1,\rho_2)$-coupling} if it is separable and has marginals $\rho_1$ and $\rho_2$.

  \medskip

Note that we defined $\bot$-expectations
    to be invariant under $\Ebot\otimes\Ebot$
    (\autopageref{page:bot-expectation}), i.e., that they implicitly
    measure first whether the quantum memories are $\ket\bot$.
    Otherwise, we would not even have $\rhlgen\PA\SKIP\SKIP\PA$
    (\ruleref{gen}{Skip} below), for example if
    $\PA:={\proj{\fsq\ket0+\fsq\ket\bot}}$.
    This is because even the program $\SKIP$
    measures whether the memory is $\ket\bot$
    (by definition of $\denotbot\cdot$), so it may change the state if the memory is in a superposition
    between $\ket\bot$ and something else.

  \paragraph{Partial/total correctness.}
The generic definition of eqRHL (\autoref{def:eqrhl.generic}) allows
    us to explicitly express in our pre-/postexpectations how nontermination should
    be treated. While this allows for maximal generality, in practice it might
    be cumbersome to always have to specify explicitly how the expectations behave on $\ket\bot$.
    Instead, we define below three special cases of eqRHL that hardcode the treatment of $\ket\bot$.

    The simplest case is eqRHL with total correctness: Here, nontermination is ``forbidden'', i.e.,
    we assign value $0$ to it.
    Recall from \autopageref{page:liftbb} that for an expectation
    $\PA$,
    $\liftbb\PA$
    is the corresponding $\bot$-expectation.
    It assigns $0$ to a state that has $\ket{\bot}$
    in the left or right memory.
    Hence, eqRHL with total correctness simply means that all pre/postconditions are of
    the form $\liftbb\PA$.
    The following definition specifies convenient syntax for this special case:
  
  \begin{definition}[Expectation-qRHL, total]\label{def:tot}
  Let $\PA$,
  $\PB$
  be expectations and $\bc$,
  $\bd$
  programs.  
  Then \symbolindexmark\rhltot{$\rhltot\PA\bc\bd\PB$}
  iff
  $\rhlgen{\liftbb\PA}\bc\bd{\liftbb\PB}$.
\end{definition}

A second common variant of Hoare logic is ``partial
  correctness''. Here we allow non-termination, i.e., if a program
  does not terminate, we assign the value $1$.
  That is, we use pre/postexpectations of the form $\liftbbp\PA+T$
  where $T$ assigns value $1$ when the left or right memory is in state $\ket\bot$:
  
\begin{definition}[Expectation-qRHL, partial]
  Let $\PA$,
  $\PB$
  be expectations and $\bc$,~$\bd$
  programs. 
  Then \symbolindexmark\rhlpar{$\rhlpar\PA\bc\bd\PB$}
  iff
  $\rhlgen{\liftbbp\PA + T}\bc\bd{\liftbbp\PB + T}$
  where $T:={\pb\paren{\proj{\ket{\bot_1}}\otimes\Pnbot}
    + \pb\paren{\Pnbot\otimes\proj{\ket{\bot_2}}}
    + \proj{\ket{\bot_1}\otimes\ket{\bot_2}}
  }$.
\end{definition}

Unfortunately, this definition does not necessarily behave as we would like. E.g.,
if both $\bc$ and $\bd$ terminate with probability $\leq\frac12$ on all inputs,
then $\rhlpar{\PA}\bc\bd\PB$ holds for any $\PA\leq\id$ and any $\PB$.
That is, any property holds with probability $1$ for those programs
which is not what we would expect!
Why does this happen? Since $\denotbot\bc(\rho_1),\denotbot\bd(\rho_2)$ are 50\% nontermination,
we can ``match up'' the nonterminating part of $\denotbot\bc(\rho_1)$ with the terminating part
of $\denotbot\bd(\rho_2)$ and vice versa in the coupling $\rho'$ of the output states.
Then $\tr T\rho'=1$ and thus $\rhlpar{\PA}\bc\bd\PB$ holds.
The problem here is that we treat nontermination as a ``wildcard''
that is allowed to match any behavior of the other program.
While there may be valid use cases for such a notation,
we believe that in most cases this is not what we want.

Instead, we define a notion we call ``semipartial''.
In this eqRHL-variant, we allow nontermination, but only when it occurs in
the two programs ``in sync''.
I.e., we consider pre/postexpectations that assign $1$ to $\ket\bot\otimes\ket\bot$,
but $0$ to a state where one program has nonterminated and the other has terminated. Formally:
\begin{definition}[Expectation-qRHL, semipartial]
  Let $\PA$,  $\PB$ be expectations and $\bc$, $\bd$ programs.  
  Then \symbolindexmark\rhlsemi{$\rhlsemi\PA\bc\bd\PB$} iff
  $$\pB\rhlgen{\liftbbp\PA+\proj{\ket{\bot_1}\otimes\ket{\bot_2}}}\bc\bd{\liftbbp\PB+\proj{\ket{\bot_1}\otimes\ket{\bot_2}}}.$$
\end{definition}

\begin{figure}
  \centering
    \lineskip=8pt
    $\RULE{any}{Skip}{}{\rhlany\PA\SKIP\SKIP\PA}$
    \qquad
    $\RULE{any}{Apply1}{}{\pb\rhlany{\adj{\paren{\opon U{\XX_1}}}\PA\paren{\opon U{\XX_1}}}{\apply U\XX}\SKIP\PA}$
    \qquad
    $\RULE{any}{ExFalso}{}{\rhlany 0\bc\bd\PB}$
    \qquad
    $\RULE{any}{Init1}{}{\pb\rhlany{\id_{\XX_1}\otimes\paren{\adj\psi\otimes\id_{\negXX\XX_1}}\PA\paren{\psi\otimes\id_{\negXX\XX_1}}}{\init\XX\psi}{\SKIP}{\PA}}$
    \qquad
    $\RULE{any}{Seq}{\rhlany{\PA}{\bc_1}{\bd_1}{\PB}\\ \rhlany{\PB}{\bc_2}{\bd_2}{\PC}}{\rhlany{\PA}{\bc_1;\bc_2}{\bd_1;\bd_2}{\PC}}$
    \qquad
    $\RULE{any}{Conseq}{\PA'\le \PA\\ \rhlany\PA\bc\bd\PB\\ \PB\le \PB'}{\rhlany{\PA'}\bc\bd{\PB'}}$
    \qquad
    $\RULE{any}{Sym}{\rhlany \PA\bc\bd\PB}{\rhlany{\SWAPlr^\ast\cdot \PA\cdot \SWAPlr}\bd\bc{\SWAPlr^\ast\cdot \PB\cdot \SWAPlr}}$
    \qquad
    $\RULE{any}{Scale}{\rhlany \PA\bc\bd\PB\\ \lambda\in [0,1]}{\rhlany{\lambda\PA}\bc\bd{\lambda\PB}}$
    \qquad
    $\RULE{any}{If1}{\rhlany{\PA_T}{\bc_T}{\bd}{\PB}\\ \rhlany{\PA_F}{\bc_F}{\bd}{\PB}}{\rhlany{\restrictast{\true}(\PA_T)+\restrictast{\false}(\PA_F)}{\ifte{M}{\XX}{\bc_T}{\bc_F}}{\bd}{\PB}}$
    \qquad
    $\RULE{any}{JointIf9}{\rhlany{\PA_{t,u}}{\bc_t}{\bd_u}\PB\text{ for $t,u\in\{\true,\false\}$}}{\pb\rhlany{\textstyle\sum_{t,u\in\{\true,\false\}}\restrictast{t,u}(\PA_{t,u})}{\ifte M\XX{\bc_{\true}}{\bc_{\false}}}{\ifte N\YY{\bd_{\true}}{\bd_{\false}}}{\PB}}$
    \qquad
{\small$ \RULE{any}{JointIf}
    {
      \rhlany{\PA_{\true}}
      {\bc_{\true}}{\bd_{\true}}
      \PB
      \\
      \rhlany{\PA_{\false}}
      {\bc_{\false}}{\bd_{\false}}
      \PB
    }
    {
      \pb\rhlany{
        \restrictast{\true,\true}(\PA_{\true})
        +
        \restrictast{\false,\false}(\PA_{\false})
      }
      {\ifte M\XX{\bc_{\true}}{\bc_{\false}}}
      {\ifte N\YY{\bd_{\true}}{\bd_{\false}}}
      {\PB}
    }$}
    \qquad
    $\RULE{any}{While1}{\pb\rhlany{\PA}{\bc}{\SKIP}{\restrictast{\true}(\PA)+\restrictast{\false}(\PB)}\\ \while M\XX\bc\text{ is terminating}}{\pb\rhlany{\restrictast{\true}(\PA)+\restrictast{\false}(\PB)}{\while M\XX\bc}{\SKIP}{\PB}}$
    \qquad
$\RULE{any}{JointWhile}{\rhlany{\PA}{\bc}{\bd}{\restrictast{\true,\true}(\PA)+\restrictast{\false,\false}(\PB)}\\ \while M\XX\bc\text{  or  }\while N\YY\bd\text{ is terminating}}{\rhlany{\restrictast{\true,\true}(\PA)+\restrictast{\false,\false}(\PB)}{\while M\XX\bc}{\while N\YY\bd}{\PB}}$

\caption[Rules for total/semipartial/partial eqRHL]{%
  \label{fig:rules}
    In these rules,
    ``$\mathsf{any}$'' can be any of ``$\mathsf{tot}$'', ''$\mathsf{semi}$'', ''$\mathsf{par}$''.
    For $\mathsf{any}=\mathsf{par}$, the termination condition
    in \rulerefx{any}{While1} can
    be replaced by $\PA\leq\id$, and for  $\mathsf{any}=\mathsf{par},\mathsf{semi}$
    the termination condition
    in \rulerefx{any}{JointWhile} can
    be replaced by $\PA\leq\id$.
    We refer to the symmetric rules of \rulerefx{any}{Apply1},
    \rulerefx{any}{Init1}, \rulerefx{any}{If1}, and \rulerefx{any}{While1} (obtained by
    applying \rulerefx{any}{Sym}) as \implicitRULE{any}{Apply2},
    \implicitRULE{any}{Init2}, \implicitRULE{any}{If2}, and
    \implicitRULE{any}{While2}.
      }
\end{figure}

\paragraph{Pure initial states.}
In many cases, it is much easier to work with the definition of eqRHL correctness
  if one can assume that the initial states of $\bc,\bd$ are pure states,
  and that the initial coupling is the tensor product of those states.
  (No nontrivial correlations.) The following lemma shows that we can do
  so without loss of generality:

\begin{lemma}\label{lemma:qrhl.pure.new}
  Let $\PA$, $\PB$ be $\bot$-expectations and $\bc$, $\bd$ programs.
  Then $\rhlgen\PA\bc\bd\PB$ holds iff
  \begin{compactitem}
  \item for all unit quantum memories $\psi_1$, $\psi_2$ over $\XX_1$, $\XX_2$, respectively,
    there is a separable $\bot$-mixed bimemory $\rho'$ over $\XX_1,\XX_2$ such that
    \begin{compactitem}
    \item $\partr{2}\rho' \couprel \denotbot\bc(\proj{\psi_1})$.
    \item $\partr{1}\rho' \couprel \denotbot\bd(\proj{\psi_2})$.
    \item $\tr \PA\,\proj{\psi_1\otimes\psi_2}\leq\tr\PB\rho'$.\footnote{Or equivalently,
        $\norm{\sqrt{\PA}\paren{\psi_1\otimes\psi_2}}\leq\tr\PB\rho'$. Or
        $\adjp{\psi_1\otimes\psi_2}\PA\paren{\psi_1\otimes\psi_2}\leq\tr\PB\rho'$.
      }
    \end{compactitem}
  \item for all unit quantum memories $\psi_1$
    over $\XX_1$, we have
    $\tr \PA\,\proj{\psi_1\otimes\ket{\bot_2}}\leq\tr\PB
    \pb\paren{\denotbot\bc(\proj{\psi_1})\otimes\proj{\ket{\bot_2}}}$.
  \item for all unit quantum memories $\psi_2$
    over $\XX_2$, we have
    $\tr \PA\,\proj{\ket{\bot_1}\otimes\psi_2}\leq\tr\PB
    \pb\paren{\proj{\ket{\bot_1}}\otimes\denotbot\bd(\proj{\psi_2})}$.
  \item 
    $\tr \PA\,\proj{\ket{\bot_1}\otimes\ket{\bot_2}}\leq\tr\PB\,
    \proj{\ket{\bot_1}\otimes\ket{\bot_2}}$.
  \end{compactitem}
\end{lemma}

\delaytextsv{proof lemma:qrhl.pure.new}{
\begin{proof}
  Let $P(\psi_1,\psi_2)$
  mean that there exists a separable $\bot$-mixed
  bimemory $\rho'$
  such that $\partr{2}\rho' \couprel \denotbot\bc(\proj{\psi_1})$,
  and $\partr{1}\rho' \couprel \denotbot\bd(\proj{\psi_2})$,
  and $\tr \PA\,\proj{\psi_1\otimes\psi_2}\leq\tr\PB\rho'$

  Notice that $P(\psi_1,\ket{\bot})$ holds iff there exists
  a separable $\bot$-mixed
  bimemory $\rho'$ such that
  \begin{inparaenum}[(i)]
  \item \label{enum:tr2rho'}
    $\partr{2}\rho' \couprel \denotbot\bc(\proj{\psi_1})$,
    and
  \item \label{enum:tr1rho'}  $\partr{1}\rho' \couprel \denotbot\bd(\proj{\ket{\bot}})=\proj{\ket{\bot}}$,
    and
  \item \label{enum:trleq} $\tr \PA\,\proj{\psi_1\otimes\ket{\bot_2}}\leq\tr\PB\rho'$.
  \end{inparaenum}
  The conditions \eqref{enum:tr2rho'},\eqref{enum:tr1rho'} are
  equivalent to
  $\rho'= \denotbot\bc(\proj{\psi_1})\otimes\proj{\ket{\bot_2}}$,
  and in this case \eqref{enum:trleq} becomes
  $\tr \PA\,\proj{\psi_1\otimes\ket{\bot_2}}\leq\tr\PB
  \pb\paren{\denotbot\bc(\proj{\psi_1})\otimes\proj{\ket{\bot_2}}}$.

  Thus $P(\psi_1,\ket{\bot})$ holds iff $\tr \PA\,\proj{\psi_1\otimes\ket{\bot_2}}\leq\tr\PB
  \pb\paren{\denotbot\bc(\proj{\psi_1})\otimes\proj{\ket{\bot_2}}}$.

  Analogously, $P(\ket{\bot},\psi_2)$ holds iff
  $\tr \PA\,\proj{\ket{\bot_1}\otimes\psi_2}\leq\tr\PB
  \pb\paren{\proj{\ket{\bot_1}}\otimes\denotbot\bd(\proj{\psi_2})}$.

  And $P(\ket\bot,\ket\bot)$ holds iff 
  $\tr \PA\,\proj{\ket{\bot_1}\otimes\ket{\bot_2}}\leq\tr\PB\,
  \proj{\ket{\bot_1}\otimes\ket{\bot_2}}$.

  Thus the statement of the lemma is equivalent to:
  \begin{equation}
    \rhlgen\PA\bc\bd\PB
    \iff    
    \forall\
    \psi_1\in\elltwov{\XX_1}\cup\{\ket{\bot_1}\},\
    \psi_2\in\elltwov{\XX_2}\cup\{\ket{\bot_2}\}.\ 
    P(\psi_1,\psi_2).
    \label{eq:qrhl.pure.restated}
  \end{equation}

  The $\Rightarrow$-direction is immediate from \autoref{def:eqrhl.generic}.
  
  We show the $\Leftarrow$-direction.
  Fix some separable $\bot$-mixed bimemory $\rho$ over $\XX_1\XX_2$.
  To prove that $\rhlgen\PA\bc\bd\PB$ holds, we need to construct a separable $\rho'$ such that:
  \begin{compactenum}[(i)]
  \item\label{item:marg1} $\partr{2}\rho' \couprel \denotbot\bc(\partr{2}\rho)$.
  \item\label{item:marg2} $\partr{1}\rho' \couprel \denotbot\bd(\partr{1}\rho)$.
  \item\label{item:exp} $\tr\PA\rho\leq\tr\PB\rho'$.
  \end{compactenum}
  Recall the definition of $\Ebot$ from \autopageref{page:def:Ebot}.
  Let $\rho^\bot:=(\Ebot\otimes\Ebot)(\rho)$.
  Then $\rho^\bot$ is also a separable mixed $\bot$-bimemory,
  thus we can write $\rho^\bot$ as  $\rho^\bot=\sum_jp_i\proj{\psi_{1j}\otimes\psi_{2j}}$
  for unit quantum memories $\psi_{1j},\psi_{2j}$ over $\XX_1,\XX_2$ and $p_j\geq0$.
  By definition of $\calE_\bot$, we can choose the 
  $\psi_{1j},\psi_{2j}$ as $\psi_{1j}\in\elltwov{\XX_1}\cup\{\ket{\bot_1}\}$
  and  $\psi_{2j}\in\elltwov{\XX_2}\cup\{\ket{\bot_2}\}$.
  (I.e., $\psi_{1j}$ is never a superposition between $\ket{\bot_1}$ and something else.
  Same for $\psi_{2j}$.)
  
  By assumption (rhs of \eqref{eq:qrhl.pure.restated}),
  for all $j$, we have $P(\psi_{1j},\psi_{2j})$.
  Thus,
  there exists a separable mixed $\bot$-bimemory $\rho'_j$ over $\XX_1\XX_2$ such that
  \begin{compactitem}
  \item $\partr{2}\rho'_j \couprel \denotbot\bc(\proj{\psi_{1j}})$.
  \item $\partr{1}\rho'_j \couprel \denotbot\bd(\proj{\psi_{2j}})$.
  \item $\tr \PA\,\proj{\psi_{1j}\otimes\psi_{2j}}\leq\tr\PB\rho'_j$.
  \end{compactitem}
  Then let $\rho':=\sum_jp_j\rho_j'$. Since all $\rho'_j$ have trace $1$,
  and $\sum_jp_j=\tr\rho\leq\infty$, $\rho'$ exists.
  
  We have \eqref{item:marg1} since
  \begin{align*}
    \partr{2}\rho'
    &=
    \sum\nolimits_j p_j \partr{2}\rho'_j
    =
    \sum\nolimits_j p_j \denotbot\bc\pb\paren{\proj{\psi_{1j}}}
    =
    \denotbot\bc\pb\paren{\textstyle\sum\nolimits_j p_j\proj{\psi_{1j}}}
    \\
    &=
      \denotbot\bc\pb\paren{\partr{2}\rho^\bot}
      =
      \denotbot\bc\pb\paren{\calE_\bot(\partr{2}\rho)}
      \starrel=
      \denotbot\bc\pb\paren{\partr{2}\rho}.
  \end{align*}
  Here $(*)$ holds since $\denotbot\bc\circ\calE_\bot=\denotbot\bc$ by construction of
  $\denotbot\bc$ and $\calE_\bot$.
  
  And \eqref{item:marg2} is shown analogously.

  And \eqref{item:exp} follows since
  \begin{align*}
    \tr\PA\rho
    &\starrel=
      \tr(\calE_\bot\otimes\calE_\bot)(\PA)\rho
      \starstarrel=
      \sum_{t,u\in\bot,\not\bot}
      \tr
      (P_{t}\otimes P_u)\PA(P_t\otimes P_u)\rho
    \\
    &\tristarrel=
      \sum_{t,u\in\bot,\not\bot}
      \tr
      \PA(P_t\otimes P_u)\rho(P_{t}\otimes P_u)
      =
      \tr
      \PA(\calE_\bot\otimes\calE_\bot)(\rho)
      =
      \tr
      \PA\rho^\bot
      \\
    &=
    \sum\nolimits_j p_j \tr\PA\proj{\psi_{1j}\otimes\psi_{2j}}
    \leq
    \sum\nolimits_j p_j \tr\PB\rho_j'
    =
    \tr\PB\rho'.
  \end{align*}
  Here $(*)$ holds since $\PA$ is a $\bot$-expectation and those are invariant under
  $\calE_\bot\otimes\calE_\bot$ by definition. And $(**)$ is by definition of $\calE_\bot$.
  And $(*\mathord**)$ is by the circularity of the trace.

  Thus we have shown \eqref{item:marg1}-\eqref{item:exp}, so 
  $\rhlgen\PA\bc\bd\PB$ holds.
\end{proof}
}

\paragraph{Equivalent reformulations.}
\label{sec:reformulate} As discussed after \autoref{def:eqrhl.first}, an alternative means of
trying to circumvent the problem that \autoref{def:eqrhl.first} does
handle nonterminating programs well is to use subcouplings instead of
couplings.

Here, we show that the notions of eqRHL with partial and
total correctness can be equivalently restated in terms of
subcouplings (instead of the extended semantics $\denotbot\cdot$
over $\bot$-memories).
However, we do not know such an equivalent reformulation for
semipartial correctness.

\begin{lemma}\label{lemma:qrhl.equiv}
  Let $\PA$, $\PB$ be expectations and $\bc$, $\bd$ programs.
  Then $\rhltot\PA\bc\bd\PB$ iff for any separable mixed bimemory $\rho$ over $\XXall_1,\XXall_2$,
  there is a separable mixed bimemory $\rho'$ over $\XXall_1,\XXall_2$ such that
  \begin{compactitem}
  \item $\partr{2}\rho' \le \denot\bc(\partr{2}\rho)$.
  \item $\partr{1}\rho' \le \denot\bd(\partr{1}\rho)$.
  \item $\tr\PA\rho\leq\tr\PB\rho'$.
  \end{compactitem}
\end{lemma}

\delaytextsv{proof lemma:qrhl.equiv}{
  \begin{proof}We first prove the ``only if'' part, namely, we need to construct $\rho'$ for any given $\rho$ such that the three conditions hold.  Let $\rho_1$ and $\rho_2$ be the marginals of $\rho$. Consider the $\bot$-mixed bimemory $\rho\oplus 0_{\bot\bot}$ as the input coupling, then $\rhltot{\PA}\bc\bd{\PB}$ implies that the output states $\denotbot{\bc}(\rho_1\oplus 0_{\bot_1})=\denot{\bc}(\rho)+(\tr\rho_1-\tr\denot{\bc}(\rho_1))\proj{\ket\bot}$ and $\denotbot{\bd}(\rho_2\oplus 0_{\bot_2})=\denot{\bd}(\rho_2)+(\tr\rho-\tr\denot{\bd}(\rho_2))\proj{\ket\bot}$ have a separable coupling $\hat{\rho}$ such that
$$\tr\PA\rho=\tr(\PA\oplus 0_{\bot\bot})(\rho\oplus 0_{\bot\bot})\le\tr(\PB\oplus 0_{\bot\bot})\hat{\rho}=\tr(\PB\oplus 0_{\bot\bot})(\Pnbot\otimes \Pnbot)\hat{\rho}(\Pnbot\otimes \Pnbot).$$
Now let $(\Pnbot\otimes \Pnbot)\hat{\rho}(\Pnbot\otimes \Pnbot)=\rho'\oplus 0_{\bot\bot}$, then $\rho'$ is separable and $\tr\PA\rho\le\tr\PB\rho'$. In order to show that $\rho'$ is exactly the bimemory we need to construct, it remains to prove that $\rho'$ is a subcoupling of the output state $\denot{\bc}(\rho_1)$ and $\denot{\bd}(\rho_2)$, which is achieved by \begin{equation}\begin{split}&\denot{\bc}(\rho_1)=\Pnbot(\tr_2\hat{\rho})\Pnbot=\tr_2\rho'+\tr_2(\Pnbot\otimes\Pbot)\hat{\rho}(\Pnbot\otimes\Pbot)\\
&\denot{\bd}(\rho_2)=\Pnbot(\tr_1\hat{\rho})\Pnbot=\tr_1\rho'+\tr_1(\Pbot\otimes\Pnbot)\hat{\rho}(\Pbot\otimes\Pnbot).\label{eq:subc.gap}\end{split}\end{equation}

Then we prove the ``if'' part. We derive $\rhlgen{\liftbb\PA}\bc\bd{\liftbb\PB}$ by \autoref{lemma:qrhl.pure.new}. It suffices to construct for any normalized quantum memories $\psi_1$, $\psi_2$ over $\XXall_1$, $\XXall_2$, a separable coupling $\rho$ of the output $\bot$-memories $\denotbot{\bc}(\proj{\psi_1})$ and $\denotbot{\bd}(\proj{\psi_2})$ such that $\tr\PA\proj{\psi_1\otimes\psi_2}\le\tr(\PB\oplus 0_{\bot\bot})\rho$, noting that the other three conditions for input bimemories $\psi_1\otimes\ket{\bot_2}$, $\ket{\bot_1}\otimes\psi_2$, and $\ket{\bot_1}\otimes\ket{\bot_2}$ have natrually been satisfied for $\bot$-expectations $\liftbb\PA$ and $\liftbb\PB$. The three premises here allow us to find memories $\sigma_1$ and $\sigma_2$ with separable coupling $\sigma$ such that $\sigma_1\le\denot{\bc}(\proj{\psi_1})$, $\sigma_2\le\denot{\bd}(\proj{\psi_2})$ and $\tr\PA\proj{\psi_1\otimes\psi_2}\le\tr\PB\sigma$. Without loss of generality, we assume that $\tr\denot{\bc}(\proj{\psi_1})\ge\tr\denot{\bd}(\proj{\psi_2})$. If $\tr\denot{\bc}(\proj{\psi_1})=\tr\sigma$, it is easy to verify that $\sigma$ is actually a coupling of the output memories and it suffices to choose $\rho=\sigma\oplus 0_{\bot\bot}$. For $\tr\denot{\bc}(\proj{\psi_1})>\tr\sigma=\tr\sigma_1$, let
\begin{equation*}\begin{split}&\rho=\sigma+\frac{\tr\denot{\bc}(\proj{\psi_1})-\tr\denot{\bd}(\proj{\psi_2})}{\tr\denot{\bc}(\proj{\psi_1})-\tr\sigma}(\denot{\bc}(\proj{\psi_1})-\sigma_1)\otimes\proj{\ket\bot}\\
&+\frac{(\denot{\bc}(\proj{\psi_1})-\sigma_1)\otimes(\denot{\bd}(\proj{\psi_2})-\sigma_2)}{\tr\denot{\bc}(\proj{\psi_1})-\tr\sigma}+(1-\tr\denot{\bc}(\proj{\psi_1}))\proj{\ket\bot\otimes\ket\bot}.\end{split}\end{equation*}
Then it is easy to verify that $\rho$ is a separable coupling of $\denotbot{\bc}(\proj{\psi_1})=\denot{\bc}(\proj{\psi_1})+(1-\tr\denot{\bc}(\proj{\psi_1}))\proj{\ket\bot}$ and $\denotbot{\bd}(\proj{\psi_2})=\denot{\bd}(\proj{\psi_2})+(1-\tr\denot{\bd}(\proj{\psi_2}))\proj{\ket\bot}$, and $\tr\PA\proj{\psi_1\otimes\psi_2}\le\tr\PB\sigma\le\tr(\PB\oplus 0_{\bot\bot})\rho$.\end{proof} 
}

\begin{lemma}\label{lemma:qrhl.equiv.par}
  Let $\PA$, $\PB$ be expectations and $\bc$, $\bd$ programs.
  Then $\rhlpar\PA\bc\bd\PB$ holds iff for any separable mixed bimemory $\rho$ over $\XXall_1,\XXall_2$,
  there is a separable mixed bimemory $\rho'$ over $\XXall_1,\XXall_2$ such that
  \begin{compactitem}
  \item $\partr{2}\rho' \le \denot\bc(\partr{2}\rho)$.
  \item $\partr{1}\rho' \le \denot\bd(\partr{1}\rho)$.
  \item $\tr\rho'\ge \tr\denot\bc(\partr{2}\rho)+\tr\denot\bd(\partr{1}\rho)-\tr\rho$.
  \item $\tr\PA\rho\leq\tr\PB\rho'+\tr\rho-\tr\rho'$.
  \end{compactitem}
\end{lemma}

\delaytextsv{proof lemma:qrhl.equiv.par}{
\begin{proof} We first prove the ``only if'' part. We construct $\rho'$ for input bimemories $\rho_1$ and $\rho_2$ with coupling $\rho$ in a similar way to the proof of \autoref{lemma:qrhl.equiv}. Consider the $\bot$-mixed quantum memory $\rho\oplus 0_{\bot\bot}$ as the input coupling, then $\rhlpar{\PA}\bc\bd{\PB}$ implies that the output states $\denotbot{\bc}(\rho_1\oplus 0_{\bot_1})=\denot{\bc}(\rho_1)+(\tr\rho_1-\tr\denot{\bc}(\rho_1))\proj{\ket\bot}$ and $\denotbot{\bd}(\rho_2\oplus 0_{\bot_2})=\denot{\bd}(\rho_2)+(\tr\rho_2-\tr\denot{\bd}(\rho_2))\proj{\ket\bot}$ have a separable coupling $\hat{\rho}$ such that
\begin{multline*}$$\tr\PA\rho=\tr(\liftbb\PA+T)(\rho\oplus 0_{\bot\bot})\le\tr(\liftbb\PB+T)\hat{\rho}\\
=\tr(\PB\oplus 0_{\bot\bot})(\Pnbot\otimes \Pnbot)\hat{\rho}(\Pnbot\otimes \Pnbot)+\tr\rho-\tr(\Pnbot\otimes\Pnbot)\hat{\rho},\end{multline*}
where $T:={\pb\paren{\proj{\ket{\bot_1}}\otimes\Pnbot}
    + \pb\paren{\Pnbot\otimes\proj{\ket{\bot_2}}}
    + \proj{\ket{\bot_1}\otimes\ket{\bot_2}}
  }$. Now let $(\Pnbot\otimes \Pnbot)\hat{\rho}(\Pnbot\otimes \Pnbot)=\rho'\oplus 0_{\bot\bot}$, then $\rho'$ is separable, $\tr\PA\rho\le\tr\PB\rho'+\tr\rho-\tr\rho'$. Moreover, it follows from \autoref{eq:subc.gap} that $\rho'$ is subcoupling of the output memories $\denot{\bc}(\rho_1)$ and $\denot{\bd}(\rho_2)$, and
\begin{multline*}\tr\rho'+\tr\rho-\tr\denot\bc(\partr{2}\rho)-\tr\denot\bd(\partr{1}\rho)=\tr\rho'+\tr\hat{\rho}-\tr\rho'-\tr(\Pnbot\otimes\Pbot)\hat{\rho}-\tr\rho'-\tr(\Pbot\otimes\Pnbot)\hat{\rho}\\
=\tr\hat{\rho}-\tr(\Pnbot\otimes\Pnbot)\hat{\rho}-\tr(\Pnbot\otimes\Pbot)\hat{\rho}-\tr(\Pbot\otimes\Pnbot)\hat{\rho}=\tr(\Pbot\otimes\Pbot)\hat{\rho}\ge 0.\end{multline*}

Then we prove the ``if'' part. We derive $\rhlgen{\liftbb\PA+T}\bc\bd{\liftbb\PB+T}$ by \autoref{lemma:qrhl.pure.new}. It suffices to construct for any normalized quantum memories $\psi_1$, $\psi_2$ over $\XXall_1$, $\XXall_2$, a separable coupling $\rho$ of the output $\bot$-memories $\denotbot{\bc}(\proj{\psi_1})$ and $\denotbot{\bd}(\proj{\psi_2})$ such that $\tr\PA\proj{\psi_1\otimes\psi_2}\le\tr(\PB\oplus 0_{\bot\bot}+T)\rho$, noting that the other three conditions for input bimemories $\psi_1\otimes\ket{\bot_2}$, $\ket{\bot_1}\otimes\psi_2$, and $\ket{\bot_1}\otimes\ket{\bot_2}$ have natrually been satisfied for $\bot$-expectations $\liftbb\PA+T$ and $\liftbb\PB+T$. The three premises here allow us to find memories $\sigma_1$ and $\sigma_2$ with separable coupling $\sigma$ such that $\sigma_1\le\denot{\bc}(\proj{\psi_1})$, $\sigma_2\le\denot{\bd}(\proj{\psi_2})$, $1+\tr\sigma-\tr\denot{\bc}(\proj{\psi_1})-\tr\denot{\bd}(\proj{\psi_2})\ge 0$, and $\tr\PA\proj{\psi_1\otimes\psi_2}\le\tr\PB\sigma+1-\tr\sigma$. Let
\begin{multline*}\rho=\sigma+(\denot{\bc}(\proj{\psi_1})-\sigma_1)\otimes\Pbot+\Pbot\otimes(\denot{\bd}(\proj{\psi_2})-\sigma_2)
\\ +(\tr\sigma+\tr\rho-\tr\denot{\bc}(\proj{\psi_1})-\tr\denot{\bd}(\proj{\psi_2}))\Pbot\otimes\Pbot.\end{multline*}
Then it is easy to verify that $\rho$ is the mixed $\bot$-bimemory we need.\end{proof}
}

In this definition, $\tr\rho-\tr\rho'\ge 0$ describes the nontermination probability. \autoref{lemma:qrhl.equiv} and \autoref{lemma:qrhl.equiv.par} mean that total and partial correctness can be alternatively defined using the concept of subcouplings, without considering the $\bot$-extension of the expectations and programs.  

\begin{lemma}\label{lemma:qrhl.pure}
  Let $\PA$, $\PB$ be expectations and $\bc$, $\bd$ programs.
  Then $\rhltot\PA\bc\bd\PB$ (resp. $\rhlpar\PA\bc\bd\PB$) holds iff for all unit quantum memories $\psi_1$, $\psi_2$ over $\XXall_1$, $\XXall_2$, respectively,
  there is a separable mixed bimemory $\rho$ over $\XX_1\XX_2$ such that
  \begin{compactitem}
    \item $\partr{2}\rho\le \denot\bc(\proj{\psi_1})$.
    \item $\partr{1}\rho\le \denot\bd(\proj{\psi_2})$.
    \item $\tr \PA\proj{\psi_1\otimes\psi_2}\leq\tr\PB\rho$.\footnote{Or equivalently,
        $\norm{\sqrt{\PA}\paren{\psi_1\otimes\psi_2}}\leq\tr\PB\rho$. Or
        $\adjp{\psi_1\otimes\psi_2}\PA\paren{\psi_1\otimes\psi_2}\leq\tr\PB\rho$.
      }\\(resp. $\tr \PA\proj{\psi_1\otimes\psi_2}\leq\tr\PB\rho+1-\tr\rho$ and $1+\tr\rho\ge\tr\denot\bc(\proj{\psi_1})+\tr\denot\bd(\proj{\psi_2})$.)
  \end{compactitem}
\end{lemma}

\delaytextsv{proof lemma:qrhl.pure}{
\begin{proof} We prove by using \autoref{lemma:qrhl.equiv} (resp. \autoref{lemma:qrhl.equiv.par}). For any mixed separable bimemory $\rho$ as input, let $\rho=\sum_ip_i\psi_i\otimes\phi_i$ where $p_i\in[0,1]$, $\psi_i$ and $\phi_i$ are normialized pure quantum memories. Then from the premises, let $\rho_i$ be the output bimemories corresponding to input $\psi_i\otimes\phi_i$. It is easy to verify that $\rho'=\sum_ip_i\rho_i$ is the ouput bimemory corresponding to $\rho$.\end{proof}
}

\section{Description of the rules}\label{sec:rules}

We describe the rules of our logic one by one here. Recall that we
essentially have four different logics: partial, semipartial, total,
and the general case from which the former three are derived. To keep
things readable, we only describe the rules for the partial,
semipartial, and total case here. (Listed in Figure~\ref{fig:rules}.)  In
\fullshort{\autoref{sec:soundness}}{the full version}, we state
and prove the rules in the general case. The rules in
Figure~\ref{fig:rules} are then simple consequences of the rules in the
general case. The sole exception are rules related to while-loops:
here not all of the partial, semipartial, total case follow directly
from the general while rule. Those cases that do not follow are proved
separately.

\subsection{Structural rules}

First, we mention the ``structural'' rules, i.e., rules that are not
related to a specific language construct. There is \rulerefx{any}{Sym} for
exchanging the two programs.
(In this rule, \symbolindexmarkonly\SWAPlr$\symbolindexmarkhighlight\SWAPlr:\elltwov{\XXall_2}\otimes\elltwov{\XXall_1}\mapsto\elltwov{\XXall_1}\otimes\elltwov{\XXall_2}$
is the unitary operator $\SWAPlrbot:\psi\otimes\phi \mapsto \phi\otimes\psi$.)
\rulerefx{any}{ExFalso} allows us to show anything from an impossible
preexpectation. \rulerefx{any}{Seq} allows us to analyze the sequential
composition of programs. \rulerefx{any}{Conseq} allows us to weaken a
judgment. (The preexpectation can be replaced by a smaller preexpectation,
and the postexpectation can be replaced by a larger preexpectation. $\leq$
is the Loewner order). And finally, \rulerefx{any}{Scale} allows us to
scale pre- and postexpectations by a scalar factor.

\subsection{One-sided rules}

Conceptually simplest are the one-sided rules, i.e., rules that have
$\SKIP$
on the right (or left) hand side. By combining them with \rulerefx{any}{Seq},
we can prove facts about pairs of programs one statement at the time.
Here, we only describe the rules with $\SKIP$ on the right side,
  the other case is analogous.

\medskip\noindent\textbf{Apply:}
First, consider the \rulerefx{any}{Apply1} rule. It is stated (like all our
rules), in a backward reasoning style, i.e., for any postexpectation
$\PA$,
the tells us the corresponding preexpectation, here
${\adj{\paren{\opon U{\XX_1}}}\PA\paren{\opon U{\XX_1}}}$.
(Recall that $\opon U{\XX_1}$
denotes $U$
applied to $\XX_1$.)
This is quite intuitive: the left program applies $\opon U{\XX_1}$,
so the preexpectation corresponding to the postexpectation $\PA$
is what we get if we apply $\opon U{\XX_1}$
and then compute the preexpectation, i.e.,
${\adj{\paren{\opon U{\XX_1}}}\PA\paren{\opon U{\XX_1}}}$.
(And it is ${\adj{\paren{\opon U{\XX_1}}}\PA\paren{\opon U{\XX_1}}}$
and not ${\PA\paren{\opon U{\XX_1}}}$
because the latter is not Hermitian and thus not a valid expectation.)

A toy example how to apply this rule: we want to what $\xx$
has to be so that it will be $\ket0$
after applying a Hadamard $H$.
Thus our postexpectation is $\opon{\proj{\ket0}}{\xx_1}$.
Applying \ruleref{any}{Apply1}, we get that
$\rhlany{\PB}{\apply H\xx}\SKIP{\opon{\proj{\ket0}}{\xx_1}}$
for
$\PB:=\adj{\oponp H{\xx_1}}\oponp{\proj{\ket0}}{\xx_1}{\oponp
  H{\xx_1}}$. (Here $\simtemplate{any}$
can be any of $\simtemplate{tot}$,
$\simtemplate{semi}$, $\simtemplate{par}$.) A simple calculation reveals: 
$\PB=\pb\opon{\paren{\adj H\proj{\ket0}H}}{\xx_1}
= \opon{\proj{\ket+}}{\xx_1}$.
Thus we learned (unsurprisingly) that to get $\ket0$, we need to start out with $\ket+$.

\medskip\noindent\textbf{Init:} The \ruleref{any}{Init1} rule stated in a
similar backwards reasoning way as \rulerefx{any}{Apply1}, but the
preexpectation is somewhat less intuitive. We will illustrate it with
a toy example.  Assume we want to know what the probability is to
measure $\ket0$ after initializing a variable $\xx$ with $\ket+$.
That is, our postexpectation is $\PA:=\opon{\proj{\ket0}}{\xx_1}$ and our left program is $\init\xx{\ket+}$.
We ask for a suitable preexpectation $\PB$ in $\rhlany\PB{\init\xx{\ket+}}\SKIP\PA$.
The \rulerefx{any}{Init1} rule gives us $\PB = \id_{\xx_1}\otimes
\pb\paren{\bra+\otimes\id_{\lnot \xx_1}}\PA\pb\paren{{\ket+}\otimes\id_{\lnot \xx_1}}$.
(Here, \symbolindexmarkonly\negXX$\symbolindexmarkhighlight{\negXX\XX_1}:=\XXall_1\XXall_2\setminus\XX_1$.)
By definition of $\PA$, we have that
$\pb\paren{\bra+\otimes\id_{\lnot \xx_1}}\PA\pb\paren{{\ket+}\otimes\id_{\lnot \xx_1}}
= \bra+\proj{\ket0}\ket+ \otimes \id_{\lnot\xx_1} = \frac12 \id_{\lnot\xx_1}$.
Note that this is not an expectation in our sense because it is an operator on all variables \textit{but} $\xx_1$. But by tensoring with $\id_{\xx_1}$, we get the expectation $\PB=\frac12\id$.
Thus the preexpectation is $\frac12\id$ which intuitively means that, no matter what the initial state,
the probability of measuring $\ket0$ will be $\frac12$, as we would expect.

\medskip\noindent\textbf{If:} The rule \ruleref{any}{If1} allows us to
prove a judgment about an if-statement from judgments about the then-
and the else-branch. If the preexpectations from the then- and
else-branch are $\PA_T$
and $\PA_F$,
then the preexpectations for the if-statement is
${\restrictast{\true}(\PA_T)+\restrictast{\false}(\PA_F)}$.
(Here \symbolindexmarkonly\restrictast$\symbolindexmarkhighlight{\restrictast t(A)}:= \adj{\paren{\opon{\Mtruefalse Mt}{\XX_1}}}A
    \paren{\opon{\Mtruefalse Mt}{\XX_1}}$.)
This is natural since $\restrictast{\true}(\PA_T)$
is $\PA_T$
restricted to the case where the conditional holds, and
$\restrictast{\false}(\PA_F)$
is $\PA_F$ restricted to the case where the conditional does not hold.

A toy example: We want to show
$\pb\rhlany{\id}{\ifte M\xx{\apply X\xx}\SKIP}\SKIP\PB$ with
$\PB:={\opon{\proj{\ket0}}{\xx_1}}$.
Here $M$ is a computational basis measurement ($M_{\true}=\proj{\ket1}$, $M_{\false}=\proj{\ket0}$),
and $X$ is the pauli-X matrix (quantum bit flip).
That is, with probability $1$ (preexpectation is $\id$), if we measure $\xx$
in the computational basis,
and, in case of outcome $1$ flip it, we get $\ket0$ (postexpectation $\PB$).
We derive easily (using rules \rulerefx{any}{Apply1} and \rulerefx{any}{Skip}) that
$\pb\rhlany{\opon{\proj{\ket1}}{\xx_1}}{\apply X\xx}\SKIP\PB$
and 
$\pb\rhlany{\opon{\proj{\ket0}}{\xx_1}}\SKIP\SKIP\PB$.
From the \rulerefx{any}{If1} rule, we then get 
$\pb\rhlany{\PA}{\ifte M\xx{\apply X\xx}\SKIP}\SKIP\PB$ with
$\PA={\restrictast{\true}(\opon{\proj{\ket1}}{\xx_1})+\restrictast{\false}(\opon{\proj{\ket0}}{\xx_1})}$.
Thus $\PA={\opon{\proj{\ket1}}{\xx_1}+\opon{\proj{\ket0}}{\xx_1}}=\id$, as desired.

\medskip\noindent\textbf{While:} The \rulerefx{any}{While1} rule is similar
to the \rulerefx{any}{If1} rule. (The preexpectation in the conclusion has
the same form.) The main difference is that we need to guess the
invariant $\PA$ because the postexpectation in the premise contains $\PA$.
The rule also requires us to prove first that the loop is terminating (except
in the case of partial correctness).
Since this is a statement about a single program
  (non-relational), it can be shown using existing approaches (e.g.,
  \cite{LY18termination}) and is outside
  the scope of this paper.

\subsection{Two-sided rules}

The one-sided rules discussed in the previous section allow us to
analyze two programs one statement at a time. However, they are not
sufficient if want to analyze the relationship of two programs that go
in lockstep. (E.g., two while loops that always take the same decision
whether to terminate.) For handling if- and while-statements that are
in sync, our logic provides the two-sided rules \rulerefx{any}{JointIf9},
\rulerefx{any}{JointIf}, and \rulerefx{any}{JointWhile}. Notice that there are
not two-sided analogues to \rulerefx{any}{Apply1} and
\rulerefx{any}{Init1}. This is because the resulting rule would be no
different from using the one-sided rule twice. However, when random
choices happen, two-sided rules are useful. In our case, this happens
in if- and while-statements because the measurement of the
loop-condition introduces randomness.

\medskip\noindent\textbf{If:} The \rulerefx{any}{JointIf9} rule allows us
to compute the preexpectation of two if-statements. It is analogous to
the \rulerefx{any}{If1} rule, except that the resulting preexpectation is
of the form
${\textstyle\sum_{t,u\in\{\true,\false\}}\restrictast{t,u}(\PA_{t,u})}$.
(Here \symbolindexmarkonly\restrictast$\symbolindexmarkhighlight{\restrictast{t,u}(\PA)}
  := \adj{\pb\paren{{\oponp{\Mtruefalse Mt}{\XX_1}}\otimes{\oponp{\Mtruefalse Nt}{\YY_2}}}}
  \PA
  \pb\paren{{\oponp{\Mtruefalse Mt}{\XX_1}}\otimes{\oponp{\Mtruefalse Nt}{\YY_2}}}$.)
That is, the preexpectations are restricted to all four combinations
of true/false for the two if-conditions. And, consequently, we have a
premise for each of those four cases. (The rule is called
\rulerefx{any}{JointIf9} because, in the general case, there are nine
cases, due to explicit treatment of non-terminating cases.) For
convenience, we additionally state \ruleref{any}{JointIf} which considers
only the cases where the two if-statements are in sync (true/true or
false/false).

\medskip\noindent\textbf{While:} Similar to \rulerefx{any}{JointIf}, the
\rulerefx{any}{JointWhile} rule allows us to reason about while loops that
are in sync. Like with \rulerefx{any}{While1}, in contrast to \rulerefx{any}{JointIf},
we need to guess the invariant $\PA$. For an example, see \autoref{sec:example}.
One difference with the \rulerefx{any}{While1} rule is that \rulerefx{any}{While1} requires us
to prove termination in the semipartial and total case (not in the partical case),
while \rulerefx{any}{JointWhile} requires us to prove termination only in the total case.
(Intuitively, this is because in the semipartial case, termination is not required, it is
only required that both programs terminate with the same probability.)


\section{Soundness of the rules}
\label{sec:soundness}

In the following, we state our rules in their generic form (from which the rules in Figure~\ref{fig:rules} are easily derived) and prove their soundness.

\newcommand\rulebox[1]{
  \begin{mdframed}
    #1
  \end{mdframed}
}

\subsection{Basic rules}
\rulebox{
  \RULE{gen}{Skip}{}{\rhlgen\PA\SKIP\SKIP\PA}
}

\begin{proof}
  We use \autoref{lemma:qrhl.pure.new} to prove this rule. For any unit quantum memories $\psi_1$, $\psi_2$ over $\XXall_1$, $\XXall_2$ as input, the output memories are $\denotbot{\SKIP}(\proj{\psi_i})=\proj{\psi_i}$, $i=1,2$. So $\psi_1\otimes\psi_2$ is a separable coupling of the output states with the expected value $\tr A(\proj{\psi_1\otimes\psi_2})$ of the postexpectation, as the same as that of the preexpectation. Moreover, note that $\denotbot{\SKIP}(\proj{\ket{\bot_i}})=\proj{\ket{\bot_i}}$. Then for input $\bot$-bimemories $\ket{\psi_1}\otimes\ket{\bot_2}$, $\ket{\bot_1}\otimes\ket{\psi_2}$, and $\ket{\bot_1}\otimes\ket{\bot_2}$, the output $\bot$-bimomories keep unchanged. Since the preexpectation and the postexpectation are both $\PA$, the expected values of them are also the same. Therefore, the four conditions in \autoref{lemma:qrhl.pure.new} are satisfied, which implies \ruleref{gen}{Skip}.\end{proof}

\rulebox{
  \RULE{gen}{ExFalso}{}{\rhlgen 0\bc\bd\PB}
}
\begin{proof} For any separable mixed $\bot$-bimemory $\rho$ as input, we simply choose $\rho'=\denotbot{\bc}(\partr{2}\rho)\otimes\denotbot{\bd}(\partr{1}\rho)$ as the separable coupling of the output $\bot$-memories, and then $\tr 0\rho=0\le\tr\PB\rho'$. From \autoref{def:eqrhl.generic}, the rule holds.\end{proof}

\rulebox{
  \RULE{gen}{Sym}{\rhlgen \PA\bc\bd\PB}{\rhlgen{\SWAPlrbot^\ast\cdot \PA\cdot \SWAPlrbot}\bd\bc{\SWAPlrbot^\ast\cdot \PB\cdot \SWAPlrbot}}
}

Here, \symbolindexmarkonly\SWAPlrbot$\symbolindexmarkhighlight\SWAPlrbot:\elltwovb{\XXall_2}\otimes\elltwovb{\XXall_1}\mapsto\elltwovb{\XXall_1}\otimes\elltwovb{\XXall_2}$ is a linear operator defined by $\SWAPlrbot(\psi\otimes\phi)=\phi\otimes\psi$ for $\phi\in\elltwovb{\XXall_1}$, $\psi\in\elltwovb{\XXall_2}$.

\begin{proof}
    For any separable mixed $\bot$-bimemory $\rho$ on $\elltwovb{\XXall_2}\otimes\elltwovb{\XXall_1}$, let $\rho_2^\prime=\denotbot{\bd}(\partr{2}\rho)$ and $\rho_1^\prime=\denotbot{\bc}(\partr{1}\rho)$. From \autoref{def:eqrhl.generic}, it suffices to find a separable coupling $\rho'$ of $\rho_2^\prime$ and $\rho_1^\prime$ such that
  \[
    \tr\pb\paren{\SWAPlrbot^\ast\cdot\PA\cdot\SWAPlrbot}\rho
    \leq
    \tr\pb\paren{\SWAPlrbot^\ast\cdot\PB\cdot\SWAPlrbot}\rho'.
  \]

Let $\rho_{12}:=\SWAPlrbot\cdot\rho\cdot\SWAPlrbot^\ast$, then it is easy to verify that $\partr{1}\rho_{12}=\partr{2}\rho$ and $\partr{2}\rho_{12}=\partr{1}\rho$. Furthermore, $\rhlgen \PA\bc\bd\PB$ implies that for input $\bot$-bimemory $\rho_{12}$, there is a separable coupling $\rho_{12}^\prime$ of $\rho_1^\prime$ and $\rho_2^\prime$ such that $\tr\PA\rho_{12}\le\tr\PB\rho_{12}^\prime$. Now let $\rho':=\SWAPlrbot^\ast\cdot\rho_{12}^\prime\cdot\SWAPlrbot$, then the result immediately follows from \begin{equation*}\tr\pb\paren{\SWAPlrbot^\ast\cdot\PA\cdot\SWAPlrbot}\rho=\tr\PA\pb\rho_{12}\ \ \text{and}\ 
\tr\pb\paren{\SWAPlrbot^\ast\cdot\PB\cdot\SWAPlrbot}\rho'=\tr\PB\rho_{12}^\prime.
  \mathQED
\end{equation*}
\end{proof}

\rulebox{
  \RULE{gen}{Seq}{\rhlgen{\PA}{\bc_1}{\bd_1}{\PB}\\ \rhlgen{\PB}{\bc_2}{\bd_2}{\PC}}{\rhlgen{\PA}{\bc_1;\bc_2}{\bd_1;\bd_2}{\PC}}
}

\begin{proof} For any separable mixed $\bot$-bimemory $\rho$ as input, $\rhlgen{\PA}{\bc_1}{\bd_1}{\PB}$ implies that there is a separable coupling $\theta$ of output $\bot$-memories $\theta_1:=\denotbot{\bc_1}(\partr{2}\rho)$, $\theta_2:=\denotbot{\bd_1}(\partr{1}\rho)$, such that $\tr\PA\rho\le\tr\PB\theta$. Furthermore, $\rhlgen{\PB}{\bc_2}{\bd_2}{\PC}$ implies that for separable mixed $\bot$-bimemory $\theta$ as input, there is a separable coupling $\sigma$ of the output $\bot$-memories $\sigma_1:=\denotbot{\bc_2}(\theta_1)=\denotbot{\bc_1;\bc_2}(\partr{2}\rho)$ and $\sigma_2=\denotbot{\bd_2}(\theta_2)=\denotbot{\bd_1;\bd_2}(\partr{1}\rho)$ such that $\tr\PB\theta\le\tr\PC\sigma$. Then we have $\tr\PA\rho\le\tr\PB\theta\le\tr\PC\sigma$. From \autoref{def:eqrhl.generic}, $\rhlgen{\PA}{\bc_1;\bc_2}{\bd_1;\bd_2}{\PC}$ is obtained.\end{proof}

\rulebox{
  \RULE{gen}{Conseq}
  {\PA'\le \PA
    \\
    \rhlgen\PA\bc\bd\PB
    \\
    \PB\le \PB'}
  {\rhlgen{\PA'}\bc\bd{\PB'}}
}

\begin{proof} For any separable mixed $\bot$-bimemory $\rho$ as input, $\rhlgen\PA\bc\bd\PB$ implies that there is a separable coupling $\rho'$ of the output memories such that $\tr\PA\rho\le\tr\PB\rho'$, then it follows immediately from $\PA'\le\PA$ and $\PB'\le\PB$ that
$$\tr\PA'\rho\le\tr\PA\rho\le\tr\PB\rho'\le\tr\PB'\rho'.$$
From \autoref{def:eqrhl.generic}, the rule holds.\end{proof}

\rulebox{
  \RULE{gen}{Scale}{\rhlgen \PA\bc\bd\PB\\ \lambda\in [0,+\infty)}{\rhlgen{\lambda\PA}\bc\bd{\lambda\PB}}
}
\begin{proof}For any separable mixed $\bot$-bimemory $\rho$ as input, $\rhlgen\PA\bc\bd\PB$ implies that there is a separable coupling $\rho'$ of the output memories such that $\tr\PA\rho\le\tr\PB\rho'$, then we have
$\tr(\lambda\PA)\rho=\lambda\tr\PA\rho\le\lambda\tr\PB\rho'=\tr(\lambda\PB)\rho'$. The rule follows from \autoref{def:eqrhl.generic}.\end{proof}

\subsection{One sided rules}
\rulebox{
  \RULE{gen}{Apply1}{}{
    \pb\rhlgen{
      \adj{\paren{\opon U{\XX_1}}}\PA\paren{\opon U{\XX_1}}
      + \restrictast{\bot_1}(\PA)
    }{\apply U\XX}\SKIP\PA
  }
}

Here, $\restrictast{\bot}(\PA):=\Pbot\PA\Pbot$ is the dual of $\restrict{\bot}$ as $\tr\restrictast{\bot}(\PA)\rho=\tr\PA\restrict{\bot}(\rho)$.

\begin{proof} We use \autoref{lemma:qrhl.pure.new} to prove this rule. For any unit quantum memories $\psi_1$, $\psi_2$ over $\XXall_1$, $\XXall_2$ as input, the output memories are  $$\denotbot{\apply U\XX}(\proj{\psi_1})=\paren{\opon U{\XX_1}}\proj{\psi_1}\adj{\paren{\opon U{\XX_1}}}$$ and $\denotbot{\SKIP}(\proj{\psi_2})=\proj{\psi_2}$. So $\paren{\opon U{\XX_1}}\proj{\psi_1\otimes\psi_2}\adj{\paren{\opon U{\XX_1}}}$ is a separable coupling of the output memories with the expected value $$\tr\PA\paren{\opon U{\XX_1}}\proj{\psi_1\otimes\psi_2}\adj{\paren{\opon U{\XX_1}}}$$ of the postexpectation, as the same as the expected value
$$\tr[\adj{\paren{\opon U{\XX_1}}}\PA\paren{\opon U{\XX_1}}+\restrictast{\bot_1}(\PA)]\proj{\psi_1\otimes\psi_2}=\tr\adj{\paren{\opon U{\XX_1}}}A\paren{\opon U{\XX_1}}\proj{\psi_1\otimes\psi_2}$$ of the preexpectation. Similarly, it is easy to see that for $\bot$-bimemories $\ket{\psi_1}\otimes\ket{\bot_2}$, $\ket{\bot_1}\otimes\ket{\psi_2}$, and $\ket{\bot_1}\otimes\ket{\bot_2}$ as input, the expected values of the preexpectation and of the postexpectation are also the same, noting that $\denotbot{\apply U\XX}(\proj{\bot_1})=\proj{\bot_1}$ and $\denotbot{\SKIP}(\proj{\ket{\bot_2}})=\proj{\ket{\bot_2}}$.\end{proof}

\rulebox{
  \RULE{gen}{Init1}{}{
    \pb\rhlgen{
      \id_{\XX_1}\otimes\paren{\adj\psi\otimes\id_{\negXX\XX_1}}\PA\paren{\psi\otimes\id_{\negXX\XX_1}}
      + \restrictast{\bot_1}(\PA)
    }{\init\XX\psi}{\SKIP}{\PA}
  }
}

Here, \symbolindexmarkonly\negXX$\symbolindexmarkhighlight{\negXX\XX_1}:=\XXall_1\XXall_2\setminus\XX_1$.
Here we use that $\psi\in\elltwov{\XX_1}$
  can be interpreted as an operator $\psi:\setC\to\elltwov{\XX_1}$,
  hence $\psi\otimes\id_{\negXX\XX_1}$ is an operator $\negXX\XX_1\to\XXall_1\XXall_2$.
Thus the preexpectation is a positive operator on $\XXall_1\XXall_2$ as required.

\begin{proof} We use \autoref{lemma:qrhl.pure.new} to prove this rule. For any unit quantum memories $\phi_1$, $\phi_2$ over $\XXall_1$, $\XXall_2$ as input, the output memories are $$\denotbot{\init\XX\psi}(\proj{\phi_1})=\proj\psi\otimes \partr{\XX_1}\proj{\phi_1}$$ and $\denotbot{\SKIP}(\proj{\phi_2})=\proj{\phi_2}$. So 
$$\proj\psi\otimes(\partr{\XX_1}\proj{\phi_1})\otimes\proj{\phi_2}=\paren{\psi\otimes\id_{\neg\XX_1}}[\partr{\XX_1}\proj{\phi_1\otimes\phi_2}]\paren{\psi^\ast\otimes\id_{\neg\XX_1}}$$
is a separable coupling of the output memories with the expected value
\begin{equation*}\begin{split}
&\tr\PA\paren{\psi\otimes\id_{\neg\XX_1}}[\partr{\XX_1}\proj{\phi_1\otimes\phi_2}]\paren{\psi^\ast\otimes\id_{\neg\XX_1}}\\
&\starrel=\tr\paren{\psi^\ast\otimes\id_{\neg\XX_1}}\PA\paren{\psi\otimes\id_{\neg\XX_1}}\partr{\XX_1}\proj{\phi_1\otimes\phi_2}\\&=\tr[\id_{\XX_1}\otimes\paren{\psi^\ast\otimes\id_{\neg\XX_1}}\PA\paren{\psi\otimes\id_{\neg\XX_1}}]\proj{\phi_1\otimes\phi_2}\end{split}\end{equation*} of the postexpectation, as the same as the expected value
$$\tr[\id_{\XX_1}\otimes\paren{\psi^\ast\otimes\id_{\neg\XX_1}}\PA\paren{\psi\otimes\id_{\neg\XX_1}}+\restrictast{\bot}(\PA)]\proj{\phi_1\otimes\phi_2}$$ of the preexpectation. Here $(*)$ is due to the
circularity of the trace (i.e., $\tr AB=\tr BA$ for a trace clasee operator $A$ and a bounded operator $B$).
Similarly, it is easy to see that for $\bot$-bimemories $\ket{\psi_1}\otimes\ket{\bot_2}$, $\ket{\bot_1}\otimes\ket{\psi_2}$, and $\ket{\bot_1}\otimes\ket{\bot_2}$ as input, the expected values of the preexpectation and of the postexpectation are also the same, noting that $\denotbot{\init\XX\psi}(\proj{\bot_1})=\proj{\bot_1}$ and $\denotbot{\SKIP}(\proj{\ket{\bot_2}})=\proj{\ket{\bot_2}}$.\end{proof}

\rulebox{
  \RULE{gen}{If1}{
    \rhlgen{\PA_T}{\bc_T}{\bd}{\PB}\\
    \rhlgen{\PA_F}{\bc_F}{\bd}{\PB}\\
    \rhlgen{\PA_\bot}{\bc_\bot}\bd\PB\\
  }{
    \rhlgen{
      \restrictast{\true}(\PA_T)+\restrictast{\false}(\PA_F)
      +\restrictast{\bot_1}(\PA_\bot)
    }{\ifte{M}{\XX}{\bc_T}{\bc_F}}{\bd}{\PB}
  }
}

Here \symbolindexmarkonly\restrictast$\symbolindexmarkhighlight{\restrictast t(A)}:= \adj{\paren{\opon{\Mtruefalse Mt}{\XX_1}}}A
\paren{\opon{\Mtruefalse Mt}{\XX_1}}$ is the Heisenberg-Schr\"{o}dinger dual of $\restrict{t}$ for $t=\true,\false$, as $\tr A\restrict{t}(\rho)=\tr\restrictast{t}(A)\rho$. Note that in this rule, $\bc_\bot$ occurs only in the premises. It can thus be chosen arbitrarily. Typical choices would include $\bc_\bot:=\SKIP$ or $\bc_\bot:=\abort$.

\begin{proof} We use \autoref{lemma:qrhl.pure.new} to prove this rule. For any unit quantum memories $\psi_1$, $\psi_2$ over $\XXall_1$, $\XXall_2$ as input, the output memories are
$$\rho_1:=\denotbot{\ifte M\XX{\bc_T}{\bc_F}}(\proj{\psi_1})=\denotbot{\bc_T}(\restrict{\true}(\proj{\psi_1}))+\denotbot{\bc_F}(\restrict{\false}(\proj{\psi_1}))$$
and $\rho_2:=\denotbot{\bd}(\proj{\psi_2})$. Let $p\cdot\alpha_T:=\restrict{\true}(\proj{\psi_1})$, $(1-p)\cdot\alpha_F:=\restrict{\false}(\proj{\psi_1})$, where $p\in[0,1]$ and $\tr\alpha_T=\tr\alpha_F=1$. Then $\rhlgen{\PA_T}{\bc_T}{\bd}{\PB}$ implies that there is a separable coupling $\rho_T$ of $\denotbot{\bc_T}(\alpha_T)$ and $\rho_2$ such that $\tr\PA_T(\alpha_T\otimes\proj{\psi_2})\le\tr\PB\rho_T$, and $\rhlgen{\PA_F}{\bc_F}{\bd}{\PB}$ implies that there is a separable coupling $\rho_F$ of $\denotbot{\bc_F}(\alpha_F)$ and $\rho_2$ such that $\tr\PA_F(\alpha_F\otimes\beta)\le\tr\PB\rho_F$. Then $\rho:=p\cdot\rho_T+(1-p)\cdot\rho_F$ is a separable coupling of $\rho_1$ and $\rho_2$ satisfying that
\begin{equation*}\begin{split}&\tr\pb\paren{\restrictast{\true}(\PA_T)+\restrictast{\false}(\PA_F)+\restrictast{\bot_1}(\PA_\bot)}\proj{\psi_1\otimes\psi_2}\\
&=\tr\PA_T[\restrict{\true}(\proj{\psi_1})\otimes\proj{\psi_2}]+\tr\PA_F[\restrict{\true}(\proj{\psi_1})\otimes\proj{\psi_2}]+0\\
&=p\tr\PA_T(\alpha_T\otimes\beta)+(1-p)\tr\PA_F(\alpha_F\otimes\beta)\le p\tr\PB\rho_T+(1-p)\tr\PB\rho_F=\tr\PB\rho.\end{split}\end{equation*}
Similarly, we have.
\begin{equation*}\begin{split}&\tr\pb\paren{\restrictast{\true}(\PA_T)+\restrictast{\false}(\PA_F)+\restrictast{\bot_1}(\PA_\bot)}\proj{\psi_1\otimes\ket{\bot_2}}\\
&=p\tr\PA_T(\alpha_T\otimes\Pbot)+(1-p)\tr\PA_F(\alpha_F\otimes\Pbot)\\
&\starrel\le p\tr\PB[\denotbot{\bc_T}(\alpha_T)\otimes\Pbot]+(1-p)\tr\PB[\denotbot{\bc_F}(\alpha_F)\otimes\Pbot]=\tr\PB(\rho_1\otimes\Pbot).\end{split}\end{equation*}
Here, $(*)$ is because $\tr\PA_T(\alpha_T\otimes\Pbot)\le\tr\PB[\denotbot{\bc_T}(\alpha_T)\otimes\Pbot]$ and $\PA_F(\alpha_F\otimes\Pbot)\le\tr\PB[\denotbot{\bc_F}(\alpha_F)\otimes\Pbot]$ which directly follow from the premises.
Moreover, 
\begin{equation*}\begin{split}&\tr\pb\paren{\restrictast{\true}(\PA_T)+\restrictast{\false}(\PA_F)+\restrictast{\bot_1}(\PA_\bot)}\proj{\Pbot\otimes\beta}=\tr\PA_\bot(\Pbot\otimes\beta)\\
&\starstarrel\le\tr\PB[\Pbot\otimes\denotbot{\bd}(\beta)],\end{split}\end{equation*}
where, $(**)$ is from the premise $\rhlgen{\PA_\bot}{\bc_\bot}\bd\PB$. Then the proof is completed by choosing $\beta=\psi_2$ and $\ket{\bot_2}$\end{proof}

\rulebox{
  \RULE{gen}{While1}
  {\pb\rhlgen{\PA}{\bc}{\SKIP}{\restrictast{\true}(\PA)+\restrictast{\false}(\PB)}
    \\
    \while M\XX\bc\text{ is terminating}
  }
  {\pb\rhlgen{\restrictast{\true}(\PA)+\restrictast{\false}(\PB)}{\while M\XX\bc}{\SKIP}{\PB}}
}

{Here, the rule \rulerefx{gen}{While1} for quantum while loops is established by using an invariant of the form $\restrict{\true}(\PA)+\restrict{\false}(\PB)$. Such a form of quantum invariant was proposed in quantum Hoare logic \cite{ying12floyd} for reasoning about (non-relational) quantuam while loops.

Note that this rule requires us to establish termination of
  $\while M\XX\bc$
  as a precondition.  Since this is a statement about a single program
  (non-relational), it can be shown using existing approaches (e.g.,
  \cite{LY18termination}) and is outside
  the scope of this paper.
}

\begin{proof}
We use \autoref{lemma:qrhl.pure.new} to prove the conclusion from the premises. Note that the preexpectation enjoys 
\begin{equation}[\restrictast{\true}(\PA)+\restrictast{\false}(\PB)](\ket{\bot_1}\otimes\id)=0,\label{eq.ortho.bot1}\end{equation}
then the expected value of the preexpectation is always $0$ for any input $\bot$-bimemory of form $\proj{\ket\bot\otimes\phi}$ where $\phi\in\elltwovb{\XXall_2}$. So, it suffices to prove that
\begin{equation}\tr\pb\paren{\restrictast{\true}(\PA)+\restrictast{\false}(\PB)}(\alpha\otimes\beta)\le\tr\PB(\rho\otimes\beta),\label{eq:while1-result}\end{equation}
where $\alpha:=\proj{\psi_1}$, $\rho:=\denotbot{\while M\XX\bc}(\alpha)$, and $\beta:=\proj{\psi_2}$ or $\Pbot$, for any unit quantum memories $\psi_1$, $\psi_2$ over $\XXall_1$, $\XXall_2$, respectively. Note that $\rho=\denot{\while M\XX\bc}(\alpha)$ since $\while M\XX\bc$ is terminating. To express $\rho$ in a more explicit form, let $\alpha_0:=\alpha$, and for $n=0,1,\dots$, let $\alpha_{n+1}:=\llbracket\bc\rrbracket\circ\restrict{\true}(\alpha_n)$. Then $\rho=\sum_{n=0}^\infty \restrict{\false}(\alpha_n)$ by definition of the semantics of $\whilekw$.

From the premise $\rhlgen{\PA}{\bc}{\SKIP}{\restrictast{\true}(\PA)+\restrictast{\false}(\PB)}$, we have \begin{equation*}\begin{split}&\tr\PA(\restrict{\true}(\alpha_n)\otimes\beta)\le\tr(\restrictast{\true}(\PA)+\restrictast{\false}(\PB))(\denotbot{\bc}\circ\restrict{\true}(\alpha_n)\otimes\beta)\\
&\eqrefrel{eq.ortho.bot1}=\tr(\restrictast{\true}(\PA)+\restrictast{\false}(\PB))(\denot{\bc}\circ\restrict{\true}(\alpha_n)\otimes\beta)=\tr(\restrictast{\true}(\PA)+\restrictast{\false}(\PB))(\alpha_{n+1}\otimes\beta),\end{split}\end{equation*} noting that $\denotbot{\bc}\circ\restrict{\true}(\alpha_n)\otimes\beta$ is the unique separable coupling of the ouput $\bot$-memories as $\beta$ is the density operator of a pure $\bot$-memory. It further implies that
\begin{equation}\begin{split}\label{eq:while1.stepn}
  &\tr\PA\pb\paren{\restrict{\true}(\alpha_n)\otimes\beta}-\tr\PA\pb\paren{\restrict{\true}(\alpha_{n+1})\otimes\beta}=\tr\PA\pb\paren{\restrict{\true}(\alpha_n)\otimes\beta}-\tr\restrictast{\true}(\PA)(\alpha_{n+1}\otimes\beta)\\
&\le\tr\restrictast{\false}(\PB)(\alpha_{n+1}\otimes\beta)=\tr\PB(\restrict{\false}(\alpha_{n+1})\otimes\beta).
\end{split}\end{equation}
Therefore, 
\begin{equation}\begin{split}
&\tr\pb\paren{\restrictast{\true}(\PA)+\restrictast{\false}(\PB)}(\alpha\otimes\beta)=\tr\PB\pb\paren{\restrict{\false}(\alpha_0)\otimes\beta}+\tr\PA\pb\paren{\restrict{\true}(\alpha_0)\otimes\beta}\\
&=\tr\PB\pb\paren{\restrict{\false}(\alpha_0)\otimes\beta}+\tr\PA\pb\paren{\restrict{\true}(\alpha_n)\otimes\beta}\\
&\ \ \ \ \ \ \ \ \ \ \ \ \ \ \ \ \ \ \ \ \ \ +\sum_{i=0}^{n-1}\left[\tr\PA\pb\paren{\restrict{\true}(\alpha_i)\otimes\beta}-\tr\PA\pb\paren{\restrict{\true}(\alpha_{i+1})\otimes\beta}\right]\\
&\eqrefrel{eq:while1.stepn}\le\tr\PB(\restrict{\false}(\alpha_0)\otimes\beta)+\tr\PA\pb\paren{\restrict{\true}(\alpha_n)\otimes\beta}+\sum_{i=0}^{n-1}\tr\PB(\restrict{\false}(\alpha_{i+1})\otimes\beta)\\
&=\tr\PB(\sum_{i=0}^{n}\restrict{\false}(\alpha_i)\otimes\beta)+\tr\PA(\restrict{\true}(\alpha_n)\otimes\beta)\le\tr\PB(\rho\otimes\beta)+\tr\PA(\restrict{\true}(\alpha_n)\otimes\beta).
  \label{eq:while1-n}
\end{split}\end{equation}
On the other hand, note that 
\begin{multline*}
  \tr\restrict{\true}(\alpha_n)=\tr\alpha_n-\tr\restrict{\false}(\alpha_n)=\tr\denot{\bc}\circ\restrict{\true}(\alpha_{n-1})-\tr\restrict{\false}(\alpha_n)\\
\le\tr\restrict{\true}(\alpha_{n-1})-\tr\restrict{\false}(\alpha_n),
\end{multline*}
by indution on $n$, it is easy to prove that
\begin{equation}\tr\restrict{\true}(\alpha_n)=1-\sum_{i=0}^n\tr\restrict{\false}(\alpha_i)\overset{n\rightarrow\infty}{\longrightarrow} 1-\tr\rho\starrel=0.\label{eq.while1.term.lim}\end{equation}
Here $(*)$ is due to the termination of $\while M\XX\bc$. Noting that $\PA$ is a bounded operator, then \autoref{eq.while1.term.lim} implies that $\lim_{n\rightarrow}\tr\PA(\restrict{\true}(\alpha_n)\otimes\beta)=0$, which further implies \autoref{eq:while1-result} together with \autoref{eq:while1-n}.
\end{proof}

We refer to the symmetric rules of \rulerefx{gen}{Apply1},
  \rulerefx{gen}{Init1}, \rulerefx{gen}{If1}, and \rulerefx{gen}{While1} (obtained by
  applying \rulerefx{gen}{Sym}) as \implicitRULE{gen}{Apply2},
  \implicitRULE{gen}{Init2}, \textsc{If2}, and
  \implicitRULE{gen}{While2}.
  For example:
  \rulebox{
  \RULE{gen}{If2}{\rhlgen{\PA_T}{\bc}{\bd_T}{\PB}\\ \rhlgen{\PA_F}{\bc}{\bd_F}{\PB}\\ \rhlgen{\PA_\bot}{\bc}{\bc_\bot}{\PB}}{\rhlgen{\restrictast{\true}(\PA_T)+\restrictast{\false}(\PA_F)+\restrictast{\bot_2}(\PA_\bot)}{\bc}{\ifte{N}{\YY}{\bd_T}{\bc_F}}{\PB}}
}

\subsection{Two sided rules}

\rulebox{
    \RULE{gen}{JointIf9}
    {
      \rhlgen{\PA_{t,u}}
      {\bc_t}{\bd_u}
      \PB
      \text{
        for $t,u\in\{\true,\false,\bot\}$
      }
    }
    {
      \begin{array}{rl}
      \pb\braces{
        \textstyle
        \sum_{t,u\in\{\true,\false,\bot\}}
        \restrictast{t,u}(\PA_{t,u})
      }
        &{\ifte M\XX{\bc_{\true}}{\bc_{\false}}}
        \crcr
      \simtemplate{gen}
      & {\ifte N\YY{\bd_{\true}}{\bd_{\false}}}
      \quad  \pb\braces{\PB}
      \end{array}
    }
  }

\begin{proof}
  For convenience, we denote by $\restrictast{i,t}$ the action of $\restrictast{t}$ from the $i$th program, for $i=1,2$ and $t=\true,\false,\bot$. Then $\restrictast{t,u}=\restrictast{1,t}\circ\restrictast{2,u}=\restrictast{2,u}\circ\restrictast{1,t}$.
By using \ruleref{gen}{If1}, it follows from the premises for every $u\in\{\true,\false,\bot\}$ that
\begin{equation}\rhlgen{\restrictast{1,\true}(\PA_{\true,u})+\restrictast{1,\false}(\PA_{\false,u})+\restrictast{1,\bot}(\PA_{\bot,u})}{\ifte M\XX{\bc_{\true}}{\bc_{\false}}}{\bd_u}{\PB}\label{equ:if1u}.\end{equation}
Let $\PA_u:=\restrictast{1,\true}(\PA_{\true,u})+\restrictast{1,\false}(\PA_{\false,u})$. Then by using \ruleref{gen}{If2}, it follows from \autoref{equ:if1u} for all $u\in\{\true,\false,\bot\}$ that
$$\rhlgen{\PA}{\ifte M\XX{\bc_{\true}}{\bc_{\false}}}{\ifte N\YY{\bd_{\true}}{\bd_{\false}}}{\PB},$$
where $\PA:=\restrictast{2,\true}(\PA_{\true})+\restrictast{2,\false}(\PA_{\false})+\restrictast{2,\bot}(\PA_\bot)= \sum_{t,u\in\{\true,\false,\bot\}}\restrictast{t,u}(\PA_{t,u})$.\end{proof}

\rulebox{
  \RULE{gen}{JointWhile}{\rhlgen{\PA}{\bc}{\bd}{\restrictast{\true,\true}(\PA)+\restrictast{\false,\false}(\PB)}\\
    \while M\XX\bc\text{  or  }\while N\YY\bd\text{ is terminating}}{\rhlgen{\restrictast{\true,\true}(\PA)+\restrictast{\false,\false}(\PB)}{\while M\XX\bc}{\while N\YY\bd}{\PB}}
}

Note that this rule requires us to establish termination of
  $\while M\XX\bc$
  as a precondition.  Since this is a statement about a single program
  (non-relational), it can be shown using existing approaches (e.g.,
  \cite{LY18termination}) and is outside
  the scope of this paper.

\begin{proof}
We use \autoref{lemma:qrhl.pure.new} to prove the conclusion from the premises. Note that the preexpectation enjoys 
\begin{equation}[\restrictast{\true}(\PA)+\restrictast{\false}(\PB)](\ket{\bot_1}\otimes\id)=[\restrictast{\true}(\PA)+\restrictast{\false}(\PB)](\id\otimes\ket{\bot_2})=0,\label{eq:orthor.bot2}\end{equation}
the expected value of the preexpectation is always $0$ for any input $\bot$-bimemory of form $\proj{\ket\bot\otimes\phi}$ or $\proj{\psi\otimes\ket\bot}$. So it suffices to prove for any normalized quantum memories $\psi_1,\psi_2$ over $\XXall_1,\XXall_2$ that there is a separable coupling $\rho$ of $\rho_1+(1-\tr\rho_1)\Pbot$ and $\rho_2+(1-\tr\rho_2)\Pbot$ such that
\begin{equation}\tr\left(\restrictast{\true,\true}(\PA)+\restrictast{\false,\false}(\PB)\right)\proj{\psi_1\otimes\psi_2}\le\mathrm{tr}B\rho,\label{equ:while2-result}\end{equation}
where $\rho_1:=\denot{\while M\XX\bc}(\proj{\psi_1})$ and $\rho_2:=\denot{\while N\YY\bd}(\proj{\psi_2})$ are the output memories. Let
\[
  \restrict{1,t}(\rho):=
  \paren{\opon{\Mtruefalse Mt}{\XX_1}}\rho\adj{\paren{\opon{\Mtruefalse Mt}{\XX_1}}},\quad
  \restrict{2,t}(\rho):=\paren{\opon{\Mtruefalse Nt}{\YY_2}}\rho\adj{\paren{\opon{\Mtruefalse Nt}{\YY_2}}},
\]
for $t=\true,\false$. Then $\restrict{t,t}=\restrict{1,t}\otimes\restrict{2,t}$.  Let $\alpha_0:=\proj{\psi_1}$, $\beta_0:=\proj{\psi_2}$, $\alpha_{n+1}:=\llbracket\bc\rrbracket\circ\restrict{1,\true}(\alpha_n)$ and $\beta_{n+1}:=\llbracket\bd\rrbracket\circ\restrict{2,\true}(\beta_n)$ for $n=0,1,\dots$ Then it is easy to verify that $\rho_1=\sum_{n=0}^\infty \restrict{1,\false}(\alpha_n)$, $\rho_2=\sum_{n=0}^\infty\restrict{2,\false}(\beta_n)$.

Now we construct by induction on $n$ a sequence of separable mixed bimemories $\eta_0,\eta_1,\dots,\eta_n,\dots$ over $\elltwov{\XXall_1\XXall_2}$ as follows: let $\eta_0:=\alpha_0\otimes\beta_0$ as the basis; suppose $\eta_n$ has been constructed, then from
$\rhlgen{\PA}{\bc}{\bd}{\restrictast{\true,\true}(\PA)+\restrictast{\false,\false}(\PB)}$, we choose $\restrict{\true,\true}(\eta_n)$ as the (unnormalized) input $\bot$-bimemory and construct $\eta_{n+1}^\prime$ as the separable coupling of the output $\bot$-memories, i.e.,
\begin{gather}
  \partr{2}\eta_{n+1}^\prime=\denotbot\bc(\partr{2}\restrict{\true}(\eta_n)),\ \ \partr{1}\eta_{n+1}^\prime=\denotbot\bd(\partr{1}\restrict{\true}(\eta_n)),\ \ \text{and}
  \label{equ:while2-body1}
  \\
  \label{equ:while2-body}
  \tr\PA\restrict{\true,\true}(\eta_n)
  \le
  \tr\pb\paren{\restrictast{\true,\true}(\PA)+\restrictast{\false,\false}(\PB)}\eta_{n+1}^\prime.
\end{gather}
Now we construct $\eta_{n+1}:=\restrict{\not\bot,\not\bot}(\eta_{n+1}^\prime)$. It follows from \autoref{eq:orthor.bot2} that
$$\tr\pb\paren{\restrictast{\true,\true}(\PA)+\restrictast{\false,\false}(\PB)}\eta_{n+1}^\prime=\tr\pb\paren{\restrictast{\true,\true}(\PA)+\restrictast{\false,\false}(\PB)}\eta_{n+1}.$$ Then \autoref{equ:while2-body} actually means that

\begin{equation}\PA\restrict{\true,\true}(\eta_n)-\PA\restrict{\true,\true}(\eta_{n+1})\le\PB\restrict{\false,\false}(\eta_{n+1}).\label{eq:while2.stepn}\end{equation}
Furthermore, we prove by induction on $n$ that $\partr{2}\eta_n\le\alpha_n$ and $\partr{1}\eta_n\le\beta_n$ for $n=0,1,\dots$. The result obviously holds for $\eta_0=\alpha_0\otimes\beta_0$. Suppose the result holds for $n$, then we prove for $n+1$. To this end, we note that
\begin{align}
  \partr{2}\restrict{\true,\true}(\eta_n)
  &= \partr{2}\restrict{1,\true}\circ(\restrict{2,\true}+\restrict{2,\false})(\eta_n)
    - \partr2 \paren{\restrict{1,\true}\circ\restrict{2,\false}}(\eta_n)
    \notag
  \\ 
  &\starrel=\restrict{1,\true}(\partr{2}\eta_n)-\partr{2}\restrict{\true,\false}(\eta_n)\le \restrict{1,\true}(\partr{2}\eta_n).
  \label{eq:tr.truetrue.eta}
\end{align}
Here $(*)$ follows since $\adj{\Mtrue N}\Mtrue N+\adj{\Mfalse N}\Mfalse N=\id$
  by definition of binary measurements and hence $\tr\circ(\restrict{2,\true}+\restrict{2,\false})=\id$.  
We have
\begin{multline*}
  \partr{2}\eta_{n+1}=\partr{2}\restrict{\not\bot,\not\bot}(\eta_{n+1}^\prime)
  \eqrefrel{equ:while2-body1}=
  \denotbot\bc\pb\paren{\partr{2}\restrict{\not\bot,\not\bot}\circ\restrict{\true,\true}(\eta_n)}=\denot{\bc}\circ\partr{2}\restrict{\true,\true}(\eta_n)\\
  \eqrefrel{eq:tr.truetrue.eta}\le
  \llbracket\bc\rrbracket\circ\restrict{1,\true}(\partr{2}\eta_n)\starrel\le\llbracket\bc\rrbracket\circ\restrict{1,\true}(\alpha_n)=\alpha_{n+1}.
\end{multline*}
Here$(*)$ is from the induction hypothesis $\partr{2}\eta_n\le\alpha_n$. Hence $\tr_2\eta_n\leq\alpha_n$ for all $n$.
Similarly, $\partr{1}\eta_{n}\le\beta_{n}$ for all $n$.

Let $\eta:=\sum_{n=0}^\infty \restrict{\false,\false}(\eta_n)$ Then it is a separable state such that
\begin{equation*}
  \partr{2}\eta
  =\sum_{n=0}^\infty\partr{2}\restrict{\false,\false}(\eta_n)
  \starrel\le\sum_{n=0}^\infty \restrict{1,\false}(\partr{2}\eta_n)
  \le\sum_{n=0}^\infty \restrict{1,\false}(\alpha_n)
  =\rho_1,
\end{equation*}
and similarly, $\partr{1}\eta\le\rho_2$. Here $(*)$ is proven analogously to \eqref{eq:tr.truetrue.eta}. Now construct
$$\rho:=\eta+[\rho_1-\partr{2}\eta+(1-\tr\rho_1)\Pbot]\otimes[\rho_2-\partr{1}\eta+(1-\tr\rho_2)\Pbot]/(1-\tr\eta).$$ It is easy to verify that $\eta\le\rho$, $\partr{2}\rho=\rho_1+(1-\tr\rho_1)\Pbot$, and $\partr{1}\rho=\rho_2+(1-\tr\rho_2)\Pbot$. Now we have 
\begin{equation}\begin{split}
&\tr\pb\paren{\restrictast{\true}(\PA)+\restrictast{\false}(\PB)}\proj{\psi_1\otimes\psi_2}=\tr\PB\restrict{\false,\false}(\eta_0)+\tr\PA\restrict{\true,\true}(\eta_0)\\
&=\tr\PB\restrict{\false,\false}(\eta_0)+\tr\PA\restrict{\true,\true}(\eta_n)+\sum_{i=0}^{n-1}\left[\tr\PA\restrict{\true,\true}(\eta_i)-\tr\PA\restrict{\true,\true}(\eta_{i+1})\right]\\
&\eqrefrel{eq:while2.stepn}\le\tr\PB\restrict{\false,\false}(\eta_0)+\tr\PA\restrict{\true,\true}(\eta_n)+\sum_{i=0}^{n-1}\tr\PB\restrict{\false,\false}(\eta_{i+1})\\
&=\tr\PB\sum_{i=0}^n\restrict{\false,\false}(\eta_i)+\tr\PA\restrict{\true,\true}(\eta_n)\le\tr\PB\eta+\tr\PA\restrict{\true,\true}(\eta_n)\\
&\le\tr\PB\rho+\tr\PA\restrict{\true,\true}(\eta_n),
  \label{eq:while2-n}
\end{split}\end{equation}
for any $n\ge 1$. From the premise we assume without loss of generality that $\while M\XX\bc$ is terminating. Then similar to \autoref{eq.while1.term.lim},
$$\tr\restrict{\true,\true}(\eta_n)\eqrefrel{eq:tr.truetrue.eta}\le\restrict{1,\true}(\partr{2}\eta_n)\le\restrict{1,\true}(\alpha_n)\overset{n\rightarrow\infty}{\longrightarrow}1-\tr\rho_1=0.$$
\autoref{equ:while2-result} immediately follows from \autoref{eq:while2-n}.\end{proof}

\subsection{While rules for partial correctness}\label{sec:while}

\rulebox{
  \RULE{par}{While1Par}
  {\pb\rhlpar{\PA}{\bc}{\SKIP}{\restrictast{\true}(\PA)+\restrictast{\false}(\PB)}
    \\
    \PA\le\id
  }
  {\pb\rhlpar{\restrictast{\true}(\PA)+\restrictast{\false}(\PB)}{\while M\XX\bc}{\SKIP}{\PB}}
}

\begin{proof} Consider any normalized quantum memories $\psi,\phi$ as input of the programs,
  and let $\alpha:=\proj\psi$, $\beta:=\proj\phi$ be their density operators and $\sigma:=\llbracket\while M\XX\bc\rrbracket(\alpha)$, then by \autoref{lemma:qrhl.pure}, we choose $\sigma\otimes\beta$ as the subcoupling of output memories, and it suffices to show
\begin{equation}\tr\pb\paren{\restrictast{\true}(\PA)+\restrictast{\false}(\PB)}(\alpha\otimes\beta)\le\tr\PB(\sigma\otimes\beta)+1-\tr\sigma.\label{equ:while1par-result}\end{equation}
To express $\sigma$ in a more explicit form, let $\alpha_0:=\alpha$, and for $n=0,1,\dots$, let $\alpha_{n+1}:=\llbracket\bc\rrbracket\circ\restrict{\true}(\alpha_n)$, then $\sigma=\sum_{n=0}^\infty \restrict{\false}(\alpha_n)$ by definition of the semantics of $\whilekw$.

The precondition $\rhlpar{\PA}{\bc}{\SKIP}{\restrictast{\true}(\PA)+\restrictast{\false}(\PB)}$ implies that for input bimemory $\restrict{\true}(\alpha_n)\otimes\beta$, there is a subcoupling $\rho_n$ of the output memories $\denot{\bc}\circ\restrict{\true}(\alpha_n)=\alpha_{n+1}$ and $[\tr\restrict{\true}(\alpha_n)]\beta$ such that $\tr\restrict{\true}(\alpha_n)+\tr\rho_n\ge\tr\alpha_{n+1}+\tr[\restrict{\true}(\alpha_n)\otimes\beta]$, and
\begin{equation}\tr\PA(\restrict{\true}(\alpha_n)\otimes\beta)\le\tr\left[\restrictast{\true}(\PA)+\restrictast{\false}(\PB)\right]\rho_n+\tr\restrict{\true}(\alpha_n)-\tr\rho_n.\label{eq:alpha.n.prime}\end{equation}
Note that $\beta$ is the density operator of a pure quantum memory, it is easy to verify that $\rho_n=\alpha_{n+1}\otimes\beta$. Then (\ref{eq:alpha.n.prime}) can be rewritten as
$$\tr(\id-\PB)(\restrict{\false}(\alpha_{n+1})\otimes\beta)\le\tr(\id-\PA)(\restrict{\true}(\alpha_n-\alpha_{n+1})\otimes\beta),$$ then
\begin{equation*}\begin{split}&\tr(\id-\PB)\left(\restrict{\false}\left(\sum_{i=1}^{n+1}\alpha_n\right)\otimes\beta\right)\le\tr(\id-\PA)\left[\restrict{\true}\left(\sum_{i=0}^n(\alpha_i-\alpha_{i+1})\right)\otimes\beta\right]\\
&=\tr(\id-\PA)(\restrict{\true}(\alpha_0-\alpha_{n+1})\otimes\beta)\starrel\le\tr(\id-\PA)(\restrict{\true}(\alpha_0)\otimes\beta).\end{split}\end{equation*}
Here $(*)$ is from the precondition $\PA\le\id$. Now, let $n\rightarrow\infty$ then we have.
$$\tr(\id-\PB)[\sigma\otimes\beta-\restrict{\false}(\alpha)\otimes\beta]\le\tr(\id-\PA)[\restrict{\true}(\alpha)\otimes\beta]$$
which can be rewritten as (\ref{equ:while1par-result}).\end{proof}

\rulebox{
  \RULE{par}{JointWhilePar}{\rhlpar{\PA}{\bc}{\bd}{\restrictast{\true,\true}(\PA)+\restrictast{\false,\false}(\PB)}\\   \PA\le\id}{\rhlpar{\restrictast{\true,\true}(\PA)+\restrictast{\false,\false}(\PB)}{\while M\XX\bc}{\while N\YY\bd}{\PB}}
}

\begin{proof}
Consider any separable mixed bimemory $\rho$ as coupling of the inputs of the programs, and let $\xi_1:=\llbracket\while M\XX\bc\rrbracket(\partr{2}\rho)$ and $\xi_2:=\llbracket\while N\YY\bd\rrbracket(\partr{1}\rho)$ be the output states. Then by \autoref{lemma:qrhl.equiv.par}, we only need to find a separable subcoupling $\sigma$ of $\xi_1$ and $\xi_2$ such that $\tr\rho+\tr\sigma\ge\tr\xi_1+\tr\xi_2$ and
\begin{equation}\tr\left(\restrictast{\true,\true}(\PA)+\restrictast{\false,\false}(\PB)\right)\rho\le\mathrm{tr}B\sigma+\tr\rho-\tr\sigma.\label{equ:while2par-result}\end{equation}
We decompose $\xi_1$ and $\xi_2$ according to the semantic functions of the while loops. Let
\[
  \restrict{1,t}(\rho):=
  \paren{\opon{\Mtruefalse Mt}{\XX_1}}\rho\adj{\paren{\opon{\Mtruefalse Mt}{\XX_1}}},\quad
  \restrict{2,t}(\rho):=\paren{\opon{\Mtruefalse Nt}{\YY_2}}\rho\adj{\paren{\opon{\Mtruefalse Nt}{\YY_2}}},
\]
for $t=\true,\false$. Then $\restrict{t,t}=\restrict{1,t}\otimes\restrict{2,t}$. It is easy to verify that 
$$\xi_1=\sum_{n=0}^\infty \restrict{1,\false}\circ(\denot{\bc}\circ\restrict{1,\true})^n(\partr{2}\rho),\quad \xi_2=\sum_{n=0}^\infty\restrict{2,\false}\circ(\denot{\bd}\circ\restrict{2,\true})^n(\partr{1}\rho).$$ We write $c_n=\sum_{k=0}^n \restrict{1,\false}\circ(\denot{\bc}\circ\restrict{1,\true})^k$ and $d_n=\sum_{k=0}^n\restrict{2,\false}\circ(\denot{\bd}\circ\restrict{2,\true})^k$ for the superoperators of finite sum.

In order to prove (\ref{equ:while2par-result}), we construct a separable mixed memory $\sigma_n$ for each number $n=0,1,\cdots$, such that
\begin{equation}\begin{split}&\partr{2}\sigma_n\le c_n(\partr{2}\rho), \partr{1}\sigma_n\le d_n(\partr{1}\rho), \tr\rho+\tr\sigma_n\ge\tr c_n(\partr{2}\rho)+\tr d_n(\partr{1}\rho)\\&\tr\left(\restrictast{\true,\true}(\PA)+\restrictast{\false,\false}(\PB)\right)\rho\le\tr B\sigma_n+\tr\rho-\tr\sigma_n.\label{equ:while2par-result-finite}\end{split}\end{equation}
Then $\sigma$ can be constructed as the limit point of the sequence $\sigma_0,\sigma_1,\cdots$.

We construct by induction on $n$. For $n=0$, we can put $\sigma_0=\restrict{\true,\true}(\rho)$, then (\ref{equ:while2par-result-finite}) is easily checked for a product memory $\rho$, and further verified for general separable $\rho$  by additivity. Now assuming that the result is true for $n$, we prove for the case of $n+1$. From $\rhlpar{\PA}{\bc}{\bd}{\restrictast{\true,\true}(\PA)+\restrictast{\false,\false}(\PB)}$, we choose $\restrict{\true,\true}(\rho)$ as the (unnormalized) input coupling and construct $\rho'$ as the corresponding subcoupling of the output states, i.e.,
\begin{gather}
  \partr{2}\rho'\le\llbracket\bc\rrbracket\left(\partr{2}\restrict{\true,\true}(\rho)\right),\ \ \partr{1}\rho'\le\llbracket\bd\rrbracket\left(\partr{1}\restrict{\true,\true}(\rho)\right), \label{equ:while2par-body1}\\
\tr\restrict{\true,\true}(\rho)+\tr\rho'\ge\tr\denot{\bc}\left(\partr{2}\restrict{\true,\true}(\rho)\right)+\tr\denot{\bd}\left(\partr{1}\restrict{\true,\true}(\rho)\right)\label{equ:while2par-body2}
  \\
  \tr\PA\restrict{\true,\true}(\rho)
  \le
  \tr\pb\paren{\restrictast{\true,\true}(\PA)+\restrictast{\false,\false}(\PB)}\rho'+\tr\restrict{\true,\true}(\rho)-\tr\rho'.\label{equ:while2par-body}
\end{gather}
Furthermore, we employ the induction hypothesis on the case of $n$ for input couping $\rho'$. Let the subcoupling of output states be $\sigma'$, then
\begin{equation}\begin{split}&\partr{2}\sigma'\le c_n(\partr{2}\rho'), \partr{1}\sigma'\le d_n(\partr{1}\rho'), \tr\rho'+\tr\sigma'\ge\tr c_n(\partr{2}\rho')+\tr d_n(\partr{1}\rho')\\&\tr\left(\restrictast{\true,\true}(\PA)+\restrictast{\false,\false}(\PB)\right)\rho'\le\tr B\sigma+\tr\rho'-\tr\sigma'.\label{equ:while2par-result-n}\end{split}\end{equation}
The result (\ref{equ:while2par-result}) for the case of $n+1$ can be achieved by putting 
\begin{equation}\sigma_{n+1}=\sigma'+(c_{n+1}\otimes\restrict{2,\true})(\rho)+(\restrict{1,\true}\otimes d_{n+1})(\rho)-\restrict{\true,\true}(\rho),\label{equ:sigmaN2sigmaN+1}\end{equation}
 together with (\ref{equ:while2par-body1})-(\ref{equ:while2par-result-n}). Therefore, the result is true for every $n$.

Now, note that $\tr\sigma_n=\tr(\partr{2}\sigma_n)\le\tr c_n(\rho)\le\tr\rho$ is upper-bounded. Moreover, from (\ref{equ:sigmaN2sigmaN+1}) one can easily prove that $\sigma_n\le \sigma_{n+1}$ for every $n$ by induction.   
Thus, the sequence $\sigma_0,\sigma_1,\cdots$ converges as it is an asending chain of trace class operators with a uniform upper bound on the trace. Let $\sigma=\lim_{n\rightarrow\infty}\sigma_n$, then (\ref{equ:while2par-result}) immediately follows from (\ref{equ:while2par-result-finite})\end{proof}

\rulebox{
  \RULE{semi}{JointWhileSemi}{\rhlsemi{\PA}{\bc}{\bd}{\restrictast{\true,\true}(\PA)+\restrictast{\false,\false}(\PB)}\\   \PA\le\id}{\rhlsemi{\restrictast{\true,\true}(\PA)+\restrictast{\false,\false}(\PB)}{\while M\XX\bc}{\while N\YY\bd}{\PB}}
}

\begin{proof}
From \autoref{lemma:qrhl.pure.new}, the semipartial correctness can be characterized in a similar way of \autoref{lemma:qrhl.equiv.par}: $\rhlsemi{\PA}{\bc}{\bd}{\PB}$ if and only if for any separable mixed bimemory $\rho$ as coupling of the inputs of the programs, there exists a separable subcoupling $\sigma$ of the output states such that $\tr\PA\rho\le\tr\PB\sigma+\min\{\tr\denot{\bc}(\partr{2}\rho),\tr\denot{\bd}(\partr{1}\rho)\}-\tr\sigma$. 

Consider any separable mixed bimemory $\rho$ as coupling of the inputs of the programs, and let $\xi_1:=\llbracket\while M\XX\bc\rrbracket(\partr{2}\rho)$ and $\xi_2:=\llbracket\while N\YY\bd\rrbracket(\partr{1}\rho)$ be the output states. Then we only need to find a separable subcoupling $\sigma$ of $\xi_1$ and $\xi_2$ such that
\begin{equation}\tr\left(\restrictast{\true,\true}(\PA)+\restrictast{\false,\false}(\PB)\right)\rho\le\mathrm{tr}B\sigma+\min\{\tr\xi_1,\tr\xi_2\}-\tr\sigma.\label{equ:while2semi-result}\end{equation}
Note that $\xi_1$ and $\xi_2$ can be decomposed as 
$$\xi_1=\sum_{n=0}^\infty \restrict{1,\false}\circ(\denot{\bc}\circ\restrict{1,\true})^n(\partr{2}\rho),\quad \xi_2=\sum_{n=0}^\infty\restrict{2,\false}\circ(\denot{\bd}\circ\restrict{2,\true})^n(\partr{1}\rho).$$ We write $c_n=\sum_{k=0}^n \restrict{1,\false}\circ(\denot{\bc}\circ\restrict{1,\true})^k$ and $d_n=\sum_{k=0}^n\restrict{2,\false}\circ(\denot{\bd}\circ\restrict{2,\true})^k$ for the superoperators of finite sum.

In order to prove (\ref{equ:while2semi-result}), we construct a separable mixed memory $\sigma_n$ for each number $n=0,1,\cdots$ such that
\begin{equation}\begin{split}&\partr{2}\sigma_n\le c_n(\partr{2}\rho), \partr{1}\sigma_n\le d_n(\partr{1}\rho),\\
&\tr\left(\restrictast{\true,\true}(\PA)+\restrictast{\false,\false}(\PB)\right)\rho\le\tr B\sigma_n+\min\{\tr c_n(\partr{2}\rho),\tr d_n(\partr{1}\rho)\}-\tr\sigma_n.\label{equ:while2semi-result-finite}\end{split}\end{equation}
Then $\sigma$ can be constructed as the limit point of the sequence $\sigma_0,\sigma_1,\cdots$.

We construct by induction on $n$. For $n=0$, we can put $\sigma_0=\restrict{\true,\true}(\rho)$, then (\ref{equ:while2semi-result-finite}) is easily checked for a product memory $\rho$, and further verified for general separable $\rho$  by additivity. Now assuming that the result is true for $n$, we prove for the case of $n+1$. From $\rhlsemi{\PA}{\bc}{\bd}{\restrictast{\true,\true}(\PA)+\restrictast{\false,\false}(\PB)}$, we choose $\restrict{\true,\true}(\rho)$ as the (unnormalized) input coupling and construct $\rho'$ as the corresponding subcoupling of the output states, i.e.,
\begin{align}
  \partr{2}\rho'&\le\llbracket\bc\rrbracket\left(\partr{2}\restrict{\true,\true}(\rho)\right),\ \ \partr{1}\rho'\le\llbracket\bd\rrbracket\left(\partr{1}\restrict{\true,\true}(\rho)\right), \label{equ:while2semi-body1}\\
  \tr\PA\restrict{\true,\true}(\rho)
  &\le
  \tr\pb\paren{\restrictast{\true,\true}(\PA)+\restrictast{\false,\false}(\PB)}\rho'
  \notag\\&\qquad
  +\min\{\tr\llbracket\bc\rrbracket\left(\partr{2}\restrict{\true,\true}(\rho)\right),\tr\llbracket\bd\rrbracket\left(\partr{1}\restrict{\true,\true}(\rho)\right)\}-\tr\rho'.
  \label{equ:while2semi-body}
\end{align}
Furthermore, we employ the induction hypothesis on the case of $n$ for input couping $\rho'$. Let the subcoupling of output states be $\sigma'$, then
\begin{equation}\begin{split}&\partr{2}\sigma'\le c_n(\partr{2}\rho'), \partr{1}\sigma'\le d_n(\partr{1}\rho'),\\
&\tr\left(\restrictast{\true,\true}(\PA)+\restrictast{\false,\false}(\PB)\right)\rho'\le\tr B\sigma+\min\{\tr c_n(\partr{2}\rho'),\tr d_n(\partr{1}\rho')\}-\tr\sigma'.\label{equ:while2semi-result-n}\end{split}\end{equation}
The result (\ref{equ:while2semi-result}) for the case of $n+1$ can be achieved by putting 
\begin{equation}\sigma_{n+1}=(\sigma'+c_{n+1}\otimes\restrict){2,\true}(\rho)+(\restrict{1,\true}\otimes d_{n+1})(\rho)-\restrict{\true,\true}(\rho),\label{equ:sigmaN2sigmaN+1-new}\end{equation}
 together with (\ref{equ:while2semi-body1})-(\ref{equ:while2semi-result-n}). Therefore, the result is true for every $n$.

Now, note that $\tr\sigma_n=\tr(\partr{2}\sigma_n)\le\tr c_n(\rho)\le\tr\rho$ is upper-bounded. Moreover, from (\ref{equ:sigmaN2sigmaN+1-new}) one can easily prove that $\sigma_n\le \sigma_{n+1}$ for every $n$ by induction.   
Thus, the sequence $\sigma_0,\sigma_1,\cdots$ converges as it is an asending chain of trace class operators with a uniform upper bound on the trace. Let $\sigma=\lim_{n\rightarrow\infty}\sigma_n$, then (\ref{equ:while2semi-result}) immediately follows from (\ref{equ:while2semi-result-finite})\end{proof}

\section{Example: Quantum Zeno effect}
\label{sec:example}

\paragraph{Motivation.}
In this section, we study (one specific incarnation of) the quantum Zeno effect as an example of application of our logic. The Zeno effect
implies that the following processes have the same effect:
\begin{compactitem}
\item Start with a qubit in state $\ket0$.
  Apply a continuous rotation (with angular velocity $\omega$) to it.
  (Thus, after time $t$, the state will have rotated by angle $\omega t$.)
\item Start with a qubit in state $\ket 0$.
  Continuously observe the state. Namely, at time $t$,
  measure whether the qubit has rotated by angle $\omega t$.
\end{compactitem}
The quantum Zeno effect implies that in both processes, the state
evolves in the same way (and that the measurement in the second situation always gives
answer ``yes'').
Notice that this means that the measurements can be used to rotate the state.

In our formalization, we will consider the discrete version of this
phenomenon: The rotation is split into $n$
rotations by a small angle, and the continuous measurement consists of
$n$
measurements.  In the limit $n\to\infty$,
both processes yield the same state, but if we consider the situation for a
concrete value of $n$, the result of the processes will be slightly different.
(And the difference can be quantified in terms of $n$.)
This makes this example a prime candidate for our logic:
We want to compare two processes (hence we need relational Hoare logic),
but the processes are not exactly equivalent (hence we cannot use qRHL from \cite{qrhl})
but only close to equivalent
(and the ``amount of equivalence'' can be expressed using expectations).

\paragraph{Formalizing the processes.} We now formalize the two
processes as programs in our language. Let $n\geq1$ be an integer.

In the first process, we have a continuous rotation, broken down into $n$ small rotations.
For simplicity, we will rotate by the angle $\pi/2$ within $n$ steps, thus each small rotation
rotates by angle $\frac\pi{2n}$. This is described by the rotation matrix  $R:=
\begin{pmatrix}
  \cos \frac{\pi}{2n} & -\sin\frac{\pi}{2n} \\
  \sin \frac{\pi}{2n} & \phantom-\cos\frac{\pi}{2n}
\end{pmatrix}.
$
Let $\yy$ be a variable of type $\bit$ (i.e., the qubit that is rotated).
In order to apply the rotation $n$ times, we will need a counter $\xx$ for the while loop.
Let $\xx$ be a variable of type $\setZ$. We will have a loop that continues while (informally speaking)
$\xx<n$.
This is formalized by the projector $P_{<n}$ onto states $\ket i$ with $i<n$.
I.e., $P_{<n}:=\sum_{-\infty<i<n}\proj{\ket i}$.
In slight abuse of notation, we also write $P_{<n}$ for the binary measurement
with Kraus operators $\{P_{<n},\id-P_{<n}\}$.
Furthermore, we need to increase the counter.
For this let \symbolindexmark\incr{$\incr$} be the unitary on $\elltwo{\setZ}$ with $\incr\ket{i}\mapsto\ket{i+1}$ 
Then the program that initializes~$\yy$ with~$\ket0$ and then applies the rotation~$R$
$n$ times can be written as:
\begin{equation}
  \bc := \init\xx{\ket0};\
  \init\yy{\ket0};\
  \while{P_{<n}}\xx{
    \paren{\apply\incr\xx;\
      \apply R\yy}
  }
  \label{eq:zeno.progs.c}
\end{equation}

In the second process, instead of applying $R$,
we measure the state in each iteration of the loop. In the first
iteration, we expect the original state $\phi_0:=\ket0$,
and after the $i$-th iteration, we expect the state $\phi_i:=R\phi_{i-1}$ for $i\geq1$.
This can be done using the program $\ifte{\proj{\phi_i}}\yy\SKIP\SKIP$
where we again write in slight abuse of notation ${\proj{\phi_i}}$
for the corresponding binary measurement.
Since the if-statement first measures $\yy$ and then executes one of the $\SKIP$-branches,
this is effectively just a measurement.
We abbreviate this as $\ifskip{\proj{\phi_i}}\yy$.

However, we cannot simply write $\ifskip{\proj{\phi_i}}\yy$ in our loop body,
because $i$ should be the value of $\xx$. So we need to define the projector that projects onto $\phi_i$
when $\xx=\ket i$. This is done by the following projector on $\elltwov{\xx\yy}$:
$P_\phi:=\sum_i\proj{\ket i\otimes\phi_i}$. Then $\ifskip{P_\phi}\yy$ will measure whether
$\yy$ contains $\phi_i$ whenever $\xx$ contains $\ket i$.

Armed with that notation, we can now formulate the second process as a program:
\begin{equation}
    \bd  := \init\xx{\ket0};\
        \init\yy{\ket0};\
        \while{P_{<n}}\xx{
        \paren{\apply\incr\xx;\
        \ifskip{P_\phi}{\xx\yy}}
        }
        \label{eq:zeno.progs.d}
\end{equation}

\paragraph{Equivalence of the programs.}
We claim that the two processes, i.e., the programs $\bc,\bd$
have approximately the same final state in $\yy$.
Having the same state can be expressed using the ``quantum equality''
described in \autoref{sec:var.mem.pred}.
Specifically, the postexpectation $\opon\equal{\yy_1\yy_2}$ corresponds
to $\yy_1$ and $\yy_2$ having the same state.
For example,  $\rhltot\id\bc\bd{\opon\equal{\yy_1\yy_2}}$ implies
that the final state of $\bc$ and $\bd$ is the same
(if we trace out all variables except  $\yy_1,\yy_2$).\footnote{
  This is seen as follows: The judgment implies that the finals states are marginals of a state that is invariant under the projector $\opon\equal{\yy_1\yy_2}$, i.e., a state with support in the space $\YY_1\QUANTEQ\YY_2$.
    That means that this state is invariant under swapping $\YY_1,\YY_2$, and thus the marginals
    corresponding to $\YY_1$ and $\YY_2$ are equal.    
  
}
The fact that the final states are approximately equal can be expressed by
multiplying the preexpectation with a real number close to $1$.
Specifically, in our case we claim that
\begin{equation}
  \label{eq:zeno}
  \rhltot{\varepsilon^n\cdot\id}\bc\bd{\opon\equal{\yy_1\yy_2}}
\end{equation}
Here  $\varepsilon := \paren{\cos\frac{\pi}{2n}}^2$. This indeed means that the final states of $\bc$ and $\bd$ are the same asymptotically
since $\varepsilon^n=\paren{\cos\frac{\pi}{2n}}^{2n}\xrightarrow{n\to\infty} 1$.
\paragraph{Warm up.}
Before we prove \eqref{eq:zeno}, we investigate a simpler case as a warm up.
We investigate the special where $n=3$,
and instead of a while-loop, we simply repeat the loop body three
times.

\begin{equation}
  \label{eq:zeno3.progs}
  \arraycolsep=0pt \def\arraystretch{1.2}
  \begin{array}{lllll}
    \bc' &{}:= \init\yy{\ket0};\
          &\apply R\yy;\
          &\apply R\yy;\
          &\apply R\yy \\
    \bd' &{}:= \init\yy{\ket0};\
          &\ifskip{\proj{\phi_1}}\yy;\
          &\ifskip{\proj{\phi_2}}\yy;\
          &\ifskip{\proj{\phi_3}}\yy
  \end{array}
\end{equation}

We claim:
\begin{equation}
  \label{eq:zeno3}
  \rhltot{\varepsilon^3\cdot\id}{\bc'}{\bd'}{\opon\equal{\yy_1\yy_2}}
\end{equation}

First, we strengthen the postcondition. Let $\PA_3:=\oponp{\proj{\phi_3\otimes\phi_3}}{\yy_1\yy_2}$.
(This postcondition is intuitively what we expect to (approximately)
hold at the end of the execution.
It means that $\yy_1$ and $\yy_2$ are both in state $\phi_3$,
the result by rotating three times using $R$.
Since $\phi_3\otimes\phi_3$ is in the image of the projector $\equal$,
it follows that $\PA_3\leq\oponp\equal{\yy_1\yy_2}$.
By \ruleref{any}{Conseq} it is thus sufficient to show $  \rhltot{\varepsilon^3\cdot\id}{\bc'}{\bd'}{\PA_3}$.
And by \ruleref{any}{Seq}, we can show that  the following sequence of Hoare judgments
for some $\PA_0,\PA_1,\PA_2$:
\begin{equation}\label{eq:zeno.chain}
  \pB\braces{\varepsilon^3\cdot\id}
  \rhltotstack{\init\yy{\ket0}}{\init\yy{\ket0}}
  \pB\braces{\PA_0}
  \rhltotstack{\apply R\yy}{\ifskip{\proj{\phi_1}}\yy}
  \pB\braces{\PA_1}
  \rhltotstack{\apply R\yy}{\ifskip{\proj{\phi_2}}\yy}
  \pB\braces{\PA_2}
  \rhltotstack{\apply R\yy}{\ifskip{\proj{\phi_3}}\yy}
  \pB\braces{\PA_3}
\end{equation}
(These are four judgments, we just use a more compact notation to put them in one line.)
We will derive suitable values $\PA_0,\PA_1,\PA_2$ by applying our rules backwards from the postcondition.

By applying \ruleref{any}{Apply1}, we get
$\pb\rhltot{\PA_3'}{\apply R\yy}\SKIP{\PA_3}$
where $\PA_3':=\oponp{\adj R}{\yy_1}\sandwich{\PA_3}$ and 
where we use $A\symbolindexmark\sandwich\sandwich B$  as an abbreviation for $AB\adj A$.
And by \ruleref{any}{If2} (using \ruleref{any}{Skip} for its premises), we get
\[
  \pB\rhltot{\oponp{\proj{\phi_3}}{\yy_2}\sandwich{\PA_3'}
    + \oponp{1-\proj{\phi_3}}{\yy_2}\sandwich{\PA_3'}
  }\SKIP{\ifskip{\proj{\phi_3}}\yy}{\PA_3'}.
\]
The precondition is lower bounded by $A_2:=
{\oponp{\proj{\phi_3}}{\yy_2}\sandwich{\PA_3'}
}$.
(The second term corresponds to the measurement failing to measure $\phi_3$,
in this case all is lost anyway, so we remove that term.)
Hence (with rules~\rulerefx{any}{Seq} and~\rulerefx{any}{Conseq}), 
$\pb\rhltot{\PA_2}
{\apply R\yy}{\ifskip{\proj{\phi_3}}\yy}
{\PA_3}$ as desired in \eqref{eq:zeno.chain}.

Analogously, we can instantiate
\begin{equation*}
  \PA_1:=
\oponp{\proj{\phi_2}}{\yy_2}\sandwich
\oponp{\adj R}{\yy_1}\sandwich{\PA_2}
\quad\text{and}\quad
\PA_0:=
\oponp{\proj{\phi_1}}{\yy_2}\sandwich
\oponp{\adj R}{\yy_1}\sandwich{\PA_1}
\end{equation*}
in \eqref{eq:zeno.chain}.
We can simplify the expressions for $\PA_0,\PA_1,\PA_2$ some more.
We have
\begin{align*}
  \PA_2 &=
  \oponp{\pb\proj{\phi_3}}{\yy_2}\sandwich
          \oponp{\adj R}{\yy_1}\sandwich{\pb\oponp{\proj{\phi_3\otimes\phi_3}}{\yy_1\yy_2}}
          \\
  &=
  \opon{\pb\proj{\adj R\phi_3\otimes\proj{\phi_3}\phi_3}}{\yy_1\yy_2}
  =
  \opon{\proj{\phi_2\otimes\phi_3}}{\yy_1\yy_2}
\end{align*}
And
\begin{align*}
  \PA_1 &=
  \oponp{\pb\proj{\phi_2}}{\yy_2}\sandwich
          \oponp{\adj R}{\yy_1}\sandwich{\pb\oponp{\proj{\phi_2\otimes\phi_3}}{\yy_1\yy_2}}
          \\
  &=
  \opon{\pb\proj{\adj R\phi_2\otimes\proj{\phi_2}\phi_3}}{\yy_1\yy_2}
  =
    \varepsilon\,\opon{\proj{\phi_1\otimes\phi_2}}{\yy_1\yy_2}.
\end{align*}
(Note the slight difference:
instead of $\proj{\phi_3}\phi_3$ 
have $\proj{\phi_2}\phi_3$ here, which simplifies to 
$\phi_2\cdot\adj{\phi_2}{\phi_3}
=\phi_2\cdot\sqrt{\varepsilon}$.)
Analogously
\begin{equation*}
  \PA_0 = \varepsilon^2\, \opon{\proj{\phi_0\otimes\phi_1}}{\yy_1\yy_2}.
\end{equation*}

It is left to show the first judgment in \eqref{eq:zeno.chain}, namely
$\rhltot{\varepsilon^3\cdot\id}{\init\yy{\ket0}}{\init\yy{\ket0}}
{\PA_0}$.
By rules \rulerefx{any}{Init1} and~\rulerefx{any}{Init2} (starting from the right), we have
\begin{multline}
  \label{eq:ex.init.twice}
  \pB\braces{\varepsilon^3\cdot\id}
  \starstarrel=
  \pB\braces{\id_{\yy_2}\otimes
    \pb\paren{\bra0_{\yy_2}\otimes\id_{\negXX\yy_2}}
    \sandwich \varepsilon^2\oponp{\proj{\phi_1}}{\yy_2}}
  \,{\SKIP}\simtemplate{tot}{\init\yy{\ket0}}
  \\
  \pB\braces{\varepsilon^2\oponp{\proj{\phi_1}}{\yy_2}}
  \starrel=
  \pB\braces{\id_{\yy_1}\otimes
    \pb\paren{\bra0_{\yy_1}\otimes\id_{\negXX\yy_1}}
    \sandwich \PA_0}
  \,{\init\yy{\ket0}}\simtemplate{tot}{\SKIP}\,
  \pB\braces{\PA_0}.
\end{multline}
Here $(*)$ uses that $\phi_0=\ket0$ and thus $\bra 0\proj{\phi_0}\ket0=1$, and $(**)$ uses that $\adj{\phi_1}\phi_0=\sqrt{\varepsilon}$ and thus
$\bra0\proj{\phi_1}\ket0=\varepsilon$.

The first judgment in  \eqref{eq:zeno.chain} then follows by \ruleref{any}{Seq}.

This completes the analysis, we have shown \eqref{eq:zeno3}.

\paragraph{Analysis of the while-programs.}
Given the experiences from the analysis of the special case (the programs from \eqref{eq:zeno3.progs}),
we now can solve the original problem, namely analyzing the programs $\bc,\bd$ from
\eqref{eq:zeno.progs.c},\eqref{eq:zeno.progs.d}.

As before, we can replace the postcondition in \eqref{eq:zeno} by the stronger postcondition
$\PB:=\oponp{\proj{\ket n\otimes\ket n\otimes\phi_n\otimes\phi_n}}{\xx_1\xx_2\yy_1\yy_2}$.
By \ruleref{any}{Conseq}, it is sufficient to show
$  \rhltot{\varepsilon^n\cdot\id}\bc\bd{\PB}$. By \ruleref{any}{Seq}, this follows if we can show
\begin{equation}
  \label{eq:zeno.seq}
  \pB\braces{\varepsilon^n\cdot\id}
  \rhltotstack{\init\xx{\ket0}}{\init\xx{\ket0}}
  \pB\braces{\PD}
  \rhltotstack{\init\yy{\ket0}}{\init\yy{\ket0}}
  \pB\braces{\PC}
  \rhltotstack{\mathrm{while}_{\bc}}{
  \mathrm{while}_{\bd}
  }
  \pB\braces{\PB}
\end{equation}
with
\begin{align*}
  \mathrm{while}_{\bc}
  &:=
  \while{P_{<n}}\xx{
      \paren{\apply\incr\xx;\
        \apply R\yy}
    }
   \\ 
  \mathrm{while}_{\bd}
  &:=
    \while{P_{<n}}\xx{
      \paren{\apply\incr\xx;\
        \ifskip{P_\phi}{\xx\yy}}
    }
\end{align*}
for suitably chosen expectations $\PC$, $\PD$.

To prove the last judgment $\rhltot\PC{\mathrm{while}_{\bc}}{\mathrm{while}_{\bd}}\PB$ in \eqref{eq:zeno.seq},
we use \ruleref{any}{JointWhile}. This rule requires us to come up with a loop invariant $\PA$.
To understand what the right loop invariant is, we draw from our experiences in the special case.
There, we had defined the expectations $\PA_0,\dots,\PA_3$, where $\PA_i$ described the
state of the programs right after the $i$-th application
of $\apply R\yy$ and $\ifskip{\proj{\phi_i}}\yy$.
We had
\begin{equation*}
  \PA_i = \varepsilon^{2-i}\,\opon{\proj{\phi_i\otimes\phi_{i+1}}}{\yy_1\yy_2}
  \text{ for }i=0,1,2
  \qquad\text{and}\qquad
  \PA_3 = \opon{\proj{\phi_3\otimes\phi_3}}{\yy_1\yy_2}
\end{equation*}
One sees easily that this would generalize as
\begin{align*}
  &\PA_i = \varepsilon^{n-i-1}\,\opon{\proj{\phi_i\otimes\phi_{i+1}}}{\yy_1\yy_2}
    \qquad \text{ for }i<n
    \\
  \text{and}\qquad
  &\PA_n = \opon{\proj{\phi_n\otimes\phi_n}}{\yy_1\yy_2}
\end{align*}
for values $n\neq 3$.
Thus we expect that these expectations $\PA_i$ also hold in the programs $\mathrm{while}_{\bc}$,
$\mathrm{while}_{\bd}$ after the $i$-th iteration (or before the $(i+1)$-st iteration).
Additionally, we keep track of the counter $\xx$, which should be
$\ket i$ after the $i$-th iteration (or before the $(i+1)$-st iteration).
This would be expressed by the expectation $\opon{\proj{\ket i\otimes\ket i}}{\xx_1\xx_2}$.
Thus, for the $i$-th iteration, we use the ``conjunction''
\begin{multline*}
  \PA^{\xx}_i := \PA_i \cdot \oponp{\proj{\ket i\otimes\ket i}}{\xx_1\xx_2}
  \fullonly\\
  =
  \begin{cases}
    \varepsilon^{n-i-1}\,
    \opon{\proj{\ket i\otimes\ket i\otimes\phi_i\otimes\phi_{i+1}}}{\xx_1\xx_2\yy_1\yy_2}
    & (i<n) \\
    \phantom{\varepsilon^{n-i-1}\,}
    \opon{\proj{\ket n\otimes\ket n\otimes\phi_n\otimes\phi_{n}}}{\xx_1\xx_2\yy_1\yy_2}
    & (i=n) \\
  \end{cases}
\end{multline*}
(Note that $\cdot$ is not generally a sensible operation on expectations.
But in this case, $\fv(\PA_i)=\yy_1\yy_2$ and $\fv(\opon{\proj{\ket i\otimes\ket i}}{\xx_1\xx_2})=\xx_1\xx_2$,
so the expectations commute and their product is again an expectation.)

The final loop invariant $\PA$ is then the ``disjunction'' of the $\PA^{\xx}_i$ for $i=0,\dots,n-1$,
meaning that in every iteration, one of the $\PA_i$ should hold.
(We do not include $\PA_i^{\xx}$ with $i=n$ here because
when applying the \rulerefx{any}{JointWhile} rule,
we only need the invariant to hold when the loop guard was passed.)
We define
$  \PA := \sum_{i=0}^n \PA^\xx_i$.
(In general, summation is not a sensible operation representation of ``disjunction'',
but in the present case, all summands are orthogonal.)

We have now derived a suitable candidate for the invariant $\PA$ to use in \ruleref{any}{JointWhile}.
We stress that the above argumentation (involving words like ``disjunction'' and ``conjunction'' of
expectations, and claims that an expectation ``holds'' at a certain point) was not a formally
well-defined argument, merely an explanation how we arrived at our specific choice for $\PA$.
From the formal point of view, all we will need in the following are the definitions of $\PA,\PA_i^\xx$.
The rest of the argument above was semi-formal motivation.

We will now show the rightmost judgment in
\eqref{eq:zeno.seq}, namely $\rhltot\PC{\mathrm{while}_{\bc}}{\mathrm{while}_{\bd}}\PB$ (for some suitable $\PC$).
If we apply \ruleref{any}{JointWhile} (with $\PA$ as defined above) to this, we get the premise\footnote{We
    also additionally get the premise that ${\mathrm{while}_{\bc}}$ is terminating.
    This can be shown with techniques from prior work (e.g., \cite{LY18termination})
    and is quite obvious in the present case. Alternatively,
    we could have stated this example with respect to partial correctness instead
    of total correctness. In that case, we do not need to prove termination.    
    }
\begin{equation}
  \pB\braces{\PA}
  \rhltotstack{ \overbrace{\apply\incr\xx;\
    \apply R\yy}^{=:\mathrm{body}_{\bc}}}
  { \underbrace{\apply\incr\xx;\
    \ifskip{P_\phi}{\xx\yy}}_{=:\mathrm{body}_{\bd}}}
  \pB\braces{
    \underbrace{P^\mathrm{both}_{<n}\sandwich\PA
      + \paren{P^\mathrm{none}_{<n}}\sandwich\PB}_{=:\PC'}
  }
  \label{eq:loop.premise}
\end{equation}
with $P^\mathrm{both}_{<n}:=\opon{P_{<n}\otimes P_{<n}}{\xx_1\xx_2}$
and
$P^\mathrm{none}_{<n}:=\opon{(\id-P_{<n})\otimes(\id-P_{<n})}{\xx_1\xx_2}$.
(Here we write $A\sandwich B$ as an abbreviation for $AB\adj A$.)
By applying rules \rulerefx{any}{If2}, \rulerefx{any}{Apply2}, and twice \rulerefx{any}{Apply1}
(with \rulerefx{any}{Seq} in between), we get
\[
  \pb\rhltot{
    \oponp{\incr}{\xx_1}
    \sandwich
    \oponp{R}{\yy_1}
    \sandwich
    \oponp{\incr}{\xx_2}
    \sandwich
    \PB_2
  }
  {\mathrm{body}_{\bc}}  
  {\mathrm{body}_{\bd}}
  {\PC'}
\]
where $\PB_2:=\oponp{P_\phi}{\xx_2\yy_2}\sandwich\PC'+\oponp{\id-P_\phi}{\xx_2\yy_2}\sandwich\PC'$.
Since $\PB_2\geq \oponp{P_\phi}{\xx_1\yy_1}\sandwich\PC'$, by \ruleref{any}{Conseq} we can weaken this to
\begin{multline*}
  \pb\rhltot{\PA'}
  {\mathrm{body}_{\bc}}  
  {\mathrm{body}_{\bd}}
  {\PC'}
  \\
  \text{with}\qquad
  \PA' := {
    \underbrace{\oponp{\incr}{\xx_1}
    \sandwich
    \oponp{R}{\yy_1}
    \sandwich
    \oponp{\incr}{\xx_2}
    \sandwich
    \oponp{P_\phi}{\xx_2\yy_2}}_{=:L}{}
    \sandwich
    \PC'
  }
\end{multline*}
If we can show that $\PA\leq\PA'$ then we have proven
\eqref{eq:loop.premise}. By definition of $\PA_i^{\xx}$, $L$, $R$, $P_\phi$, $\incr$, $P^\mathrm{both}_{<n}$,
we have
\begin{align*}
  L \sandwich P^\mathrm{both}_{<n} \sandwich \PA^{\xx}_i
  &=
    \varepsilon^{n-i-1}\,
    \opon{\proj{\adj\incr\ket{i}\otimes\adj\incr\ket{i}\otimes \adj R\phi_i\otimes \proj{\phi_i}\phi_{i+1}}}{\xx_1\xx_2\yy_1\yy_2}
  \\
  &=
    \varepsilon^{n-i}\,
    \opon{\proj{\ket{i-1}\otimes\ket{i-1}\otimes\phi_{i-1}\otimes \phi_i}}{\xx_1\xx_2\yy_1\yy_2}
    = \PA^{\xx}_{i-1}.
\end{align*}
And $L\sandwich P^\mathrm{both}_{<n}\sandwich\PA^{\xx}_n=0$ since $P^\mathrm{both}_{<n}\sandwich\PA^{\xx}_n=0$.
Thus $L\sandwich P^\mathrm{both}_{<n}\sandwich\PA
=\sum_{i=0}^{n-1} A^{\xx}_{i-1}
\geq
\sum_{i=0}^{n-2} A^{\xx}_{i}
$.
And by definition of $\PB$, $L$, $R$, $P_\phi$, $\incr$, $P^\mathrm{none}_{<n}$, we have
\begin{align*}
  L \sandwich P^\mathrm{none}_{<n} \sandwich \PB
  &=
    \opon{\pb\proj{\adj\incr\ket{n}\otimes\adj\incr\ket{n}\otimes \adj R\phi_n\otimes \proj{\phi_n}\phi_{n+1}}}{\xx_1\xx_2\yy_1\yy_2}
  \\
  &=
    \opon{\pb\proj{\ket{n-1}\otimes\ket{n-1}\otimes\phi_{n-1}\otimes \phi_n}}{\xx_1\xx_2\yy_1\yy_2}
    = \PA^{\xx}_{n-1}.
\end{align*}
Thus $\PA' = L\sandwich  \PC'
\geq
\sum_{i=0}^{n-2} A^{\xx}_{i}
+
A^{\xx}_{n-1}
=
\PA.$
Thus we have proven \eqref{eq:loop.premise}.
By \ruleref{any}{JointWhile},
this implies
$\rhltot{\PC'}{\mathrm{while}_{\bc}}{\mathrm{while}_{\bd}}\PB$
with
$\PC'$ as defined in \eqref{eq:loop.premise}.
With $\PC := A^{\xx}_0 \leq \PC'$,
$\rhltot{\PC}{\mathrm{while}_{\bc}}{\mathrm{while}_{\bd}}\PB$
follows by \ruleref{any}{Conseq}. This is the rightmost judgment in \eqref{eq:zeno.seq}.

Using rules \rulerefx{any}{Init1}, \rulerefx{any}{Init2}, and~\rulerefx{any}{Seq}, we get
$\rhltot\PD{\init\yy{\ket0}}{\init\yy{\ket0}}\PC$ with
$\PD := \varepsilon^n\cdot\oponp{\proj{\ket0\otimes\ket0}}{\xx_1\xx_2}$.
(This is done very similarly to \eqref{eq:ex.init.twice}.)
This shows the middle  judgment in \eqref{eq:zeno.seq}.

Also using rules \rulerefx{any}{Init1}, \rulerefx{any}{Init2}, and~\rulerefx{any}{Seq}, we get
$\rhltot{\varepsilon^n\cdot\id}{\init\xx{\ket0}}{\init\xx{\ket0}}\PD$.
This shows the leftmost judgment in \eqref{eq:zeno.seq}.

Thus we have shown the three judgments in \eqref{eq:zeno.seq}.
By \ruleref{any}{Seq}, it follows that
$\rhltot{\varepsilon^n\cdot\id}\bc\bd\PB$.
Since $\PB\leq \oponp\equal{\yy_1\yy_2}$, by \ruleref{any}{Conseq},
we get 
\eqref{eq:zeno}.


\ifanonymous\else
\paragraph{Acknowledgments.}
 We thank Gilles Barthe, Tore Vincent Carstens, and Justin Hsu for valuable discussions. This work was supported by the Air Force Office of Scientific Research through the project ``Verification of quantum cryptography'' (AOARD Grant FA2386-17-1-4022), by institutional research funding IUT2-1 of the Estonian Ministry of Education and Research, by the Estonian Centre of Exellence in IT (EXCITE) funded by ERDF, by the project ``Research and preparation of an ERC grant application on Certified Quantum Security'' (MOBERC12), by the National Natural Science Foundation of China (Grant No: 61872342), and by the ERC consolidator grant CerQuS (819317).
\fi

\shortonly{  \bibliography{expectation-qrhl} }

\shortonly{    \delaytextnotevenwarn  \end{document}}

\shortonly\clearpage

\appendix

\onecolumn

\section{Comparison with Barthe et al.}
\label{sec:compare.theirs}

In this section, we compare our logic with the one presented in
\cite{theirs}. To make the comparison easier, when we refer to results
from \cite{theirs}, we restate them in our notation (e.g., we write
\symbolindexmark\rhltheirs{$\rhltheirs \PA\bc\bd \PB$}
instead of $\bc\sim\bd\vdash \PA\Rightarrow \PB$,
and we use our syntax for programs).

\begin{itemize}
\item The most important difference is that we define judgments with
  respect to separable couplings (see
  Definitions~\ref{def:eqrhl.first}--\ref{def:eqrhl.generic}). We
  believe that this design choice 
  makes rules in our logic easier to prove and thus makes the logic potentially more powerful.
  See \autoref{footnote:why.sep} for the reasons why using separable couplings simplifies things.

\item We allow infinite dimensional Hilbert spaces (i.e., the types of
  variables can be infinite sets such as integers or reals). In
  contrast, \cite{theirs} only allows finite dimensional Hilbert space
  (e.g., bits or bounded integers are possible but not integers).
\item
   \cite{theirs} does not behave well with
  non-terminating programs. The definition requires the
  existence of a coupling of $\denot\bc(\rho_1),\denot\bd(\rho_2)$
  for any two initial states $\rho_1,\rho_2$.
  This implies that for all (possibly unrelated) input states,
  $\bc,\bd$ need to have exactly the same termination probability.
  In contrast, we can also express Hoare logic with total/partial correctness.

\item In \cite{theirs}, not only judgments $\rhltheirs \PA\bc\bd \PB$ are
  considered, but also an extended form
  $\Sigma\vDash \rhltheirs \PA\bc\bd \PB$.  This means that
  $\rhltheirs \PA\bc\bd \PB$ holds under certain conditions on the
  coupling of the input states (e.g., separability of some of the
  variables or the fact that some variables, when measured, have the
  same outcome probability distribution).  Such conditional judgments
  occur as the result from various reasoning rules in \cite{theirs},
  e.g., \textsc{IF1} (analogue to our \rulerefx{any}{JointIf}), \textsc{LP1}
  (analogue to our \rulerefx{any}{JointWhile}), \textsc{Frame}.  These
  extended judgments make reasoning much harder as we will elaborate.
  They are not needed in our logic.
  
  In \cite{theirs}, the \textsc{SC} rule (analogue to
  \ruleref{any}{Seq} here) allows to combine judgments $\rhltheirs \PA\bc\bd \PB$
  and $\rhltheirs{\PB}{\bc'}{\bd'}{\PC}$
  into $\rhltheirs{\PA'}{\bc;\bc'}{\bd;\bd'}{\PC}$.
  However, it does not apply to judgments of the form
  $\Sigma\vdash\rhltheirs \PA\bc\bd \PB$. For this purpose, \cite{theirs}
  introduces an extended rule \textsc{SC+}:
  \[
    \RULE{their}{SC+}{
      \Sigma\vdash \rhltheirs \PA\bc\bd \PB
      \\
      \Sigma'\vdash\rhltheirs{\PB}{\bc'}{\bd'}{\PC}
      \\
      \Sigma\vDash^{(\bc,\bd)}\Sigma'
    }{
      \rhltheirs{\PA'}{\bc;\bc'}{\bd;\bd'}{\PC}
    }.
  \]
  Here $\Sigma\vDash^{(\bc,\bd)}\Sigma'$  (``coupling-entails'') means that for input states
  $\rho_1,\rho_2$ with a coupling satisfying $\Sigma$, \emph{every}
  coupling of $\denot\bc(\rho_1),\denot\bd(\rho_2)$ satisfies $\Sigma'$.

  The necessity of using this rule leads to two difficulties:
  \begin{itemize}
  \item The logic does not provide support for deriving the premise
    $\Sigma\vDash^{(\bc,\bd)}\Sigma'$. That is, the preservation of
    the side-conditions needs to be shown by direct recourse to the
    semantics of $\bc,\bd$ (i.e., to the functions
    $\denot\bc,\denot\bd$), and to the definition of
    $\vDash^{(\bc,\bd)}$. In the examples from \cite{theirs}, only
    very simple programs $\bc,\bd$ occur here and the premises
    $\Sigma\vDash^{(\bc,\bd)}\Sigma'$ are stated as obvious. However,
    if $\bc,\bd$ themselves are nontrivial programs, then proving
    $\Sigma\vDash^{(\bc,\bd)}\Sigma'$ might be infeasible.\footnote{%
      The fullversion of \cite{theirs} (Appendix A13) gives an
      algorithm for reasoning about measurability conditions (not
      separability conditions, though) by collecting constraints.
      However, that algorithm runs in exponential time in the number
      of qubits in the programs (because it directly computes the
      operators that the programs apply to the quantum state).  In
      particular, if the size of the programs we analyze is parametric
      in some parameter $n$ (e.g., if we are verifying Grover's
      algorithm on $n$ qubits where $n$ is not explicitly given), then
      we cannot apply the algorithm.

      Furthermore, in the presence of nested while-loops, the
      algorithm does not terminate. For any while-loop with
      loop-body $\bc$, the algorithm computes $\denot\bc$ explicitly
      ($P^*$ in line 15 of the algorithm) for which there is no terminating
      algorithm if $\bc$ itself contains a loop.

      Finally, no formal connection between the algorithm and
      judgments $\Sigma\vDash^{(\bc,\bd)}\Sigma'$ is given. It is
      shown that the algorithm outputs side conditions
      $\Sigma$ that makes the programs ``comparable'' (defined in
      Appendix A13), but it is not clear how comparability of
      $\bc,\bd$ implies $\Sigma\vDash^{(\bc,\bd)}\Sigma'$ for suitable
      $\Sigma,\Sigma'$. (In particular, what would be $\Sigma'$?)
      
    }
    In fact,
    whether $\Sigma'$ is satisfied might depend on how the
    probabilities in the executions of $\bc,\bd$ are related which is
    basically the kind of judgments we are trying to derive using the
    relational Hoare logic in the first place.  Thus an additional
    Hoare logic for deriving $\vDash^{(\bc,\bd)}$ would be needed.
  \item $\Sigma\vDash^{(\bc,\bd)}\Sigma'$ is actually a surprisingly
    strong requirement. This is because $\Sigma'$ needs to hold for
    \emph{all} couplings of the final states of $\bc,\bd$, not only
    the ones we actually care about in the further reasoning. For
    example, let $\bc$ be a program that assigns uniformly randomly
    $\ket0$ or $\ket1$ to $\xx$, and $\bd:=\bc$. (Such a program can
    be easily constructed from an initialization and a measurement.)
    Let $\Sigma$ be the side-condition that means that $\xx_1,\xx_2$
    are separable.  (This is a typical side-condition as considered in
    \cite{theirs}.)  Then $\Sigma\vDash^{(\bc,\bd)}\Sigma$ does not
    hold! Even though $\xx_1,\xx_2$ are assigned independent values
    (and thus intuitively should be separable), there \emph{exists} a
    coupling of $\xx_1,\xx_2$ (namely an EPR pair) that is not
    separable. Essentially, this means that \textsc{SC+} probably can
    be very rarely used when a separability condition is involved and
    the first program in the composition initializes the involved
    variables in a way that involves probabilism (e.g., a program with
    measurements).
  \end{itemize}

  Furthermore, \cite{theirs} does not provide rules \textsc{IF+} or
  \textsc{LP+} analogous to \textsc{SC+}.  The rules \textsc{IF1}
  (analogue to our \rulerefx{any}{JointIf}) and \textsc{LP1} (analogous to our
  \rulerefx{any}{JointWhile}) only apply when the premises all hold without
  side-conditions. So if we have an if-statement or a while-loop where
  the judgments concerning the then/else-branches or the loop-body
  hold only under a side-condition, we have no way of proving
  judgments about the if-statement or the while-loop, even if we want
  the latter judgments to hold only under side-conditions.

\item One-to-one comparison of rules:
  \begin{itemize}
  \item \textbf{\rulerefx{any}{Skip} vs.~\textsc{Skip}:} Same rule in both logics.
  \item \textbf{\rulerefx{any}{Init1} vs.~\textsc{Init}, \textsc{Init-L}:} Our \ruleref{any}{Init1} and the
    rule \textsc{Init-L} from \cite{theirs} are the same, except for
    the minor difference that we allow to specify the assigned state
    $\psi$.
    (In our rule the precondition for the program $\init\xx{\ket0}$
    is
    ${\id_{\xx_1}\otimes\paren{\adj\psi\otimes\id_{\negXX\xx_1}}\PA\paren{\psi\otimes\id_{\negXX\xx_1}}}$
    while in \cite{theirs} it is written as
    $\sum_{ij}\pb\paren{\ket{i}\bra0\otimes\ket
      j\bra0}\PA\pb\paren{\ket0\bra i\otimes\ket0\bra j}$ but it is
    easy to verify that both formulas evaluate to the same operator in
    the finite dimensional case.)  We did not explicitly write down a
    direct analogue to the two-sided \textsc{Init} rule from
    \cite{theirs} but we note that the two sided rule follows
    immediately by sequentially applying \rulerefx{any}{Init1} and
    \rulerefx{any}{Init2} (combined with \rulerefx{any}{Seq}).
  \item \textbf{\rulerefx{any}{Apply1} vs. \textsc{UT}, \textsc{UT-L}:} Our
    \ruleref{any}{Apply1} and the rule \textsc{UT-L} from \cite{theirs} are
    the same.\footnote{\textsc{UT-L} is stated only for the
      application of unitaries, \rulerefx{any}{Apply1} is stated more
      generally for the application of isometries. However, since
      \cite{theirs} only considers the finite dimensional case, and in
      finite dimensions unitaries and isometries coincide (for square
      matrices), this is only seemingly a difference.}
  \item \textbf{\rulerefx{any}{Seq} vs.~\textsc{SC}.}  Our \ruleref{any}{Seq}
    and the rule \textsc{SC} from \cite{theirs} are the same.
  \item \textbf{\rulerefx{any}{Conseq} vs.~\textsc{Conseq}.}  Our \ruleref{any}{Conseq}
    and the rule \textsc{Conseq} from \cite{theirs} are the same.
  \item \textbf{\rulerefx{any}{If1} vs.~\textsc{IF-L}:} Our \rulerefx{any}{If1} and the rule \textsc{IF-L} from \cite{theirs} are essentially the same, only with a slight difference on the number of branches.
  \item \textbf{\rulerefx{any}{JointIf}, \rulerefx{any}{JointIf9}
      vs.~\textsc{IF}:} Our two-sided \ruleref{any}{JointIf9} and the rule
    \textsc{IF} from \cite{theirs} are essentially the same, only with a slight difference on the number of branches. Additionally we have the \ruleref{any}{JointIf} as a consequence of \ruleref{any}{JointIf9}, which is also true in the logic of \cite{theirs}.
  \item \textbf{\rulerefx{any}{While1} vs. \textsc{LP-L}:}
    Our \ruleref{any}{While1} and the rule \textsc{LP-L} from \cite{theirs} are the same for terminating quantum loops, but the rule for non-terminating quantun loops in partial and semipartial correctness is unique in our logic.
  \item \textbf{\rulerefx{any}{JointWhile} vs. \textsc{LP}:} The comments
    made for \rulerefx{any}{While1} and \textsc{LP-L} apply for the
    two-sided loop rules as well.
    
    In addition, we stress that our \refrule{any}{JointWhile}
      does not require the termination of the while rules on both sides in the partial and semipartial case. In contrast, \textsc{LP} requires that all while loops are terminating.

      \item \textbf{\textsc{Case}:} The rule \textsc{Case} is only present
    in \cite{theirs}. It states
    \[
      \RULE{their}{Case}{
        \{p_i\}\text{ is a probability distribution}
        \\
        \rhltheirs{A_i}\bc\bd{B}
      }{
        \pB\rhltheirs{\sum p_iA_i}\bc\bd{B}
      }
    \]
    This rule might allow to do a reasoning via case distinction but
    it is not obvious how exactly it is applied. For example, if we
    try to prove $\pb\rhltheirs{\proj{\ket 0}+\proj{\ket 1}}\bc\bd{\PB}$
    by showing $\pb\rhltheirs{\proj{\ket i}}\bc\bd{\PB}$
    for $i=0,1$,
    then an application of \rulerefx{their}{Case} only gives us
    $\pb\rhltheirs{\tfrac12\pb\paren{\proj{\ket 0}+\proj{\ket
          1}}}\bc\bd{\PB}$, not
    $\pb\rhltheirs{{\proj{\ket 0}+\proj{\ket 1}}}\bc\bd{\PB}$.
    (\rulerefx{their}{Case} is not used in the examples in \cite{theirs}.)
  \item \textbf{\textsc{Frame}:}
    \cite{theirs} contains a frame rule
    \[
      \RULE{their}{Frame}{
        \rhltheirs{\PA}{\bc}{\bd}{\PB}
      }{
        \Sigma\vdash \rhltheirs{\PA\otimes\PC}{\bc}{\bd}{\PB\otimes\PC}
      }
    \]
    where $\PC$
    is an expectation operating on variables distinct from those
    $\bc,\bd,\PA,\PB$
    operate on.  Such a rule without the additional separability condition $\Sigma$ is very useful for modular analysis of
    programs (we can analyze $\bc,\bd$
    with respect to the local expectations $\PA,\PB$
    first, and then derive global statements involving $\PC$).
    Such a rule is unfortunately missing in our logic. (It is an open
    problem whether it holds in our logic.)

    However, as we mentioned previously, the \ruleref{their}{Frame} in the current form becomes useless in
    derivations due to the extra separability condition $\Sigma$. (\rulerefx{their}{Frame} is not used in the examples in \cite{theirs}.)
  \item \textbf{\rulerefx{any}{ExFalso}:} This rule is only stated in our
    paper but is trivial to prove for \cite{theirs} when additionally assuming termination of the programs.
  \end{itemize}
\end{itemize}


\fullonly{
\renewcommand\symbolindexentry[4]{
  \noindent\hbox{\hbox to 1.5in{$#2$\hfill}\parbox[t]{4.2in}{#3}\hbox to 1cm{\hfill #4}}\\[2pt]}

  {\sectionToToc
   \printsymbolindex}

  \renewenvironment{theindex}{%
    \section*{Keyword Index}%
    \parindent\z@
    \parskip\z@ \@plus .3\p@\relax
    \let\item\@idxitem
    \begin{multicols}{2}}{\end{multicols}}
  {\sectionToToc \printindex}

    {\sectionToToc  \printbibliography}

}

\end{document}
